\documentclass[a4paper,UKenglish,pdfa,cleveref,autoref,numberwithinsect]{lipics-noline-v2021}

\usepackage{microtype}

\title{Hiding pebbles when the output alphabet is unary}

\author{Gaëtan Douéneau-Tabot}{Université Paris Cité, CNRS, IRIF, F-75013, Paris, France
\and Direction générale de l'armement - Ingénierie des projets, Paris, France
}{doueneau@irif.fr}{}{}

\authorrunning{G. Douéneau-Tabot}

\Copyright{Gaëtan Douéneau-Tabot}



\category{}

\supplement{}

\funding{}

\acknowledgements{The author is grateful to Olivier Carton,
L{\^e} Th{\`a}nh D{\~u}ng Nguy{\^e}n and Pierre Pradic 
for discussing about this work.}

\EventEditors{}
\EventNoEds{0} 
\EventLongTitle{}
\EventShortTitle{ICALP 2022}
\EventAcronym{}
\EventYear{}
\EventDate{}
\EventLocation{}
\EventLogo{}
\SeriesVolume{}
\ArticleNo{}

\usepackage{tikz}
\usepackage[utf8]{inputenc}
\usetikzlibrary{shapes,arrows,positioning,automata,shadows}
\usepackage{amsfonts, amsmath, amssymb,dsfont, mathrsfs}
\usepackage{amsthm}
\usepackage{hyperref}
\usepackage[ruled]{algorithm2e}
\usepackage{stmaryrd}
\usepackage{pifont}
\usepackage{soulutf8}
\usepackage{yhmath,mathtools}
\usepackage{bm}

\newtheorem{sublemma}[theorem]{\bfseries Sublemma}
\newtheorem{depro}[theorem]{\bfseries Proposition-Definition}

\newcommand{\mb}[1]{\mathbb{#1}}
\newcommand{\mc}[1]{\mathcal{#1}}
\newcommand{\mf}[1]{\mathfrak{#1}}

\newcommand{\vide}{\varepsilon}
\newcommand{\defined}{\coloneqq}
\renewcommand{\phi}{\varphi}
\renewcommand{\epsilon}{\textcolor{green}{\varepsilon}}
\newcommand{\Nat}{\mb{N}}
\newcommand{\Rel}{\mb{Z}}

\newcommand{\myBlue}{blue!90}
\newcommand{\tred}[1]{\textcolor{\myBlue}{#1}}
\newcommand{\ce}{\mc{L}}
\newcommand{\de}{\mc{R}}
\newcommand{\Le}{\mb{T}}

\newcommand{\hei}{{3|M|}}

\newcommand{{\lmark}}{\vdash}
\newcommand{{\rmark}}{\dashv}

\newcommand{\con}{{\ddagger}}

\newcommand{\indic}[1]{\chi_{#1}}

\newcommand{\pre}[1]{{\widetilde{#1}}}
\newcommand{\ov}[1]{{\overline{#1}}}
\newcommand{\apar}[1]{{\widehat{#1}}}
\newcommand{\facto}[2]{\operatorname{\normalfont \textsf{Fact}}_{#1}(#2)}
\newcommand{\cro}[1]{ \tred{\bm{\llbracket} #1\bm{\rrbracket}}}
\newcommand{\lefttree}{\langle}
\newcommand{\righttree}{\rangle}
\newcommand{\tree}[1]{\lefttree #1 \righttree}
\newcommand{\marq}[1]{\underline{#1}}
\newcommand{\neutral}{1}
\newcommand{\deutral}{{-}1}

\newcommand{\nod}{\mf{t}}
\newcommand{\Nod}{\mf{T}}
\newcommand{\forest}{\mc{F}}
\newcommand{\gorest}{\mc{G}}

\newcommand{\set}[1]{\{#1\}}
\newcommand{\lin}[2]{\operatorname{\normalfont\textsf{lin}}^{#1}(#2)}
\newcommand{\multi}[1]{\{\!\!\{#1\}\!\!\}}

\newcommand{\Archs}[2]{\operatorname{\normalfont\textsf{Arcs}}_{#1}^{#2}}
\newcommand{\tnorm}[1]{\text{\normalfont {#1}}}
\newcommand{\pc}{\prec}
\newcommand{\arch}[2]{\operatorname{\normalfont\textsf{arc}}^{#1}(#2)}
\newcommand{\shape}[1]{\operatorname{\normalfont\textsf{shape}}(#1)}
\newcommand{\Shapes}{\operatorname{\normalfont\textsf{Shapes}}}
\newcommand{\Indep}{\operatorname{\normalfont\textsf{Ind}}}
\newcommand{\Indeps}[1]{\Indep^k(#1)}
\newcommand{\Deps}[1]{\operatorname{\normalfont\textsf{Dep}}^k(#1)}
\newcommand{\sumd}[1]{\operatorname{\normalfont\textsf{sum-dep}}_{#1}}
\newcommand{\sumi}[1]{\operatorname{\normalfont\textsf{sum-ind}}_{#1}}

\newcommand{\fr}[2]{\operatorname{\normalfont \textsf{Fr}}_{#2}(#1)}
\newcommand{\frm}[2]{\operatorname{\normalfont \textsf{Fr}}^{-1}_{#2}(#1)}
\newcommand{\Ske}[1]{\operatorname{\normalfont\textsf{Skel}}(#1)}
\newcommand{\depth}[2]{\operatorname{\normalfont\textsf{depth}}^{#1}(#2)}
\newcommand{\idem}[1]{{\eta_{#1}}}
\newcommand{\rod}[3]{\operatorname{\normalfont \textsf{prod}}_{#1}^{#2}\left(#3\right)}
\newcommand{\prodd}{\operatorname{\normalfont \textsf{prod}}}
\newcommand{\Nodes}[1]{\operatorname{\normalfont \textsf{Nodes}}(#1)}
\newcommand{\Types}[2]{\operatorname{\normalfont \textsf{Types}}^{#1}(#2)}
\newcommand{\itera}[1]{\operatorname{\normalfont \textsf{Iters}}(#1)}

\newcommand{\lefe}{\operatorname{\normalfont \textsf{left}}}
\newcommand{\rige}{\operatorname{\normalfont \textsf{right}}}
\newcommand{\fact}[1]{\operatorname{\normalfont \textsf{forest}}_{#1}}
\newcommand{\Facts}[3]{\operatorname{\normalfont \textsf{Forests}}_{#1}^{#2}(#3)}

\newcommand*\circled[1]{\tikz[baseline=(char.base)]{
            \node[shape=circle,draw,inner sep=1pt] (char) {#1};}}
\newcommand{\had}{\mathop{\circled{\footnotesize $\operatorname{\normalfont\scalebox{0.6}{\textsf{H}}}$}}}
\newcommand{\cau}{\mathop{\circled{\footnotesize $\operatorname{\normalfont\scalebox{0.6}{\textsf{C}}}$}}}

\newcommand{\oras}{\mf{H}}
\newcommand{\exte}{\mf{h}}

\newcommand{\trans}{\mc{T}}

\newcommand{\mach}{\mc{M}}
\newcommand{\ar}{\mb{A}}
\newcommand{\br}{\mb{B}}

\newcommand{\cou}{\operatorname{\normalfont\textsf{count}}}
\newcommand{\nb}{\operatorname{\normalfont\textsf{nb}}}

\newcommand{\itpow}[1]{{#1}\operatorname{\textsf{-powers}}}

\renewcommand{\implies}{\Rightarrow}

\begin{document}
\maketitle

\begin{abstract} 
Pebble transducers are nested two-way transducers
which can drop marks (named ``pebbles'')  on their input word.
 Blind transducers have been introduced by Nguy{\^{e}}n et al.
 as a subclass of pebble transducers, which can nest two-way transducers
 but cannot drop pebbles on their input.
 
 In this paper, we study the classes of functions
 computed by pebble and blind transducers,
 when the output alphabet is unary.
 Our main result shows how to decide if a function computed
  by a pebble transducer can be computed by a blind
  transducer. We also provide characterizations
  of these classes in terms of Cauchy and Hadamard products,
  in the spirit of rational series. Furthermore, pumping-like characterizations
  of the functions computed by blind transducers are given.
  \end{abstract}



\section{Introduction}

Transducers are finite-state machines obtained by adding outputs
 to finite automata. In this paper, we assume that
these machines are always deterministic and have finite inputs, hence they compute
functions from finite words to finite words. In particular,
a \emph{deterministic two-way transducer} consists of a two-way automaton
which can produce outputs. This model
computes the class of \emph{regular functions}, which is often considered
as one of the functional counterparts of \emph{regular languages}. It has been
largely studied for its numerous regular-like properties:
closure under composition~\cite{chytil1977serial}, equivalence
with logical transductions~\cite{engelfriet2001mso} or
regular transducer expressions~\cite{dave2018regular}, decidable
equivalence problem~\cite{gurari1982equivalence}, etc.

\subparagraph*{Pebble transducers and blind transducers.}
Two-way transducers can only describe functions
whose output size is at most linear in the input size.
A possible solution to overcome this limitation
is to consider nested two-way transducers. In particular,
the nested model of \emph{pebble transducers} has been studied
for a long time (see e.g. \cite{globerman1996complexity,engelfriet2015two}).

A \emph{$k$-pebble transducer} is built by induction
on $k \ge 1$. For $k=1$, a $1$-pebble transducer
is just a two-way transducer. For $k \ge 2$, a $k$-pebble transducer
is a two-way transducer that, when 
on any position $i$ of its input word, can launch
a $(k{-}1)$-pebble transducer. This submachine works
on the original input where position $i$ is marked by a ``pebble''.
The original two-way transducer then outputs the concatenation
of all the outputs returned by the submachines
that it has launched along its computation.
The intuitive behavior of a $3$-pebble transducer is
depicted in Figure~\ref{fig:blind}. It can be seen
as program with $3$ nested  loops.
The class of word-to-word functions computed
by $k$-pebble transducers for some $k \ge 1$
is known as \emph{polyregular functions}.
It has been quite intensively studied over the past
few years due to its regular-like properties such as
closure under composition~\cite{engelfriet2015two},
equivalence with logical interpretations~\cite{bojanczyk2019string}
or other transducer models~\cite{bojanczyk2018polyregular}, etc.

\begin{figure}[h!]

\centering
\begin{tikzpicture}{scale=1}

	\newcommand{\couleur}{blue}
	\newcommand{\texte}{\small \bfseries \sffamily \mathversion{bold} }

	\draw (-0.25,4.2) rectangle (8.25,4.5);
	\node[above] at (6,4.1) {$\substack{\text{Input word}}$};
	\node[above] at (0.05,4.1) {$\lmark$};
	\node[above] at (7.95,4.1) {$\rmark$};

	\node[above] at (-1.7,3.55) {\textcolor{\couleur}{\texte Main machine}};
	\draw[-,thick,\couleur](0,4) -- (7,4);
	\draw[-,thick, \couleur] (7,4) arc (90:-90:0.1);
	\draw[-,thick,\couleur](7,3.8) -- (1.5,3.8);	
	\draw[-, thick,\couleur] (1.5,3.8) arc (90:270:0.1);
	\draw[-,thick,\couleur](1.5,3.6) -- (2.5,3.6);	
	\fill[fill = \couleur,even odd rule] (2.5,3.6) circle (0.08);


	\draw (-0.25,3) rectangle (8.25,3.3);
	\node[above] at (6,2.9) {$\substack{\text{Input word}}$};
	\node[above] at (0.05,2.9) {$\lmark$};
	\node[above] at (7.95,2.9) {$\rmark$};
	\fill[fill = \couleur] (2.5,3.15) circle (0.15);

	\node[above] at (-2.5,2.35) {\textcolor{red!60}{\texte Submachine launched in $\textcolor{blue}{\bullet}$} };
	\draw[-,thick,red!60](0,2.8) -- (3,2.8);
	\draw[-,thick, red!60] (3,2.8) arc (90:-90:0.1);
	\draw[-,thick,red!60](3,2.6) -- (1,2.6);	
	\draw[-, thick,red!60] (1,2.6) arc (90:270:0.1);
	\draw[-,thick,red!60](1,2.4) -- (5,2.4);	
	\fill[fill = red!60,even odd rule] (5,2.4) circle (0.08);

        \draw[->,very thick,\couleur,dashed](2.5,3.6) to[out= -120, in = 60] (0,2.85);	
	
	\node[above,\couleur] at (2.5,2.67) {$\substack{\text{pebble}}$};


	\draw (-0.25,1.8) rectangle (8.25,2.1);
	\node[above] at (6,1.7) {$\substack{\text{Input word}}$};
	\node[above] at (0.05,1.7) {$\lmark$};
	\node[above] at (7.95,1.7) {$\rmark$};
	\fill[fill = red!60] (5,1.95) circle (0.15);
	\fill[fill = \couleur] (2.5,1.95) circle (0.15);

	\node[above] at (-2.74,1.15) {\textcolor{brown}
	{\texte Subsubmachine launched in $\textcolor{red!60}{\bullet}$} };
	\draw[-,thick,brown](0,1.6) -- (1,1.6);
	\draw[-,thick, brown](1,1.6) arc (90:-90:0.1);
	\draw[-,thick,brown] (1,1.4) -- (0.5,1.4);	
	\draw[-, thick,brown] (0.5,1.4) arc (90:270:0.1);
	\draw[-,thick,brown] (0.5,1.2) -- (6,1.2);	
	\fill[fill = brown,even odd rule] (6,1.2) circle (0.08);

        \draw[->,very thick,red!60,dashed](5,2.4) to[out= -150, in = 40] (0,1.65);	
	
	\node[above,red!60] at (5,1.47) {$\substack{\text{pebble}}$};
	\node[above,\couleur] at (2.5,1.47) {$\substack{\text{pebble}}$};

\end{tikzpicture}

\caption{\label{fig:pebble} Behavior of a $3$-pebble transducer}

\end{figure}
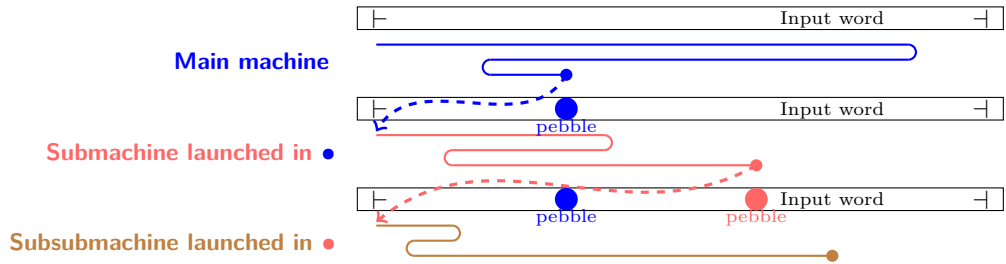

A subclass of pebble transducers named 
\emph{blind transducers} was recently
introduced in~\cite{nous2020comparison}.
For $k=1$, a $1$-blind transducer
is just a two-way transducer. For $k \ge 2$, a $k$-blind transducer
is a two-way transducer that can launch a
$(k{-}1)$-blind transducer like a $k$-pebble transducer.
However, there is no pebble marking the input of the submachine
 (i.e. it cannot see the position $i$ from which it was called).
 The behavior of a $3$-blind transducer is depicted in Figure~\ref{fig:blind}.
 It can be seen as a program with $3$ nested loops
 which cannot see the upper loop indexes.
We call \emph{polyblind functions} the class
of functions computed by blind transducers.
It is closed under composition
and deeply related to lambda-calculus~\cite{nguyen2021implicit}.

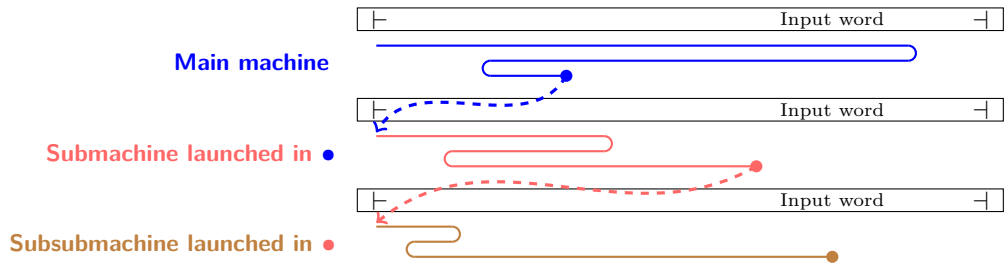
\begin{figure}[h!]

\centering
\begin{tikzpicture}{scale=1}

	\newcommand{\couleur}{blue}
	\newcommand{\texte}{\small \bfseries \sffamily \mathversion{bold} }

	
	\draw (-0.25,4.2) rectangle (8.25,4.5);
	\node[above] at (6,4.1) {$\substack{\text{Input word}}$};
	\node[above] at (0.05,4.1) {$\lmark$};
	\node[above] at (7.95,4.1) {$\rmark$};

	\node[above] at (-1.7,3.55) {\textcolor{\couleur}{\texte Main machine}};
	\draw[-,thick,\couleur](0,4) -- (7,4);
	\draw[-,thick, \couleur] (7,4) arc (90:-90:0.1);
	\draw[-,thick,\couleur](7,3.8) -- (1.5,3.8);	
	\draw[-, thick,\couleur] (1.5,3.8) arc (90:270:0.1);
	\draw[-,thick,\couleur](1.5,3.6) -- (2.5,3.6);	
	\fill[fill = \couleur,even odd rule] (2.5,3.6) circle (0.08);

	\draw (-0.25,3) rectangle (8.25,3.3);
	\node[above] at (6,2.9) {$\substack{\text{Input word}}$};
	\node[above] at (0.05,2.9) {$\lmark$};
	\node[above] at (7.95,2.9) {$\rmark$};

	\node[above] at (-2.5,2.35) {\textcolor{red!60}{\texte Submachine launched in $\textcolor{blue}{\bullet}$}};
	\draw[-,thick,red!60](0,2.8) -- (3,2.8);
	\draw[-,thick, red!60] (3,2.8) arc (90:-90:0.1);
	\draw[-,thick,red!60](3,2.6) -- (1,2.6);	
	\draw[-, thick,red!60] (1,2.6) arc (90:270:0.1);
	\draw[-,thick,red!60](1,2.4) -- (5,2.4);		
	\fill[fill = red!60,even odd rule] (5,2.4) circle (0.08);

        \draw[->,very thick,\couleur,dashed](2.5,3.6) to[out= -120, in = 60] (0,2.85);	
        

	\draw (-0.25,1.8) rectangle (8.25,2.1);
	\node[above] at (6,1.7) {$\substack{\text{Input word}}$};
	\node[above] at (0.05,1.7) {$\lmark$};
	\node[above] at (7.95,1.7) {$\rmark$};

	\node[above] at (-2.74,1.15) {\textcolor{brown}
	{\texte Subsubmachine launched in $\textcolor{red!60}{\bullet}$} };
	\draw[-,thick,brown](0,1.6) -- (1,1.6);
	\draw[-,thick, brown](1,1.6) arc (90:-90:0.1);
	\draw[-,thick,brown] (1,1.4) -- (0.5,1.4);	
	\draw[-, thick,brown] (0.5,1.4) arc (90:270:0.1);
	\draw[-,thick,brown] (0.5,1.2) -- (6,1.2);	
	\fill[fill = brown,even odd rule] (6,1.2) circle (0.08);

        \draw[->,very thick,red!60,dashed](5,2.4) to[out= -150, in = 40] (0,1.65);

\end{tikzpicture}
\caption{\label{fig:blind}  Behavior of a $3$-blind transducer}
\end{figure}

We study here polyregular and polyblind functions 
whose output alphabet is unary. Up to identifying
a word with its length, we thus consider 
 functions from finite words to $\Nat$.
 With this restriction, we show that one can decide if a polyregular function
 is polyblind, and connect
 these classes of functions to \emph{rational series}.

\subparagraph*{Relationship with rational series.}
Rational series over the semiring $(\Nat,+, \times)$.
are a well-studied class of
 functions from finite words to $\Nat$.
 They can be defined as the closure
 of (unary output) regular functions under sum $+$, Cauchy product
 $\cau$ (product for formal power series) and Kleene star $*$
(iteration of Cauchy products). It is also
 well known that rational series are closed under Hadamard product $\had$
 (component-wise product) \cite{berstel2011noncommutative}.

The first result of this paper states that polyregular functions
 are exactly the subclass of rational series ``without star'',
 that is the closure of regular functions under $+$ and $\cau$
($\had$ can also be used but it is not necessary).
This theorem is obtained by combining several former works.
Our second result establishes that polyblind
functions are exactly the closure of regular functions under $+$ and $\had$.
It is shown in a self-contained way.

The aforementioned classes are depicted in Figure~\ref{fig:conclu}.
All the inclusions are strict and this paper
provides a few separating examples (some of
them were already known in \cite{doueneau2021pebble}).

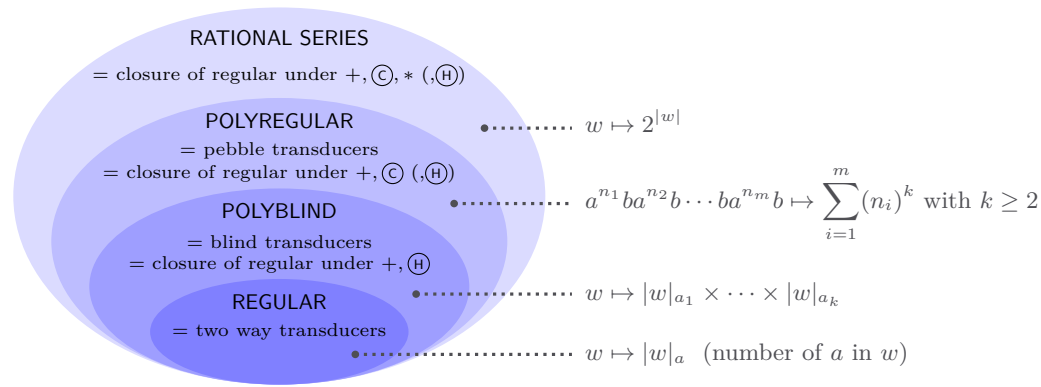
\begin{figure}[h!]

    	\begin{center}
	\hspace*{-0.4cm}
        \begin{tikzpicture}{scale=0.9}
        
            \fill[fill = blue!14]  (0,1.8) ellipse (3.5cm and 2.5cm);
            \fill[fill = blue!26]  (0,1.2) ellipse (3cm and 1.9cm);
            \fill[fill = blue!38]  (0,0.6) ellipse (2.5cm and 1.3cm);
            \fill[fill = blue!50] (0,0) ellipse (1.7cm and 0.7cm);

            \draw(0,3.9) node  {\scalebox{0.8}{$\textsf{RATIONAL SERIES}$}};
            \draw(0,3.4) node {$\substack{\text{= closure
            of regular under $+, \cau, *$ (,$\had$)}}$};

            \draw(0,0.4) node  {\scalebox{0.8}{$\textsf{REGULAR}$}};
            \draw(0,0) node {$\substack{\text{= two way transducers}}$};
            
          \draw(0,1.6) node  {\scalebox{0.8}{$\textsf{POLYBLIND}$}};     
           \draw(0,1.2) node {$\substack{\text{= blind
            transducers}}$};   
            \draw(0,0.9) node {$\substack{\text{= closure of
            regular under $+, \had$}}$};
            
            \draw(0,2.8) node  {\scalebox{0.8}{$\textsf{POLYREGULAR}$}};
            \draw(0,2.4) node {$\substack{\text{= pebble
            transducers}}$};    
            \draw(0,2.1) node {$\substack{\text{= closure
            of regular under $+, \cau$ (,$\had$)}}$};

             \draw[darkgray,dotted,very thick] (1,-0.3) -- (3.85,-0.3);
             \fill[darkgray] {(1,-0.3) circle (0.05)};
             \node[right] at (3.9,-0.3) {\small \textcolor{darkgray}{$w \mapsto |w|_a$~~(number of $a$ in $w$)}};

             \draw[darkgray,dotted,very thick] (1.8,0.5) -- (3.85,0.5);
             \fill[darkgray] {(1.8,0.5) circle (0.05)};
             \node[right] at (3.9,0.5) {\small \textcolor{darkgray}{$w \mapsto |w|_{a_1} \times
             \cdots \times  |w|_{a_k}$}};

             \draw[darkgray,dotted,very thick] (2.3,1.7) -- (3.85,1.7);
             \fill[darkgray] {(2.3,1.7) circle (0.05)};
             \node[right] at (3.9,1.7) {\small \textcolor{darkgray}{$\displaystyle a^{n_1} b a^{n_2} b \cdots  b a^{n_m} b \mapsto \sum_{i=1}^m  (n_i)^k$ with $k \ge 2$}};

             \draw[darkgray,dotted,very thick] (2.7,2.7) -- (3.85,2.7);
             \fill[darkgray] {(2.7,2.7) circle (0.05)};
             \node[right] at (3.9,2.8) {\small \textcolor{darkgray}{$w \mapsto 2^{|w|}$}};

             
        \end{tikzpicture}
        \end{center}
        
    \caption{\label{fig:conclu} Classes of functions from finite words to $ \Nat$ studied in this paper}
\end{figure}

\subparagraph*{Class membership problems.}
We finally show how to decide whether
a polyregular function with unary output is polyblind. It is by far the most
involved and technical result of this paper. Furthermore,
the construction of a blind transducer is effective, hence
this result can be viewed as program
optimization. Indeed, given a program with nested loops,
our algorithm is able to build an equivalent program using ``blind'' loops
if it exists.

In general, decision problems for transductions
are quite difficult to solve, since contrary to regular languages,
there are no known ``canonical'' objects (such as a minimal automaton)
to represent (poly)regular functions. It is thus complex to
decide an intrinsic property of a function,
since it can be described in several seemingly unrelated
manners. Nevertheless, the membership problem from rational series to
polyregular (resp. regular) functions was shown to be decidable
in  \cite{doueneau2020register, doueneau2021pebble}.
It is in fact equivalent to checking if the output of the rational series
is bounded by a polynomial (resp. a linear function)
in its input's length.

However, both polyregular and polyblind functions can have
polynomial growth rates. To discriminate between them,
we thus introduce the new notion of \emph{repetitiveness}
(which is a pumping-like property
for functions) and show that it exactly
captures the polyregular functions
that are polyblind. The proof is a rather complex induction
on the depth $k \ge 1$ of the $k$-pebble transducer representing
the function. We show at the same time that repetitiveness
is decidable and that a blind transducer can effectively be built
whenever this property holds.
Partial results were obtained in \cite{doueneau2021pebble}
to decide ``blindness'' of the functions computed by
$2$-pebble transducers.
Some of our tools are inspired by this paper,
such as the use of bimachines and factorization forests.
Nevertheless, our general result requires new proof techniques (e.g. the induction
techniques which insulate the term of
``highest growth rate'' in the function) and concepts
(e.g. repetitiveness).

\subparagraph*{Outline.}
We first describe in Section~\ref{sec:prelim}
the notions of pebble and blind transducers.
In the case of unary outputs, we recall the equivalent
models of pebble, marble and blind bimachines introduced in
\cite{doueneau2021pebble}. These bimachines 
are easier to handle in the proofs, since they manipulate 
a monoid morphism instead of having two-way moves.
In Section~\ref{sec:rational} we recall the definitions
of sum $+$, Cauchy product $\cau$, Hadamard product
$\had$ and Kleene star $*$ for rational series. We
then show how to describe polyregular and polyblind functions
with these operations. Finally, we claim in Section~\ref{sec:membership}
that the membership problem from polyregular to polyblind
functions is decidable. The proof of this technical result
is sketched in sections~\ref{sec:tech1} and~\ref{sec:tech2}.
Due to space constraints we focus on
the most significant lemmas.

\section{Preliminaries}

\label{sec:prelim}

$\Nat$ is the set of nonnegative integers.
If $ i \le j$, the set $[i{:}j]$ is $\{i, i{+}1, \dots, j\} \subseteq \Nat$
(empty if $j <i$). The capital letter $A$ denotes an alphabet,
i.e. a finite set of letters.
The empty word is denoted by $\vide$. 
If $w \in A^*$, let $|w| \in \Nat$ be its length,
and for $1 \le i \le |w|$ let $w[i]$ be its $i$-th letter.
If $I = \{i_1 < \cdots < i_{\ell}\} \subseteq [1{:}|w|]$,
 let $w[I] \defined w[i_1] \cdots w[i_{\ell}]$.
If $a \in A$, let $|w|_a$ be the number of letters $a$ occurring in $w$. 
We assume that the reader is familiar with the basics of automata theory,
in particular one-way and two-way automata, and monoid morphisms.

\subparagraph*{Two-way transducers.}
A deterministic two-way transducer is a deterministic
two-way automaton (with input in $A^*$) enhanced with the ability to produce
outputs (from $B^*$) when performing a transition. The output of the transducer is
defined as the concatenation of these productions along the unique accepting
run on the input word (if it exists): it thus describes
a (partial) function  $A^* \rightarrow B^*$. Its behavior is 
depicted in Figure~\ref{fig:twotwo}.
A formal definition can be found e.g. in \cite{doueneau2021pebble}.
These machines compute the class of \emph{regular functions}.

\begin{figure}[h!]

\centering
\begin{tikzpicture}{scale=1}

	\newcommand{\couleur}{blue}
	\newcommand{\texte}{\small \bfseries \sffamily \mathversion{bold} }

	
	\draw (-0.25,4.2) rectangle (8.25,4.5);
	\node[above] at (6,4.1) {$\substack{\text{Input word}}$};
	\node[above] at (0.05,4.1) {$\lmark$};
	\node[above] at (7.95,4.1) {$\rmark$};

	\node[above] at (-1.7,3.55) {\textcolor{\couleur}{\texte Run of the machine}};
	\draw[-,thick,\couleur](0,4) -- (7,4);
	\draw[-,thick, \couleur] (7,4) arc (90:-90:0.1);
	\draw[-,thick,\couleur](7,3.8) -- (1.5,3.8);	
	\draw[-, thick,\couleur] (1.5,3.8) arc (90:270:0.1);
	\draw[-,thick,\couleur](1.5,3.6) -- (2.5,3.6);	
	\fill[fill = \couleur,even odd rule] (2.5,3.6) circle (0.08);

\end{tikzpicture}
\caption{\label{fig:twotwo}  Behavior of a two-way transducer}
\end{figure}
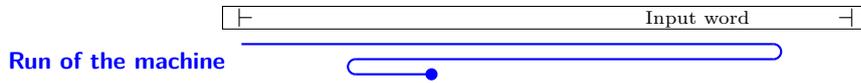

\begin{example}
The function $a_1 \cdots a_n \mapsto a_1 \cdots a_n\#a_n \cdots a_1$ can be
computed by a two-way transducer which reads its input
word from left to right and then from right to left.
\end{example}
From now on, the output alphabet $B$ of our machines will always
be a singleton. By identifying $B^*$ and $\Nat$, we assume that
the functions computed have type  $A^* \rightarrow \Nat$.
\begin{example} Given $a \in A$,
the function $\nb_{a}: A^* \rightarrow \Nat, w \mapsto |w|_a$
is regular.
\end{example}

\subparagraph*{Blind and pebble transducers.} Blind and pebble transducers
extend two-way transducers by allowing to ``nest'' such machines.
A $1$-blind (resp. $1$-pebble) transducer
is just a two-way transducer.
For $k \ge 2$, a $k$-blind (resp. $k$-pebble) transducer
is a two-way transducer which, when performing a transition
from a position $1 \le i \le |w|$ of its input $w \in A^*$, can launch
a $(k{-}1)$-blind (resp. $(k{-}1)$-pebble) transducer with input $w$
(resp. $w[1{:}i{-}1] \marq{w[i]} w[i{+}1{:}|w|]$
i.e. $w$ where position $i$ is marked).
The two-way transducer then uses the
output of this submachine as if it was the output produced
along its transition. The intuitive behaviors are depicted
for $k=3$ in figures~\ref{fig:pebble} and~\ref{fig:blind}.
Formal definitions can be found e.g.
 in~\cite{nous2020comparison,
doueneau2021pebble}.

\begin{example} \label{ex:nb} Let $a_1, \dots, a_k \in A$, then
$\nb_{a_1, \dots, a_k}: w \mapsto |w|_{a_1} \times \cdots \times |w|_{a_k}$
can be computed by a $k$-blind transducer. The main transducer
processes its input from left to right, and it calls inductively  a $(k{-}1)$-blind transducer
for  $\nb_{a_1, \dots, a_{k-1}}$ each time it it sees an $a_k$.
\end{example}

\begin{example} \label{ex:itpow} The function $\itpow{2}: a^{n_1} b \cdots  a^{n_m} b
\mapsto \sum_{i=1}^{m} n_i^2$ can be computed by a $2$-pebble transducer.
Its main transducer ranges over all the $a$ of the input,
and calls a $1$-pebble (= two-way) transducer for each $a$,
which produces $n_i$ if the $a$ is in the $i$-th block
(it uses the pebble to detect which block is concerned).
Similarily, the function $\itpow{k}: a^{n_1} b \cdots  a^{n_m} b \mapsto \sum_{i=1}^{m} n_i^k$
for $k \ge 1$ can be computed by a $k$-pebble transducer. 
\end{example}

\begin{definition}
We define the class of \emph{polyregular functions}
(resp. \emph{polyblind functions}) as the class of
functions computed by a $k$-pebble (resp. $k$-blind) transducer 
 for some $k \ge 1$.
\end{definition}
It is not hard to see that polyblind functions are a subclass of polyregular functions.
Indeed, a blind transducer is just a pebble transducer ``without pebbles''.

\subparagraph*{Bimachines.} 
In this paper, we shall describe formally the regular, polyregular and polyblind functions with
another computation model. A \emph{bimachine} is a transducer which makes
a single left-to-right pass on its input,
but it can use a morphism into a finite monoid to check regular properties
of the prefix (resp. suffix) ending  (resp. starting) in the current position.
This notion of bimachines enable us to easily use algebraic techniques
in the proofs, and in particular \emph{factorization forests}
over finite monoids.

\begin{definition}[] \label{def:bimachine-s}
A bimachine  $\mach \defined (A, M, \mu, \lambda)$ is:
\begin{itemize}
\item an input alphabet $A$, a finite monoid $M$ and a monoid morphism $\mu: A^* \rightarrow M$;
\item an output function $\lambda: M \times A \times M \rightarrow \Nat$.
\end{itemize}
\end{definition}
$\mach$ computes $f: A^* \rightarrow \Nat$
defined by
$
f(w) \defined \sum_{i=1}^{|w|} \lambda(\mu(w[1{:}i{-}1]), w[i], \mu(w[i{+}1{:}|w|]))
$
for $w \in A^*$.
The production of each term of this sum is depicted intuitively
in Figure~\ref{fig:bim-s}.

 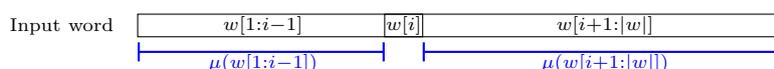
\begin{figure}[h!]

\centering
\begin{tikzpicture}{scale=1}

	\newcommand{\couleur}{blue}
	\newcommand{\texte}{\small \bfseries \sffamily \mathversion{bold} }

	\draw (-0.25,4.2) rectangle (8.25,4.5);
	\draw (3,4.2) rectangle (3.5,4.5);
	\node[above] at (-1.25,4.1) {$\substack{\text{Input word}}$};

	\draw[|-|,thick,\couleur](-0.25,4) -- (3,4);
	\draw[|-|,thick,\couleur](3.5,4) -- (8.25,4);	
	\node[above] at (3.25,4.105) {$\substack{w[i]}$};
	\node[above] at (1.4,4.105) {$\substack{w[1{:}i{-}1]}$};
	\node[above] at (5.9,4.105) {$\substack{w[i{+}1{:}|w|]}$};
	\node[above,\couleur] at (1.4,3.6) {$\substack{\mu(w[1{:}i{-}1])}$};
	\node[above,\couleur] at (5.9,3.6) {$\substack{\mu(w[i{+}1{:}|w|])}$};

\end{tikzpicture}

\caption{\label{fig:bim-s} Behavior of a bimachine when producing
$\lambda(\mu(w[1{:}i{-}1]), w[i], \mu(w[i{+}1{:}|w|))$}
\end{figure}

When dealing with bimachines, we consider without loss of generality total functions
such that $f(\vide) = 0$ (the domains of two-way transducers are
regular languages \cite{shepherdson1959reduction},
and the particular image of $\vide$ does not matter).
In this context, it is well known that regular
functions of type $A^* \rightarrow \Nat$ are exactly those computed by bimachines
(in other words, it means that regular functions 
and rational functions are the same).
Now we recall how \cite{doueneau2021pebble} has generalized this result to
$k$-blind and $k$-pebble transducers.
Intuitively, a \emph{bimachine with external functions} is a bimachine
enriched with the possibility to launch a submachine
for each letter of the input (it outputs the sum of all the outputs 
returned by these submachines).

\begin{definition}[\cite{doueneau2021pebble}] \label{def:bimachine-external}
	A \emph{bimachine with external pebble (}resp. \emph{external blind},
	resp. \emph{external marble)} \emph{functions} $\mach = (A,M,\mu,\oras,\lambda)$ consists of:
	\item
	\begin{itemize}
	\item an input alphabet $A$;
	\item a finite monoid $M$ and a morphism $\mu: A^* \rightarrow M$;
	\item a finite set $\oras$ of external functions $\exte: (A \uplus \marq{A})^* \rightarrow \Nat$
	(resp. $A^* \rightarrow \Nat$, resp. $A^* \rightarrow \Nat$);
	\item an output function  $\lambda: M\times A \times M \rightarrow \oras$.
	\end{itemize}
	\vspace*{0.2cm}
\end{definition}
\label{page:extei}
Given $1 \le i \le |w|$ a position of $w \in A^*$, let
 \mbox{$\exte_i \defined \lambda(\mu(w[1{:}{i{-}1}]),w[i],\mu(w[i{+}1{:}|w|])) \in \oras$}.
A bimachine $\mach$ with external pebble functions computes
a function $f: A^* \rightarrow \Nat$
defined by $f(w) \defined \sum_{1 \le i \le |w|} \exte_i(w[1{:}i{-}1]\marq{w[i]}w[i{+}1{:}|w|])$. 
Intuitively, it means that for each position $1 \le i \le |w|$, $\mach$ calls an external
function $\exte_i$ (depending on a regular
property of $w[1{:}i{-}1]$,$w[i]$ and $w[i{+}1{:}|w|]$)
with input $w[1{:}i{-}1]\marq{w[i]}w[i{+}1{:}|w|]$ (that is ``$w$ with a pebble
on position $i$''),
and uses the result $\exte_i(w[1{:}i{-}1]\marq{w[i]}w[i{+}1{:}|w|])$
of this function in its own output.

For a bimachine with external blind (resp. marble) functions,
the output is defined by $f(w) \defined \sum_{1 \le i \le |w|} \exte_i(w)$
(resp. $f(w) \defined \sum_{1 \le i \le |w|} \exte_i(w[1{:}i])$). In this case
$\mach$ calls $\exte_i$ with argument $w$ (resp. the prefix of $w$ ending in 
position $i$).

\begin{definition}[\cite{doueneau2021pebble}]
Given $k \ge 1$, a $k$-pebble (resp. $k$-blind, resp. $k$-marble) bimachine is:
\begin{itemize}
\item for $k = 1$, a bimachine (withtout external functions, see Definition~\ref{def:bimachine-s});
\item for $k \ge 2$, it a bimachine with external pebble (resp. external blind,
resp. external marble) functions
(see Definition~\ref{def:bimachine-external})
which are computed by $(k{-}1)$-pebble (resp. $(k{-}1)$-blind, resp. $(k{-}1)$-marble) bimachines.
These $(k{-}1)$-bimachines are implicitly fixed and given by the
external functions of the main bimachine.
\end{itemize}
\end{definition}

\begin{remark} For pebble bimachines,
a natural question is whether the inner bimachines can ask which
pebble was dropped by which ancestor, or whether they can only see
that ``there is a pebble''. Both models are in fact equivalent,
since the number of pebbles is bounded.
\end{remark}
Now we recall in Theorem~\ref{theo:equiv}
that pebble and blind bimachines are respectively equivalent to the
aforementioned pebble and blind transducers. More interestingly,
the marble bimachines (which call ``prefixes'')
also correspond to pebble transducers.
In our proofs, we shall preferentially use the
marble model to represent polyregular functions.

\begin{theorem}[\cite{doueneau2021pebble}] \label{theo:equiv}
For all $k \ge 1$, $k$-pebble transducers, $k$-pebble bimachines and
$k$-marble bimachines compute the same class
of functions. Furthermore, $k$-blind transducers
and $k$-blind bimachines compute the same class
of functions.
\end{theorem}

\section{From pebbles to rational series}

\label{sec:rational}

The class of \emph{rational series} (which are
total functions $A^* \rightarrow \Nat$)
is the closure of regular functions
under sum, Cauchy product and Kleene star.
It is well known that it can be described by \emph{weighted automata},
furthermore this class is closed under Hadamard product (see e.g.
\cite[Theorem 5.5]{berstel2011noncommutative}).
Let us recall the definition of these operations for $f,g: A^* \rightarrow \Nat$:
\begin{itemize}
\item the sum $f+g: w \mapsto f(w) + g(w)$;
\item the Cauchy product $ f \cau g:w \mapsto \sum_{i=0}^{|w|} f(w[1{:}i])g(w[i{+}1{:}|w|])$;
\item the Hadamard product  $f \had g: w \mapsto f(w) g(w)$;
\item if and only if $f(\vide) = 0$, the Kleene star
$ f^* \defined \sum_{n \ge 0} f^n$ where
$f^{0}: \vide \mapsto 1, u \neq \vide \mapsto 0$
is neutral for Cauchy product
and $f^{n+1} \defined f \cau f^{n}$.
\end{itemize}

\begin{example} In  Example~\ref{ex:nb} we have
$\nb_{a_1, \dots, a_k} = \nb_{a_1} \had \cdots \had \nb_{a_k}$.
\end{example}

\begin{example} Let $f,g: \{a,b\}^* \rightarrow \Nat$ defined
by $f(wa) = 2$ for $w \in A^*$ and $f(w) = 0$ otherwise;
 $g(a^n b w) = n$ for $w \in A^*$ and $g(w) = 0$ otherwise.
It is easy to see that $f \cau g(a^{n_1} b \cdots  a^{n_m} b)
= \sum_{i=1}^{m} n_i(n_i{-}1)$.
Hence $\itpow{2} = \nb_{a} + f \cau g$
(see Example~\ref{ex:itpow}).
\end{example}
We are now ready to state the first main results of this paper.
The first one shows that polyregular functions
correspond to the subclass of rational series where the use
of star is disallowed (and it also corresponds to rational
series whose ``growth'' is polynomial).

\begin{theorem} \label{theo:cauchy}
Let $f:A^* \rightarrow \Nat$, the following conditions are (effectively) equivalent:
\begin{enumerate}
\item \label{it:pr:poly} $f$ is polyregular;
\item   \label{it:pr:rat} $f$ is a rational series
with polynomial growth, i.e. $f(w) = \mc{O}(|w|^k)$ for some $k \ge 1$;
\item \label{it:pr:clos} $f$ belongs to
the closure of regular functions under sum and Cauchy product;
\item  \label{it:pr:closh}  $f$ belongs to 
the closure of regular functions under sum,
Cauchy and Hadamard products.
\end{enumerate}
\end{theorem}

\begin{proof}[Proof sketch.] $\ref{it:pr:poly} \Leftrightarrow~\ref{it:pr:rat}$ follows from
\cite{doueneau2020register,doueneau2021pebble}.
For $\ref{it:pr:rat} \implies~\ref{it:pr:clos}$,
it is shown in \cite[Exercise~1.2 of Chapter~9]{berstel2011noncommutative}
that the rational series of polynomial growth can be obtained
by sum and Cauchy products from the characteristic series
of rational languages (which are regular functions).
Having $\ref{it:pr:clos} \implies~\ref{it:pr:closh}$ is obvious.
Finally for $\ref{it:pr:closh} \implies~\ref{it:pr:rat}$, it
suffices to note that regular functions have polynomial growth, 
and this property is preserved by $+, \cau$ and $\had$.
\end{proof}
Note that Hadamard products do not increase
the expressive power of this class. However,
removing Cauchy products gives polyblind functions,
as shown in Theorem~\ref{theo:hadamard}.

\begin{theorem} \label{theo:hadamard}
Let $f:A^* \rightarrow \Nat$, the following conditions are equivalent:
\begin{enumerate}
\item \label{bd:it:poly} $f$ is polyblind;
\item \label{bd:it:clos} $f$ belongs to 
the closure of regular functions under sum
and Hadamard product.
\end{enumerate}
\end{theorem}

\begin{proof}[Proof sketch]
We first show $\ref{bd:it:poly} \implies~\ref{bd:it:clos}$ by induction on $k \ge 1$
when $f$ is computed by a $k$-blind bimachine.
For $k = 1$, it is obvious. Let us
describe the induction step from $k \ge 1$ to $k{+}1$.
Let $f$ be computed by a bimachine $\mach = (A,M, \mu, \oras, \lambda)$,
with external blind functions
computed by $k$-blind bimachines.
 Given $\exte \in \oras$, let $f'_{\exte}$ be the function
 which maps $w\in A^*$ to the cardinal
 $|\{1 \le i \le |w|: \lambda(\mu(w[1{:}i{-}1]), w[i], \mu(w[i{+}1{:}|w|])) = \exte \}|$.
 Then $f'_{\exte}$ is a regular function and
$ f = \sum_{\exte \in \oras} f'_{\exte} \had \exte$.
 The result follows by induction hypothesis.
 
For $\ref{bd:it:clos} \implies~\ref{bd:it:poly}$, we show
that functions computed by blind bimachines
are closed under sum and Hadamard product.
For the sum, we use the blind transducer model
to simulate successively the computation
of the two terms of a sum.
For the Hadamard product, we show by induction
on $k \ge 1$ that if $f$ is computed
by a $k$-blind bimachine and $g$ is polyblind, 
then $f \had g$ is polyblind. If $k = 1$, we transform
the bimachine for $f$ in a blind bimachine
computing $f \had g$ by replacing each output $n \in \Nat$
by a call to a machine computing $n \times g$. If $k \ge 2$,
using  the notations of the previous paragraph we have
$ f = \sum_{\exte \in \oras} f'_{\exte} \had \exte$, thus
$ f \had g = \sum_{\exte \in \oras} f'_{\exte} \had (\exte \had g)$.
By induction hypothesis, each $\exte \had g$
is polyblind. Hence we compute $f \had g$
by replacing in $\mach$ each external function
$\exte$ by the function $\exte \had g$.
\end{proof}

We conclude this section by recalling that polyregular
(resp. polyblind) functions with unary output enjoy a
``pebble minimization'' property, which allows to reduce
the number of nested layers depending on the growth
rate of the output.

\begin{definition} \label{def:growth}
We say that a function $f: A^* \rightarrow \Nat$
has \emph{growth at most $k$} if $f(w) = \mc{O}(|w|^k)$.
\end{definition}

\begin{theorem}[\cite{doueneau2020register,doueneau2021pebble, nous2020comparison}]
\label{theo:minimize}
A a polyregular (resp. polyblind) function $f$
can be computed by a $k$-pebble (resp. $k$-blind) transducer
if and only it has growth at most $k$.
Furthermore this property can be decided
and the construction is effective.
\end{theorem}

\begin{remark} The result also holds for blind
transducers with non-unary outputs \cite{nous2020comparison}.
However, it turns out to be false for pebble transducers 
with non-unary outputs (unpublished
work of the author of \cite{bojanczyk2018polyregular}).
\end{remark}
In the next section, we complete the decidability
picture by solving the membership problem
from polyregular to polyblind functions.

\section{Membership problem from polyregular to polyblind}

\label{sec:membership}

In this section, we state the most technical and interesting
result of this paper, which consists in deciding
if a polyregular function is polyblind. We also give in Theorem~\ref{theo:readable-car}
a semantical characterization of polyregular functions
which are polyblind, using the notion of \emph{repetitiveness}.
Intuitively, a $k$-repetitive function is a function which,
when given two places in a word where the same
factor is repeated, cannot distinguish between the first and the second place.
Hence its output only depends on the total number
of repetitions of the factor.

\begin{definition}[Repetitive function]
\label{def:repetitive} Let $k \ge 1$.
We say that $f: A^* \rightarrow \Nat$ is \emph{$k$-repetitive}
if there exists $\idem{} \ge 1$,
such that the following holds for all
$\alpha,\alpha_0,  {u_1}, \alpha_1, \dots,  {u_k}, \alpha_k, \beta \in A^*$
and $\omega \ge 1$ multiple of $\idem{}$.
Let  $W: \Nat^{k} \rightarrow A^*$ defined by:
\begin{equation*}
W : X_1, \dots, X_k \mapsto 
\left(\alpha_0 \prod_{i=1}^k  {u_i}^{\omega X_i}\alpha_i\right).
\end{equation*}
and let $w \defined W(1, \dots, 1)$.
Then there exists a function
$F: \Nat^{k} \rightarrow \Nat$ such that
for all $\ov{X} \defined X_1, \dots, X_k \ge 3$
and $\ov{Y} \defined Y_1, \dots, Y_k \ge 3$, we have:
\begin{equation*}
f(\alpha w^{2\omega- 1}W(\ov{X}) w^{\omega-1} W(\ov{Y}) w^\omega \beta)
= F(X_1 + Y_1, \dots, X_k + Y_k)
\end{equation*}
\end{definition}

\begin{remark}
If $f$ is $k$-repetitive, $f$ is also $(k{-}1)$-repetitive,
by considering an empty $u_k$.
\end{remark}
Let us now give a few examples
in order so see when this criterion holds, or not.

\begin{example} The function $\nb_{a_1, \dots, a_k}$ (see Example~\ref{ex:nb})
is $k$-repetitive for all $k \ge 1$.
\end{example}

\begin{example} \label{ex:repun}
If $f: A^* \rightarrow \Nat$ where $A = \{a\}$ is a singleton, then $f$
is $k$-repetitive for all $k \ge 1$. Indeed,
$\alpha w^{2\omega- 1}W(\ov{X}) w^{\omega-1} W(\ov{Y}) w^\omega \beta \in A^*$
is itself a function of $X_1 + Y_1, \dots, X_k+Y_k$, hence so
is its image $f(\alpha w^{2\omega- 1}W(\ov{X}) w^{\omega-1} W(\ov{Y}) w^\omega)$.
\end{example}

\begin{example} \label{ex:not-repetitive} For all $k \ge 2$, the function
$\itpow{k}: a^{n_1} b \cdots  a^{n_m} b \mapsto \sum_{i=1}^{m} n_i^k$
is not $1$-repetitive. Let us choose any $\omega \ge 1$
and fix $\alpha = \beta \defined \vide$,
$u_1  \defined a$
and $\alpha_0 = \alpha_1 \defined b$,
then:
\begin{equation*}
\begin{aligned}
\itpow{k}(W(X_1, Y_1))
&= \itpow{k}\left((ba^\omega b)^{2\omega-1} b a^{\omega X_1} b (b a^\omega b)^{\omega-1} b a^{\omega Y_1} b(ba^\omega b)^{\omega}\right)\\
& = \omega^{k} (4\omega -2) + \omega^{k} X_1^{k} +  \omega^{k} Y_1^{k}  \\
\end{aligned}
\end{equation*}
which is \emph{not} a function of $X_1 + Y_1$ for $k \ge 2$.
\end{example}
We are now ready to state our main result of this paper.

\begin{theorem}
\label{theo:readable-car}
Let $k \ge 1$ and $f: A^* \rightarrow \Nat$ be computed by a $k$-pebble transducer.
Then $f$ is polyblind if and only if it is $k$-repetitive.
Furthermore, this property can be decided and one can
effectively build a blind transducer for $f$ when it exists.
\end{theorem}

\begin{remark} By Theorem~\ref{theo:minimize}, we can
even build a $k$-blind transducer.
\end{remark}
Theorem~\ref{theo:readable-car} provides
 a tool to show that some polyregular
functions are not polyblind:

\begin{example}
\label{ex:norma}
By Example~\ref{ex:not-repetitive}, the polyregular function $\itpow{k}$
 is not $k$-repetitive for $k \ge 2$. Therefore it is not polyblind.
Also note that it is computable by a $k$-pebble transducer,
but not by an $\ell$-pebble transducer for $\ell < k$
(Theorem~\ref{theo:minimize}).
\end{example}
An immediate consequence of
Theorem~\ref{theo:readable-car} is obtained below
for functions which have both a unary output alphabet and
a unary input alphabet.
This result was already obtained by
the authors of \cite{nous2020comparison}
using other techniques, in an unpublished work.

\begin{corollary} \label{cor:unary} Polyblind and polyregular functions
over a unary input alphabet coincide.
\end{corollary}

\begin{proof}[Proof.] Polyregular functions with unary input are $k$-repetitive
by Example~\ref{ex:repun}.
\end{proof}

\section{Repetitive functions and permutable bimachines}

\label{sec:tech1}

The rest of this paper is devoted to the proof of Theorem~\ref{theo:readable-car}.
Since $k$-pebble transducers and $k$-marble bimachines
compute the same functions (Theorem~\ref{theo:equiv}),
it follows directly from Theorem~\ref{theo:car} below
(the notions used are defined in the next sections).

From now on, we assume that a $k$-marble bimachine
uses the same morphism $\mu: A^* \rightarrow M$ in its main
bimachine, and inductively in all its submachines computing the external functions.
We do not lose any generality here, since it is always possible
to get this situation by taking the product of all the morphisms
used. We also assume that $\mu$ is surjective
(we do no lose any generality, since
it is possible to replace $M$ by the image $\mu(M)$).

\begin{theorem}
\label{theo:car}
Let $k \ge 2$ and $f: A^* \rightarrow \Nat$ be computed by
$\mach = (A,M, \mu, \oras, \lambda)$ a $k$-marble bimachine.
The following conditions are equivalent:
\begin{enumerate}
\item \label{it:char:poly} $f$ is polyblind;
\item  \label{it:char:k-rep} $f$ is $k$-repetitive;
\item \label{it:char:syntax} $\mach$ is $2^{\hei}$-permutable (see Definition~\ref{def:factor-perm})
and the function $f''$ (built from Proposition~\ref{prop:xxx}) is polyblind.
\end{enumerate}
Furthermore this property is decidable
and the construction is effective.
\end{theorem}
Let us now give the skeleton of the proof of Theorem~\ref{theo:car},
which is rather long and involved. The propositions
on which it relies are shown in the two following sections.

\begin{proof}[Proof of Theorem~\ref{theo:car}.]
$\ref{it:char:poly} \implies~\ref{it:char:k-rep}$
follows from Proposition~\ref{prop:blind-confuse}.
For $\ref{it:char:k-rep} \implies~\ref{it:char:syntax}$,
if $f$ is $k$-repetitive then $\mach$ is $2^{\hei}$-permutable by Proposition~\ref{prop:factor-perm}.
Then Proposition~\ref{prop:xxx} gives $f = f' {+} f''$ where
$f'$ is polyblind and $f''$ is computable by a $(k{-}1)$-marble bimachine.
By Proposition~\ref{prop:blind-confuse}, $f'$ is $(k{-1})$-repetitive.
Since $f$ is also $(k{-1})$-repetitive, then $f'' = f{-}f'$ is also $(k{-}1)$-repetitive
by Claim~\ref{claim:soustraction}. Thus by induction hypothesis
$f''$ is polyblind.
For $\ref{it:char:syntax} \implies~\ref{it:char:poly}$,
since in Proposition~\ref{prop:xxx} we have $f = f'{+}f''$ where $f'$ and $f''$
are polyblind, then $f$ is (effectively) polyblind.
The decidability is also obtained by induction: one has to
check that $\mach$ is $2^{\hei}$-permutable (which is decidable)
and inductively that $f''$ is polyblind.
\end{proof}

Now we present the tools used in this proof.
We first show that polyblind functions are repetitive.
Then we introduce the notion of \emph{permutability},
and show that  repetitive functions are computed by
permutable marble  bimachines.

\subsection{Polyblind functions are repetitive}

Intuitively, a blind transducer cannot distinguish between two
``similar'' iterations of a given factor in a word,
since it cannot drop a pebble for doing so. We get the following
using technical but conceptually easy pumping arguments.

\begin{proposition} \label{prop:blind-confuse}
A polyblind function is $k$-repetitive for all $k \ge 1$.
\end{proposition}

\begin{proof}[Proof idea.] We first show that a regular function
(computed by a bimachine without external functions) is
$k$-repetitive for all $k \ge 1$. In this case, $\eta$ is chosen
as the idempotent index of the monoid used by the bimachine.
Then, it is easy to conclude by noting that $k$-repetitiveness
is preserved under sums and Hadamard products.
\end{proof}
The induction which proves Theorem~\ref{theo:car}
also requires the following result. Its proof is obvious
since $k$-repetitiveness is clearly preserved under subtractions.

\begin{claim} \label{claim:soustraction} If $f,g$ are $k$-repetitive
and $f{-}g \ge 0$, then it is $k$-repetitive.
\end{claim}

\subsection{Repetitive functions are computed by permutable machines}

\label{ssec:prod-marbles}

In this subsection, we show that $k$-marble bimachines
which compute $k$-repetitive functions have a specific property
named \emph{permutability} (which turns out to be decidable).

\subparagraph*{Productions.}
We first introduce the notion of \emph{production}.
In the rest of this paper, the notation $\multi{\cdots}$ represents
a multiset (i.e. a set with multiplicities).
If $S$ is a set and $m$ is a multiset, we write $m \subseteq S$
to say that each element of $m$ belongs to $S$
(however, there can be multiplicities in $m$
but not in $S$). For instance $\multi{1,1,2,3,3} \subseteq \Nat$.

\label{page:prod}

\begin{definition}[Production]
Let $\mach = (A, M, \mu, \oras, \lambda)$ be a $k$-marble bimachine,
$w \in A^*$.
We define the \emph{production} of
$\mach$ on $\multi{i_1,  \dots, i_k}\subseteq [1{:}|w|]$ as follows if $i_1 \le \cdots \le i_k$:
\begin{itemize}
\item if $k=1$, $\rod{\mach}{w}{\multi{i_1}} \defined \lambda(\mu(w[1{:}i_1{-}1]),w[i_1],\mu(w[i_1{+}1{:}|w|])) \in \Nat$;
\item if $k \ge 2$, let $\exte \defined \lambda(\mu(w[1{:}i_k{-}1]),w[i_k],\mu(w[i_k{+}1{:}|w|])) \in \oras$
and $\mach_{\exte}$ be the $(k{-}1)$-marble bimachine computing $\exte$. Then
$
\rod{\mach}{w}{\multi{i_1, \cdots, i_k}} \defined \rod{\mach_{\exte}}{w[1{:}i_k]}{\multi{i_1, \dots, i_{k{-}1}}}
$.
\end{itemize}
\vspace{0.2cm}
\end{definition}
The value $\rod{\mach}{w}{\multi{i_1, \cdots, i_k}}$ with $i_1 \le \cdots \le i_k$ thus corresponds to
the value output when doing a call on position $i_k$, then on $i_{k-1}$, etc.
Now let $\multi{I_1, \dots, I_k}$ be a multiset of sets of positions of $w$ (i.e. for all
$1 \le i \le k$ we have $I_i \subseteq [1{:}|w|]$).
We define the production of $\mach$ on
$\multi{I_1, \dots, I_k}$ as the combination of all possible productions
on positions among these sets:
\begin{equation*}
\rod{\mach}{w}{\multi{I_1, \dots, I_k}} \defined
\sum_{\substack{\multi{i_1, \dots, i_k} \subseteq [1{:}|w|] \\ \text{ with }  i_j \in I_j \text{ for } 1 \le j \le k}}
\rod{\mach}{w}{\multi{i_1, \dots, i_k}}.
\end{equation*}
\begin{remark}  Note that we no longer have $i_1 \le \cdots \le i_k$.
\end{remark}
By rewriting the sum which defines the function $f$ from $\mach$, we get the following.

\begin{lemma} \label{lem:coupe-prod}
Let $\mach$ be a $k$-marble bimachine computing a function
$f: A^* \rightarrow \Nat$. Let $w \in A^*$ and $J_1, \dots J_n$ be a partition of $[1{:}|w|]$:
\begin{equation*}
f(w) = \sum_{\multi{I_1, \dots, I_k} \subseteq \{J_1, \dots, J_n\}} \rod{\mach}{w}{\multi{I_1, \dots, I_k} }.
\end{equation*}
\end{lemma}

Following the definition of bimachines,
the production $\rod{\mach}{w}{\multi{i_1, \cdots, i_k}}$
should only depend on $w[i_1], \dots, w[i_k]$ and of
the image under $\mu$ of the factors between these positions.
Now we formalize this intuition in a more general setting.

\begin{definition}[Multicontext] Given $x \ge 0$,
an $x$-\emph{multicontext} consists of two sequences of words
$v_0, \cdots, v_x \in A^*$ and $u_1, \cdots, u_x \in A^*$. It is denoted
$v_0 \cro{u_1} v_1 \cdots \cro{u_x} v_x$.
\end{definition}
Let $w \defined v_0 u_1 \cdots u_k v_k \in A^*$.
For $1 \le i \le k$, let $I_i \subseteq [1{:}|w|]$
be the set of positions corresponding to $u_i$.
We define the production on the multicontext
$v_0 \cro{u_1} v_1 \cdots \cro{u_k} v_k$ by:
\begin{equation}
\label{eq:facteurs}
\rod{\mach}{}{v_0 \cro{u_1} v_1 \cdots \cro{u_k} v_k} \defined
\rod{\mach}{w}{\multi{I_1, \dots, I_k}}.
\end{equation}
 As stated in Proposition-Definition~\ref{depro:factor-k},
 this quantity only depends on the $u_i$ and the images
of the $v_i$ under the morphism $\mu$ of $\mach$,
which leads to a new notion of productions.

\begin{depro} \label{depro:factor-k}
Let $\mach = (A,M, \mu, \oras, \lambda)$ be a $k$-marble bimachine.
Let  $v_0 u_1 \cdots u_k v_k \in A^*$ and
$v'_0 u_1 \cdots u_k v'_k \in A^*$
be such that $\mu(v_i) = \mu(v'_i)$ for all $0 \le i \le k$.
Then:
$
\rod{\mach}{}{v_0 \cro{u_1} v_1 \cdots \cro{u_k} v_k} =
\rod{\mach}{}{v'_0 \cro{u_1}v'_1 \cdots \cro{u_k} v'_k}
$.
Let $m_i \defined \mu(v_i) = \mu(v'_i)$, we define
$\rod{\mach}{}{m_0 \cro{u_1} m_1 \cdots \cro{u_k} m_k}$
as the previous value.
\end{depro}

\begin{remark} In the following, we shall directly manipulate
multicontexts of the form $m_0 \cro{u_1} m_1 \cdots \cro{u_x} m_x$
with $m_i \in M$ and $u_i \in A^*$. Note that an $x$-multicontext
and a $y$-multicontext can be concatenated to obtain
an $(x+y)$-multicontext.
\end{remark}
We finally introduce the notion of $(x,K)$-\emph{iterator},
which corresponds to an $x$-multicontext in which the words $u_i$
have length at most $K \ge 0$ and have idempotent images
under $\mu$. Recall that $e \in M$ is \emph{idempotent} if and only if $ee=e$.

\begin{definition}[Iterator] 
Let $x, K \ge 0$ and
$\mu: A^* \rightarrow M$.
An \emph{$(x,K)$-iterator} for $\mu$
is an $x$-multicontext of the form
$
m_0  \left( \prod_{i=1}^{x} e_i \cro{u_i} e_i m_i \right)
$
such that $m_0, \dots, m_{x} \in M$ and
\mbox{$u_1, \dots, u_{x} \in A^*$} are such that
for all $1 \le i \le x$, $|u_i| \le K$
and $e_i = \mu(u_i) \in M$ is idempotent.
\end{definition}

\subparagraph*{Permutable $k$-marble bimachines.}
We are now ready to state the definition of \emph{permutability}
for marble bimachines. Intuitively, this property means that, under some idempotency conditions,
$\rod{\mach}{}{m_0 \cro{u_1} m_1 \cdots \cro{u_k} m_k}$
only depends on the $1$-multicontexts of each $\cro{u_i}$, which
are $m_0 \mu(u_1) \cdots m_i \cro{u_i}
m_{i+1} \mu(u_{i+1}) \cdots m_k \in M$ for $1 \le i \le k$.
In particular, it does not depend on the relative position of the $u_i$ nor on the $m_i$
which separate them. Hence, it will be possible to simulate a
permutable $k$-marble bimachine by a $k$-blind bimachine
which can only see the $1$-multicontext of one position
at each time.

\begin{definition} \label{def:factor-perm} Let $\mach$ be a $k$-marble bimachine
using a surjective morphism $\mu: A^* \rightarrow M$. Let $K \ge 0$,
we say that $\mach$ is $K$-\emph{permutable} if the following holds
whenever $\ell+x+r= k$:
\begin{itemize}
\item for all $(\ell,K)$-iterator $\ce$  and $(r,K)$-iterator $\de$;
\item for all $(x,K)$-iterator
$
 m_0 \left(\prod_{i=1}^x e_i \cro{u_i} e_i m_i\right)
$
such that $ e \defined m_0 \left(\prod_{i=1}^x e_i m_i\right)$
idempotent;
\item for all $1 \le j \le x$, define the left and right contexts:
\begin{equation*}
        \lefe_j \defined e \left(\prod_{i=1}^{j} m_{i{-1}}  e_i\right) 
        \tnorm{~~and~~}
        \rige_j \defined \left(\prod_{i=j}^{x} e_i m_i\right) e.
\end{equation*}
\end{itemize}
Then if $\sigma$ is a permutation of  $[1{:}x]$, we have:
\begin{equation*}
\begin{aligned}
&\rod{\mach}{}{\ce~e m_0 \left(\prod_{i=1}^x e_i \cro{u_i} e_i m_i\right) e~\de}
&= \rod{\mach}{}{\ce~\left(\prod_{i=1}^x \lefe_{\sigma(i)}
\cro{u_{\sigma(i)}} \rige_{\sigma(i)} \right)~\de}. 
\end{aligned}
\end{equation*}
\end{definition}
An visual description of permutability is depicted in Figure~\ref{fig:ecarte}.

\begin{figure}[h!]

     \begin{subfigure}[b]{1\textwidth}
     \centering
     \begin{tikzpicture}{scale=1}
	\node[above] at (0,-0.25) {$\ce$};
	\node[above] at (11,-0.25) {$\de$};	
	
	\draw[-,very thick,black](0.5,0) -- (10.5,0);
	
	\draw[-,very thick,black](0.5,0) -- (0.5,0.15);
	\draw[-,very thick,black](1,0) -- (1,0.15);
	\draw[-,very thick,black](2,0) -- (2,0.15);

	\draw[-,very thick,black](3,0) -- (3,0.15);
	\draw[-,very thick,black](5,0) -- (5,0.15);

	\draw[-,very thick,black](6,0) -- (6,0.15);
	\draw[-,very thick,black](8,0) -- (8,0.15);

	\draw[-,very thick,black](9,0) -- (9,0.15);
	\draw[-,very thick,black](10,0) -- (10,0.15);
	\draw[-,very thick,black](10.5,0) -- (10.5,0.15);

	\draw[very thick,\myBlue, fill = \myBlue] (2,0) rectangle (3,-0.1);
	\draw[very thick,\myBlue, fill = \myBlue] (5,0) rectangle (6,-0.1);	
	\draw[very thick,\myBlue, fill = \myBlue] (8,0) rectangle (9,-0.1);
	
	\node[above] at (0.75,0) {\small $e$};	
	\node[above] at (1.5,0) {\small $m_0e_1$};		
	\node[above] at (2.5,0) {\small $e_1$};	
	\node[below] at (2.5,-0.1) {\small $\cro{u_1}$};	
	\node[above] at (4,0) {\small $e_1m_1e_2$};	
	\node[above] at (5.5,0) {\small $e_2$};	
	\node[below] at (5.5,-0.1) {\small $\cro{u_2}$};	
	\node[above] at (7,0) {\small $e_2m_2e_3$};
	\node[above] at (8.5,0) {\small $e_3$};	
	\node[below] at (8.5,-0.1) {\small $\cro{u_3}$};	
	\node[above] at (9.5,0) {\small $e_3m_3$};
	\node[above] at (10.25,0) {\small $e$};
	
	\draw[thick, dashed,purple] (1,0.75)--(10,0.75);
	\draw[thick,purple] (1,0.75)--(1,0.5);
	\draw[thick,purple] (10,0.75)--(10,0.5);
	\node[above,purple] at (5.5,0.75) {\small $=e$};

	\draw[thick, dashed,purple] (0.5,-0.75)--(2,-0.75);
	\draw[thick,purple] (0.5,-0.75)--(0.5,-0.5);
	\draw[thick,purple] (2,-0.75)--(2,-0.5);
	\node[below,purple] at (1.25,-0.7) {\small $=: \lefe_1$};

	\draw[thick, dashed,purple] (3,-0.75)--(10.5,-0.75);
	\draw[thick,purple] (3,-0.75)--(3,-0.5);
	\draw[thick,purple] (10.5,-0.75)--(10.5,-0.5);
	\node[below,purple] at (3.75,-0.7) {\small $=: \rige_1$};
	
	\draw[thick, dashed,purple] (0.5,-1.25)--(5,-1.25);
	\draw[thick,purple] (0.5,-1.25)--(0.5,-1);
	\draw[thick,purple] (5,-1.25)--(5,-1);
	\node[below,purple] at (4,-1.2) {\small $=: \lefe_2$};

	\draw[thick, dashed,purple] (6,-1.25)--(10.5,-1.25);
	\draw[thick,purple] (6,-1.25)--(6,-1);
	\draw[thick,purple] (10.5,-1.25)--(10.5,-1);
	\node[below,purple] at (7,-1.2) {\small $=: \rige_2$};
	
	\draw[thick, dashed,purple] (0.5,-1.75)--(8,-1.75);
	\draw[thick,purple] (0.5,-1.75)--(0.5,-1.5);
	\draw[thick,purple] (8,-1.75)--(8,-1.5);
	\node[below,purple] at (7.25,-1.7) {\small $=: \lefe_3$};

	\draw[thick, dashed,purple] (9,-1.75)--(10.5,-1.75);
	\draw[thick,purple] (9,-1.75)--(9,-1.5);
	\draw[thick,purple] (10.5,-1.75)--(10.5,-1.5);
	\node[below,purple] at (9.75,-1.7) {\small $=: \rige_3$};
	
\end{tikzpicture}
\caption{\label{sfig:ec1} Initial multicontext and definition of the $\lefe_j$ / $\rige_j$ for $1 \le j \le 3$}
     \end{subfigure}
     
     \begin{subfigure}[b]{1\textwidth}
     \centering
     \begin{tikzpicture}{scale=1}
	\node[above] at (1,-0.25) {$\ce$};
	\node[above] at (12,-0.25) {$\de$};	
	
	\draw[-,very thick,black](1.5,0) -- (4.5,0);

	\draw[-,very thick,black](1.5,0) -- (1.5,0.15);
	\draw[-,very thick,black](2.5,0) -- (2.5,0.15);
	\draw[-,very thick,black](3.5,0) -- (3.5,0.15);
	\draw[-,very thick,black](4.5,0) -- (4.5,0.15);
	
	\draw[-,very thick,black](5,0) -- (8,0);
	
	\draw[-,very thick,black](5,0) -- (5,0.15);
	\draw[-,very thick,black](6,0) -- (6,0.15);
	\draw[-,very thick,black](7,0) -- (7,0.15);
	\draw[-,very thick,black](8,0) -- (8,0.15);
	
	\draw[-,very thick,black](8.5,0) -- (11.5,0);

	\draw[-,very thick,black](8.5,0) -- (8.5,0.15);
	\draw[-,very thick,black](9.5,0) -- (9.5,0.15);
	\draw[-,very thick,black](10.5,0) -- (10.5,0.15);
	\draw[-,very thick,black](11.5,0) -- (11.5,0.15);

	\draw[very thick,\myBlue, fill = \myBlue] (2.5,0) rectangle (3.5,-0.1);
	\draw[very thick,\myBlue, fill = \myBlue] (6,0) rectangle (7,-0.1);	
	\draw[very thick,\myBlue, fill = \myBlue] (9.5,0) rectangle (10.5,-0.1);

	\node[above] at (2,0) {\small $\lefe_3$};	
	\node[above] at (3,0) {\small $e_3$};	
	\node[below] at (3,-0.1) {\small $\cro{u_3}$};	
	\node[above] at (4,0) {\small $\rige_3$};	

	\node[above] at (5.5,0) {\small $\lefe_1$};	
	\node[below] at (6.5,-0.1) {\small $\cro{u_1}$};
	\node[above] at (6.5,0) {\small $e_1$};
	\node[above] at (7.5,0) {\small $\rige_1$};

	\node[above] at (9,0) {\small $\lefe_2$};		
	\node[above] at (10,0) {\small $e_2$};	
	\node[below] at (10,-0.1) {\small $\cro{u_2}$};	
	\node[above] at (11,0) {\small $\rige_2$};	
	
\end{tikzpicture}
         \caption{Multicontext which separates the factors using the $\lefe_j$ / $\rige_j$ and $\sigma$}
     \end{subfigure}

\caption{\label{fig:ecarte} Productions which must be equal in Definition \ref{def:factor-perm},
with $x=3$ and  $\sigma = (3,1,2)$}
\end{figure}
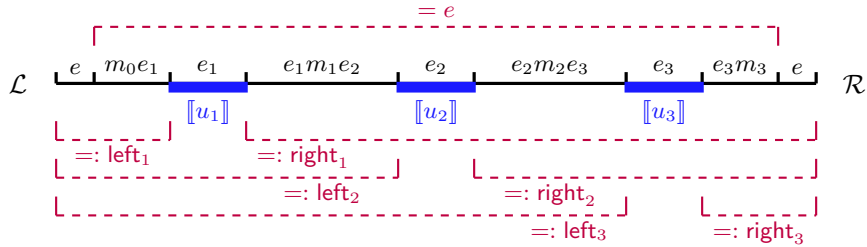
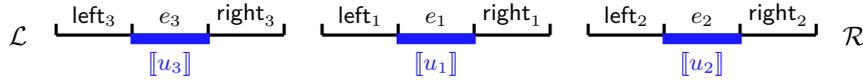%
The next result follows from a technical proof
based on iteration techniques (it can be understood as
a kind of pumping lemma on iterators).
\begin{proposition} \label{prop:factor-perm} Let $\mach$ be a
$k$-marble bimachine computing
a $k$-repetitive function. Then $\mach$ is $K$-permutable
for all $K \ge 0$.
\end{proposition}
Let us finally note that being $K$-permutable (for a fixed $K$)
is a decidable property. Indeed,  it suffices to range over all $\ell + x + r = k$
and $(\ell,K)$-, $(x,K)$- and $(r,K)$-iterators (there are finitely many
of them, since they correspond to bounded sequences which alternate
between monoid elements and words of bounded lengths),
and compute their productions.

\section{From permutable bimachines to polyblind functions}

\label{sec:tech2}

The purpose of this section is to show Proposition~\ref{prop:xxx},
which allows us to perform the induction step in the proof of Theorem~\ref{theo:car},
by going from $k$ to $k{-}1$ marbles.

\begin{proposition} \label{prop:xxx}
Let $\mach =(A,M, \mu, \oras, \lambda)$ be a $2^{\hei}$-permutable $k$-marble bimachine
computing a function $f: A^* \rightarrow \Nat$. One
can build a polyblind function $f': A^* \rightarrow \Nat$
and a function  $f'' A^* \rightarrow \Nat$ computed by a $(k{-}1)$-marble bimachine
such that $f = f'+ f''$.
\end{proposition}

\begin{proof}[Proof overview.] It follows from Equation~\ref{eq:i-d} and Lemma~\ref{lem:sumi} that:
\begin{equation*}
\begin{aligned}
f &= (\sumd{\mach} {+} \sumi{\mach})\circ \fact{\mu} \substack{\tnorm{~~(see Theorem~\ref{theo:simon}
for the definition of } \fact{\mu})}\\
&=  \underbrace{\sumi{\mach}'  \circ \fact{\mu}}_{ =: f'}
+ \underbrace{(\sumd{\mach} {+} \sumi{\mach}'') \circ \fact{\mu}}_{ =: f''} \\
\end{aligned}
\end{equation*}
By Lemma~\ref{lem:sumi} one has $f'' \ge 0$. Furthermore, $\sumi{\mach}'$
is polyblind, hence  $f'$ is polyblind since
the class of polyblind functions is closed
under pre-composition by regular functions (even when the outputs
are not unary, see \cite{nous2020comparison}).
Similarly, it follows from lemmas~\ref{lem:sumd} and~\ref{lem:sumi}
that $\sumd{\mach}$ and $\sumi{\mach}''$ are polyregular
with growth at most $k{-}1$ (see Definition~\ref{def:growth}),
hence so is their sum and its pre-composition
by a regular function~\cite{bojanczyk2018polyregular}.
Thus $f''$ is polyregular and has growth at most $k{-}1$.
By theorems~\ref{theo:equiv} and~\ref{theo:minimize},
it can be computed by a  $(k{-}1)$-marble bimachine.
\end{proof}
Now we give the statements and 
the proofs of lemmas~\ref{lem:sumd} and~\ref{lem:sumi}.
An essential tool is the notion of  \emph{factorization forest},
which is recalled below.

\subsection{Factorization forests}

If $\mu:A^* \rightarrow M$ is a morphism into a finite monoid
and $w \in A^*$, a \emph{factorization forest} (called \emph{$\mu$-forest} in the following)
of $w$ is an unranked tree structure defined as follows.

\begin{definition}[Factorization forest \cite{simon1990factorization}] \label{def:facto}
Given a morphism $\mu: A^* \rightarrow M$ and $w\in A^*$, we say that
$\forest$ is a  $\mu$-\emph{forest} of $w$ if:
\item
\begin{itemize}
\item either $\forest = a$ and  $w = a  \in A $;
\item or $\forest = \tree{\forest_1} \cdots \tree{\forest_n}$,
$w = w_1 \cdots w_n$ and for all $1 \le i \le n$,
$\forest_i$ is a $\mu$-forest of $w_i \in A^+$.
Furthermore, if $n \ge 3$ then
$\mu(w_1) = \dots = \mu(w_n)$ is an idempotent of $M$.
\end{itemize}
\end{definition}

\begin{remark} If $\tree{\forest_1} \cdots \tree{\forest_n}$ is a $\mu$-forest,
then so is $\tree{\forest_i} \tree{\forest_{i+1}} \dots \tree{\forest_{j}}$
for $1 \le i \le j \le n$. The empty forest $\vide$ is the unique $\mu$-forest
of the empty word $\vide$.
\end{remark}
We shall use the standard tree vocabulary of ``height'' (a leaf is a tree of
height $1$), ``child'', ``sibling'', ``descendant'' and ``ancestor''
(both defined in a non-strict way: a node is itself
one of its ancestors/descendants), etc.
We denote by $\Nodes{\forest}$ the set of nodes of $\forest$.
In order to simplify the statements, we identify a node $\nod \in \Nodes{\forest}$ with
the subtree rooted in this node. Thus $\Nodes{\forest}$ can also
be seen as the set of subtrees of $\forest$, and $\forest \in \Nodes{\forest}$.
We say that a node is \emph{idempotent} if is has at least
$3$ children (see Definition~\ref{def:facto}).

Given $\mu: A^* \rightarrow M$, we
denote by $\Facts{\mu}{}{w}$ the set of
$\mu$-forests of $w \in A^*$.
If $K \ge 0$, let $\Facts{\mu}{K}{w}$ be the $\mu$-forests
of $w$ of height at most $K$.
Note that $\Facts{\mu}{}{w}$ is a set of tree structures
 which can also be seen as a language over $\apar{A} \defined A \uplus \{\lefttree,\righttree\}$.
Indeed, a forest of $w$ can also be seen as
 ``the word $w$ with brackets'' in Definition~\ref{def:facto}.

\begin{figure}[h!]
	\newcommand{\couleur}{blue!70}
\centering
\begin{tikzpicture}{scale=1}

	\newcommand{\texte}{\small \bfseries \sffamily \mathversion{bold} }

	\node[above] at  (0,0)  {$a$};
	\node[above] at  (1,0)  {$a$};	
	\node[above] at  (2,0)  {$c$};	
	\node[above] at  (3,0)  {$a$};
	\node[above] at  (4,0)  {$c$};
	\node[above] at  (5,0)  {$b$};
	\node[above] at  (6,0)  {$b$};
	\node[above] at  (7,0)  {$c$};
	\node[above] at  (8,0)  {$b$};
	\node[above] at  (9,0)  {$b$};
	\node[above] at  (10,0)  {$b$};
	\node[above] at  (11,0)  {$c$};

	\draw[double] (4,0.975) -- (8,0.975);
	\draw[] (3,1.5) -- (6,1.5);
	\draw[double] (2,1.975) -- (11,1.975);
	\draw[] (0,2) -- (1,2);
	\draw[] (0.5,2.5) -- (6.5,2.5);
	
	\draw[] (0,2) -- (0,0.5);	
	\draw[] (1,2) -- (1,0.5);	
	\draw[] (0.5,2) -- (0.5,2.5);	
	
	\draw[] (2,2) -- (2,0.5);
	\draw[] (9,2) -- (9,0.5);
	\draw[] (10,2) -- (10,0.5);	
	\draw[] (11,2) -- (11,0.5);	

	\draw[] (6.5,2) -- (6.5,2.5);	
	\draw[] (3.5,2.5) -- (3.5,2.75);	

	\draw[] (3,1.5) -- (3,0.5);	
	\draw[] (4,1) -- (4,0.5);	
	\draw[] (5,1) -- (5,0.5);	
	\draw[] (6,1.5) -- (6,0.5);	
	\draw[] (7,1) -- (7,0.5);		
	\draw[] (8,1) -- (8,0.5);	

	\draw[] (4.5,2) -- (4.5,1.5);
	
	\fill[fill=\couleur]  (4.5,1.5)  circle (0.12);
	\fill[fill=\couleur]  (3,0.5)  circle (0.12);
	\fill[fill=\couleur]  (4,0.5)  circle (0.12);
	\fill[fill=\couleur]  (8,0.5)  circle (0.12);
	\fill[fill=\couleur]  (6,1)  circle (0.12);
		
\end{tikzpicture}

\caption{\label{fig:iterable}  The $\mu$-forest $\tree{aa}\tree{c\tree{a\tree{cbbcb}}bbc}$ on $aacacbbcbbbc$}

\end{figure}

\begin{example} \label{ex:factorization} Let $A = \{a,b,c\}$,
$M {=} (\{\neutral,\deutral,0\}, \times)$ with
$\mu(a) \defined \deutral$
and $\mu(b) \defined \mu(c) \defined  0$.
Then
$\forest \defined\tree{aa}\tree{c\tree{a\tree{cbbcb}}bbc} \in \Facts{\mu}{5}{aacacbbcbbbc}$
(we dropped the brackets around single letters for more readability)
is depicted in Figure~\ref{fig:iterable}. Double lines
are used to denote idempotent nodes
(i.e. with more than $3$ children).
\end{example}

A celebrated result states that for any word, a forest of  bounded height
always exist and it can be computed by a bimachine
(with non-unary output alphabet), or a two-way transducer.
The following theorem can also be found in~\cite[Lemma~6.5]{bojanczyk2018polyregular}.

\begin{theorem}[\cite{simon1990factorization}] \label{theo:simon}
Given a morphism $\mu: A^* \rightarrow M$, we have 
$\Facts{\mu}{\hei}{w} \neq \varnothing$ for all $w \in A^*$.
Furthermore, one can build a two-way transducer (with a non-unary output alphabet)
which computes a total
function $\fact{\mu}: A^* \rightarrow (\apar{A})^*,w \mapsto \forest \in \Facts{\mu}{\hei}{w}$.
\end{theorem}

\subsection{Iterable nodes and productions}

We define the iterable nodes $\itera{\forest} \subseteq \Nodes{\forest}$
as the set of nodes which have both a left
and a right sibling. Such nodes are thus exactly the middle children
of idempotent nodes.

\begin{definition} Let $\forest \in \Facts{\mu}{}{w}$, we define inductively
the \emph{iterable nodes} of $\forest$:
\begin{itemize}
\item if $\forest = a\in A$ is a leaf, $\itera{\forest} \defined \varnothing$;
\item otherwise if $\forest = \tree{\forest_1} \cdots \tree{\forest_n}$, then:
\begin{equation*}
\itera{\forest} \defined \{\forest_i: 2 \le i \le n{-}1\} \cup  \bigcup_{1 \le i \le n} \itera{\forest_i}.
\end{equation*}
\end{itemize}
\end{definition}
Now we define a notion of \emph{skeleton} which selects
the right-most and left-most children.

\begin{definition} Let $\forest \in \Facts{\mu}{}{w}$ and $\nod \in \Nodes{\forest}$, we define
the \emph{skeleton} of $\nod$ by:
\item
\begin{itemize}
\item if $\nod = a\in A$ is a leaf, then $\Ske{\nod} \defined \{ \nod \}$;
\item otherwise if $\nod = \tree{\forest_1} \cdots \tree{\forest_n}$, then
$\Ske{\nod} \defined \{\nod\} \cup \Ske{\forest_1} \cup \Ske{\forest_n}$.
\vspace{0.2cm}
\end{itemize}
\end{definition}
Intuitively, $\Ske{\nod} \subseteq \Nodes{\forest}$ contains all the descendants
of $\nod$ except those which are descendant of a middle child.
We then define the frontier of $\nod$, denoted
$\fr{\nod}{\forest} \subseteq \{1, \dots, |w|\}$
as the set of positions of $w$ which belong to $\Ske{\nod}$
(when seen as leaves of $\forest$).

\begin{example} \label{ex:iterable} In Figure~\ref{fig:iterable},
the top-most blue node $\nod$ is iterable.
Furthermore $\Ske{\nod}$ is the set of blue nodes,
$\fr{\nod}{\forest} = \{4,5,9\}$
and $w[\fr{\nod}{\forest} ] = acb$.
\end{example}
Using the frontiers, we can naturally lift the notion of productions of a $k$-marble
bimachine from multisets of sets of positions to multisets of nodes in a forest.

\begin{definition} Let $\mach = (A,M, \mu, \oras, \lambda)$ be a $k$-marble bimachine,
$w \in A^*$, $\forest \in \Facts{\mu}{}{w}$
 and $\nod_1, \dots, \nod_k \in \Nodes{\forest}$.
 We let 
 $
 \rod{\mach}{\forest}{\multi{\nod_1,\dots, \nod_k}}
 \defined \rod{\mach}{w}{\multi{\fr{\nod_1}{\forest}, \dots, \fr{\nod_k}{\forest}}}
$.
\end{definition}
Using Lemma~\ref{lem:coupe-prod}, we can recover
the function $f$ from the productions over 
the nodes. Lemma~\ref{lem:coupe-iterables} roughly
ranges over all possible tuples of calling positions
of $\mach$.

\begin{lemma} \label{lem:coupe-iterables} Let $f: A^* \rightarrow \Nat$ be computed
by a $k$-marble bimachine $\mach = (A,M, \mu, \oras, \lambda)$. 
If $w \in A^*$ and $\forest \in \Facts{\mu}{}{w}$, we have:
\begin{equation*}
f(w) = \sum_{ \multi{\nod_1, \dots, \nod_k}\subseteq \itera{\forest} \cup \{\forest\}}
\rod{\trans}{\forest}{\multi{\nod_1, \dots, \nod_k}}.
\end{equation*}
\end{lemma}

\begin{proof}
It follows from \cite{doueneau2021pebble} that for all
$w \in A^*$ and  $\forest \in \Facts{\mu}{}{w}$, the set of frontiers 
$ \{\fr{\nod}{\forest}: \nod \in \itera{\forest} \cup \{\forest\}\}$
is a partition of $[1{:} |w|]$. We then apply Lemma~\ref{lem:coupe-prod}.
\end{proof}

\subsection{Dependent multisets of nodes}

The multisets $\multi{\nod_1, \dots, \nod_k}\subseteq \itera{\forest} \cup \{\forest\}$ of 
Lemma~\ref{lem:coupe-iterables} will be put into two categories. The \emph{independent}
multisets are those whose nodes are distinct and ``far enough'' in $\forest$.
The remaining ones are said \emph{dependent};
their number is bounded by a polynomial of degree $k{-}1$.

\begin{definition}[Independent multiset] Let $\mu: A^* \rightarrow M$, $w \in A^*$ and
$\forest \in \Facts{\mu}{}{w}$, we say that a
multiset $\Nod \defined \multi{\nod_1, \dots, \nod_k} \subseteq \itera{\forest}$ is \emph{independent}
if for all $1 \le i \neq j \le k$:
\begin{itemize}
\item $\nod_i$ is not an ancestor of $\nod_j$;
\item $\nod_i$ is not the immediate left sibling of an ancestor of $\nod_j$;
\item $\nod_i$ is not the immediate right sibling of an ancestor of $\nod_j$.
\end{itemize}
\vspace{0.2cm}
\end{definition}
Note that if $\Nod$ is independent, then $\forest \not \in \Nod$ since it
is not an iterable node. We denote by $\Indeps{\forest}$ the set of independent multisets.
Conversely, let $\Deps{\forest}$ be the set of
multisets $\multi{\nod_1, \dots, \nod_k} \subseteq \itera{\forest}\cup \{\forest\}$
which are \emph{dependent} (i.e. not independent).
By Lemma~\ref{lem:coupe-iterables}, if $\mach = (A,M, \mu, \oras, \lambda)$
computes $f: A^* \rightarrow \Nat$ and $\forest \in \Facts{\mu}{}{w}$, then:
\begin{equation}
\label{eq:i-d}
\begin{aligned}
f(w) =& \sum_{ \Nod \in \Indeps{\forest}}
\rod{\mach}{\forest}{\Nod}
+ \sum_{ \Nod \in \Deps{\forest}}
\rod{\mach}{\forest}{\Nod}
\end{aligned}
\end{equation}
The idea is now to compute these two sums 
separately. We begin with the second one.

\begin{lemma} \label{lem:sumd} Given a $k$-marble bimachine
$\mach = (A,M, \mu, \oras, \lambda)$, the following function
is (effectively) polyregular and has growth at most $k{-}1$:
\begin{equation*}
\sumd{\mach}: (\apar{A})^* \rightarrow \Nat, \forest \mapsto
\left\{
    \begin{array}{l}
        \displaystyle \sum_{\Nod \in \Deps{\forest}}
        \rod{\mach}{\forest}{\Nod}
	\substack{\tnorm{{}~if } \forest \in \Facts{\mu}{\hei}{w} \tnorm{{} for some } w \in A^*;} \\
        0 \substack{\tnorm{{} otherwise.}}\\
    \end{array}
\right.
\end{equation*}
\end{lemma}

\begin{remark}
For this result, we do not need to assume that $\mach$ is permutable.
\end{remark}

\subsection{Independent multisets of nodes}

In order to complete the description of $f$ from
Equation~\ref{eq:i-d}, it remains to treat the productions
over independent multisets of nodes.
When \mbox{$\multi{\nod_1, \dots, \nod_k} \in \Indeps{\forest}$}, all the $\nod_i$
must be distinct, hence we shall denote it by a set $\set{\nod_1, \dots, \nod_k}$.
We define the counterpart of $\sumd{\mach}$:
\begin{equation*}
\sumi{\mach}: (\apar{A})^* \rightarrow \Nat, \forest \mapsto
\left\{
    \begin{array}{l}
        \displaystyle \sum_{\Nod \in \Indeps{\forest}}
        \rod{\mach}{\forest}{\Nod}
	\substack{\tnorm{{}~if } \forest \in \Facts{\mu}{\hei}{w} \tnorm{{} for some } w \in A^*;} \\
        0 \substack{\tnorm{{} otherwise.}}\\
    \end{array}
\right.
\end{equation*}

\begin{lemma} \label{lem:sumi} Given a $k$-marble bimachine $\mach$
which is $2^{\hei}$-permutable,
one can build a polyblind function $\sumi{\mach}': (\apar{A})^* \rightarrow \Nat$
and a polyregular function $\sumi{\mach}'': (\apar{A})^* \rightarrow \Nat$
with growth at most $k{-}1$, such that
$\sumi{\mach} = \sumi{\mach}' + \sumi{\mach}''$.
\end{lemma}
If $\sumi{\mach}$ was a polynomial, then
$\sumi{\mach}'$ should roughly be its term of highest degree
and $\sumi{\mach}''$ corresponds to
 a corrective term.

The rest of this section is devoted
to the proof of Lemma~\ref{lem:sumi}.
In order to simplify the notations, we extend a morphism
$\mu$ to its $\mu$-forests by $\mu(\forest) \defined \mu(w)$
when $\forest \in \Facts{\mu}{}{w}$.
Given $\nod \in \Nodes{\forest}$, we denote by $\depth{\forest}{\nod}$ its
depth in the tree structure $\forest$ (the root has depth $1$, and it is defined
inductively as usual). Now we introduce the notion of \emph{linearization}
of $\nod \in \Nodes{\forest}$, which is used to describe
$w[\fr{\nod}{\forest}]$ as a $1$-multicontext.

\begin{definition}[Linearization] \label{def:lili-node}
Let $\mu: A^* \rightarrow M$,
$w \in A^*$ and $\forest \in \Facts{\mu}{}{w}$.
The \emph{linearization} of $\nod \in \Nodes{\forest}$
is a $1$-multicontext $m\cro{u}m'$ defined by induction:
\item
\begin{itemize}
\item if $\nod = \forest$ then
$\lin{\forest}{\nod} \defined \cro{w[\fr{\forest}{\forest}]}$;
\item otherwise $\forest = \tree{\forest_1} \cdots \tree{\forest_n}$,
and $\nod \in \Nodes{\forest_i}$ for some $1 \le i \le n$. We define:\\
$
\lin{\forest}{\nod} \defined \mu(\forest_1)\cdots \mu(\forest_{i-1})
\lin{\forest_i}{\nod} \mu(\forest_{i+1})\cdots \mu(\forest_{n}).
$
\end{itemize}
\vspace*{0.2cm}
\end{definition}
We finally introduce the notion of \emph{architecture}. Intuitively, it is
a simple tree which describes the positions of a set of nodes $\Nod \in \Indeps{\forest}$
in its forest $\forest$. We build it inductively on the example depicted in Figure~\ref{fig:archi}.
At the root, we see that there is no node of $\Nod$ in the left subtree, hence we replace
it by its image under $\mu$. The right subtree is an idempotent node
whose leftmost and rightmost subtrees have no node in $\Nod$.
We thus replace this idempotent node by a leaf containing the multiset of the 
linearizations and depths of the $\nod \in \Nod$. Since our machine $\mach$ is
permutable, this simple information will be enough to recover
$\rod{\mach}{\forest}{\Nod}$.

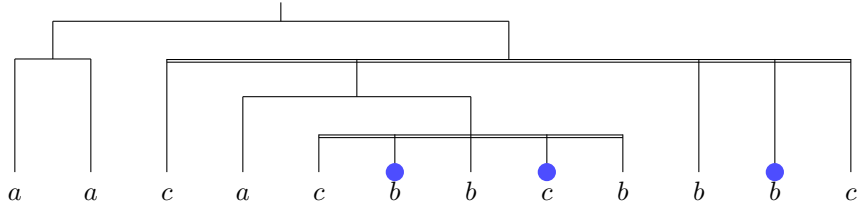
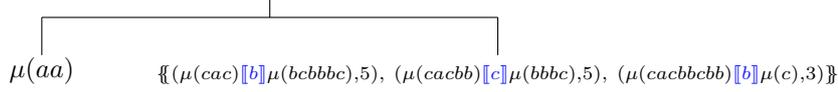
\begin{figure}[h!]

	\newcommand{\couleur}{blue!70}

     \begin{subfigure}[b]{1\textwidth}
     \centering
\begin{tikzpicture}{scale=1}

	\newcommand{\texte}{\small \bfseries \sffamily \mathversion{bold} }

	\node[above] at  (0,0)  {$a$};
	\node[above] at  (1,0)  {$a$};	
	\node[above] at  (2,0)  {$c$};	
	\node[above] at  (3,0)  {$a$};
	\node[above] at  (4,0)  {$c$};
	\node[above] at  (5,0)  {$b$};
	\node[above] at  (6,0)  {$b$};
	\node[above] at  (7,0)  {$c$};
	\node[above] at  (8,0)  {$b$};
	\node[above] at  (9,0)  {$b$};
	\node[above] at  (10,0)  {$b$};
	\node[above] at  (11,0)  {$c$};

	\draw[double] (4,0.975) -- (8,0.975);
	\draw[] (3,1.5) -- (6,1.5);
	\draw[double] (2,1.975) -- (11,1.975);
	\draw[] (0,2) -- (1,2);
	\draw[] (0.5,2.5) -- (6.5,2.5);
	
	\draw[] (0,2) -- (0,0.5);	
	\draw[] (1,2) -- (1,0.5);	
	\draw[] (0.5,2) -- (0.5,2.5);	
	
	\draw[] (2,2) -- (2,0.5);
	\draw[] (9,2) -- (9,0.5);
	\draw[] (10,2) -- (10,0.5);	
	\draw[] (11,2) -- (11,0.5);	

	\draw[] (6.5,2) -- (6.5,2.5);	
	\draw[] (3.5,2.5) -- (3.5,2.75);	

	\draw[] (3,1.5) -- (3,0.5);	
	\draw[] (4,1) -- (4,0.5);	
	\draw[] (5,1) -- (5,0.5);	
	\draw[] (6,1.5) -- (6,0.5);	
	\draw[] (7,1) -- (7,0.5);		
	\draw[] (8,1) -- (8,0.5);	

	\draw[] (4.5,2) -- (4.5,1.5);

	\fill[fill=\couleur]  (5,0.5)  circle (0.12);
	\fill[fill=\couleur]  (7,0.5)  circle (0.12);
	\fill[fill=\couleur]  (10,0.5)  circle (0.12);

\end{tikzpicture}
\caption{In blue, a set $\Nod$ of $3$ independent nodes in the forest from Figure~\ref{fig:iterable}}
     \end{subfigure}

     \begin{subfigure}[b]{1\textwidth}
     \centering\begin{tikzpicture}{scale=1}

	\newcommand{\texte}{\small \bfseries \sffamily \mathversion{bold} }


	\draw (0,9.5) -- (0,9);
	\draw (6,9.5) -- (6,9);
	\draw (0,9.5) -- (6,9.5);
	\draw (3,9.5) -- (3,9.75);

	\node[above] at (0,8.5) {$\mu(aa)$};
	\node[above] at (6,8.5) {$\substack{\multi{(\mu(cac)\cro{b}\mu(bcbbbc),5),~
	(\mu(cacbb)\cro{c}\mu(bbbc),5),~
	(\mu(cacbbcbb)\cro{b}\mu(c),3)}}$};

\end{tikzpicture}
\caption{The corresponding architecture}
     \end{subfigure}
 
\caption{\label{fig:archi} A set of independent nodes and its architecture.}

\end{figure}

\begin{definition}[Architecture]
Let $w \in A^*$, $\forest \in \Facts{\mu}{}{w}$ and $\Nod \in \Indeps{\forest}$.
We define the \emph{architecture} of $\Nod$ in $\forest$ by induction as follows:
\begin{itemize}
\item if $\forest = \vide$, then $k=0$. We define $\arch{\forest}{\Nod} \defined \vide$;
\item if $\forest = a$, then $k=0$. We define $\arch{\forest}{\Nod} \defined \mu(a)$;
\item otherwise $\forest = \tree{\forest_1} \cdots \tree{\forest_n}$ with $n \ge 1$:
\begin{itemize}
\item if $k=0$, we set $\arch{\forest}{\Nod} = \tree{\mu(\forest)}$;
\item else if $\Nod_1 \defined \Nod \cap \Nodes{\forest_1} \neq \varnothing$,
then $\Nod_1 \in \Indep^{|\Nod_1|}(\forest_1)$ (since $\forest_1 \not \in \Nod$ by iterability)
and $\Nod \smallsetminus \Nod_1 \in \Indep^{k-|\Nod_1|}(\tree{\forest_2} \cdots \tree{\forest_n})$
(since $\forest_2 \not \in \Nod$ by independence). We set:\\
$
\arch{\forest}{\Nod} \defined \tree{\arch{\forest_1}{\Nod_1}}
\arch{\tree{\forest_2} \cdots \tree{\forest_n}}{\Nod \smallsetminus \Nod_1}.
$
\item else if $\Nod_n \defined  \Nod {\cap} \Nodes{\forest_n} \neq \varnothing$,
we define symmetrically:\\
$
\arch{\forest}{\Nod} \defined 
\arch{\tree{\forest_1} \cdots \tree{\forest_{n-1}}}{\Nod \smallsetminus \Nod_n}
\tree{\arch{\forest_n}{\Nod_n}}
$
\item else $\Nod_1 = \Nod_n = \varnothing$ but $k > 0$, thus $n \ge 3$
and $\mu(\forest)$ is idempotent. We define:\\
$\arch{\forest}{\Nod} \defined \tree{\multi{(\lin{\forest}{\nod}, \depth{\forest}{\nod}): \nod \in \Nod}}
$
\end{itemize}
\end{itemize}
\vspace{0.2cm}
\end{definition}
Given a morphism, the set of architectures over bounded-height forests is finite.

\begin{claim} \label{claim:bounded-archis-K}
The set
$
\Archs{\mu}{\hei} \defined \{\arch{\forest}{\Nod}: \Nod
\in \Indeps{\forest}, \forest \in \Facts{\mu}{\hei}{w}, w \in A^* \}
$ is finite, given $k \ge 0$ and a morphism $\mu: A^* \rightarrow M $.
\end{claim}

\begin{claimproof}
The architectures from $\Archs{\mu}{\hei}$ are tree structures of height at most
$\hei$. Furthermore they have a branching bounded by $k{+}3$
and their leaves belong to a finite set (they are
either idempotents $e \in M$, or multisets of at most
$k$ elements of the form $(m\cro{u}m',d)$
with $m,m' \in M$, $|u| \le 2^{\hei}$
and $1 \le d \le \hei$).
\end{claimproof}
Using the permutability of the $k$-marble bimachine, we show that the production
over a set of independent nodes only depends on its architecture.
This result enables us to define the notion of production over
an architecture.

\begin{depro} \label{depro:prod-archi}
Let $\mach = (A,M, \mu, \oras, \lambda)$ be a $2^{\hei}$-permutable $k$-marble bimachine.
Let $w,w' \in A^*$, $\forest\in \Facts{\mu}{\hei}{w}$ and $\forest' \in \Facts{\mu}{\hei}{w'}$,
$\Nod \in \Indeps{\forest}$
and $\Nod'\in \Indeps{\forest'}$
such that
$\ar \defined \arch{\forest}{\Nod} = \arch{\forest'}{\Nod'}$.
Then
$
\rod{\mach}{\forest}{\Nod} = \rod{\mach}{\forest'}{\Nod'}
$.\\
We define $\rod{\mach}{}{\ar}$ as the above value.
\end{depro}
By using the previous statements, we get for all $w \in A^*$ and $\forest \in \Facts{\mu}{\hei}{w}$:
\begin{equation*}
\begin{aligned}
\sumi{\mach}(\forest) &
= \sum_{\Nod \in \Indeps{\forest}}
\rod{\mach}{\forest}{\Nod} = \sum_{\ar \in \Archs{\mu}{\hei}}
\sum_{\substack{\Nod \in \Indeps{\forest}\\
\arch{\forest}{\Nod} = \ar}}
\rod{\mach}{\forest}{\Nod}
\\
& = \sum_{\ar \in \Archs{\mu}{\hei}}
\rod{\mach}{}{\ar} \times \cou_{\ar}(\forest)\\
\end{aligned}
\end{equation*}
where $\cou_{\ar}(\forest) \defined
|\{\Nod \in \Indeps{\forest}:
\arch{\forest}{\Nod} = \ar\}|$.
It describes the number of multisets of independent nodes
which have a given architecture. Now we show how to compute this function
as a sum of a polyblind function and a polyregular function with lower growth.

\begin{lemma} \label{lem:counts}
Let $\mu: A^* \rightarrow M$.
Given an architecture $\ar \in \Archs{\mu}{\hei}$, one can build:
\begin{itemize}
\item a polyblind function ${\cou'_{\ar}}: (\apar{A})^* \rightarrow \Nat$;
\item a polyregular function ${\cou''_\ar}: (\apar{A})^* \rightarrow \Nat$
with growth at most $k{-}1$;
\end{itemize}
such that
$\cou_{\ar}(\forest) =  {\cou'_{\ar}}(\forest) +  {\cou''_{\ar}}(\forest)$ for all
$\forest \in \Facts{\mu}{\hei}{w}$ and $w \in A^*$.
\end{lemma}
To conclude the proof of Lemma~~\ref{lem:sumi}, we define the
function (which is polyblind):
\begin{equation*}
\sumi{\mach}' \defined \sum_{\ar \in \Archs{\mu}{\hei}} \rod{\mach}{}{\ar} \times {\cou'_{\ar}}.
\end{equation*}
We define similarly the following function which is polyregular and has growth at most $k{-}1$:
\begin{equation*}
\sumi{\mach}'' = \sum_{\ar \in \Archs{\mu}{\hei}} \rod{\mach}{}{\ar} \times {\cou''_{\ar}}.
\end{equation*}

\section{Outlook}

This paper provides a technical solution
to a seemingly difficult membership
problem. This result can be interpreted
both in terms of nested transducers (i.e. programs with visible or
blind recursive calls) and in terms of rational series. We conjecture that
the new techniques introduced here (especially
the induction techniques), and the concepts of
productions on  words and forests,
give an interesting toolbox to tackle
other decision problem for transducers
such as equivalence or membership issues.
It could also be interesting to characterize
polyblind functions as the series computed by specific
weighted automata over $(\Nat, +, \times)$.


\newpage

\bibliography{marbles}

\begin{thebibliography}{10}

\bibitem{berstel2011noncommutative}
Jean Berstel and Christophe Reutenauer.
\newblock {\em Noncommutative rational series with applications}, volume 137.
\newblock Cambridge University Press, 2011.

\bibitem{bojanczyk2018polyregular}
Mikolaj Boja{\'{n}}czyk.
\newblock Polyregular functions.
\newblock {\em arXiv preprint arXiv:1810.08760}, 2018.

\bibitem{bojanczyk2019string}
Miko{l}aj Boja{\'{n}}czyk, Sandra Kiefer, and Nathan Lhote.
\newblock String-to-string interpretations with polynomial-size output.
\newblock In {\em 46th International Colloquium on Automata, Languages, and
  Programming, {ICALP} 2019}, pages 106:1--106:14, 2019.
\newblock \href {https://doi.org/10.4230/LIPIcs.ICALP.2019.106}
  {\path{doi:10.4230/LIPIcs.ICALP.2019.106}}.

\bibitem{chytil1977serial}
Michal~P Chytil and Vojt{\v{e}}ch J{\'a}kl.
\newblock Serial composition of 2-way finite-state transducers and simple
  programs on strings.
\newblock In {\em 4th International Colloquium on Automata, Languages, and
  Programming, {ICALP} 1977}, pages 135--147. Springer, 1977.

\bibitem{dave2018regular}
Vrunda Dave, Paul Gastin, and Shankara~Narayanan Krishna.
\newblock Regular transducer expressions for regular transformations.
\newblock In {\em Proceedings of the 33rd Annual ACM/IEEE Symposium on Logic in
  Computer Science}, pages 315--324. ACM, 2018.

\bibitem{doueneau2021pebble}
Ga{\"{e}}tan Dou{\'{e}}neau{-}Tabot.
\newblock Pebble transducers with unary output.
\newblock In {\em 46th International Symposium on Mathematical Foundations of
  Computer Science, {MFCS} 2021, August 23-27, 2021, Tallinn, Estonia}, 2021.

\bibitem{doueneau2020register}
Ga{\"{e}}tan Dou{\'{e}}neau{-}Tabot, Emmanuel Filiot, and Paul Gastin.
\newblock Register transducers are marble transducers.
\newblock In {\em 45th International Symposium on Mathematical Foundations of
  Computer Science, {MFCS} 2020, August 24-28, 2020, Prague, Czech Republic},
  2020.

\bibitem{engelfriet2015two}
Joost Engelfriet.
\newblock Two-way pebble transducers for partial functions and their
  composition.
\newblock {\em Acta Informatica}, 52(7-8):559--571, 2015.

\bibitem{engelfriet2001mso}
Joost Engelfriet and Hendrik~Jan Hoogeboom.
\newblock {MSO} definable string transductions and two-way finite-state
  transducers.
\newblock {\em ACM Transactions on Computational Logic (TOCL)}, 2(2):216--254,
  2001.

\bibitem{globerman1996complexity}
Noa Globerman and David Harel.
\newblock Complexity results for two-way and multi-pebble automata and their
  logics.
\newblock {\em Theoretical Computer Science}, 169(2):161--184, 1996.

\bibitem{gurari1982equivalence}
Eitan~M Gurari.
\newblock The equivalence problem for deterministic two-way sequential
  transducers is decidable.
\newblock {\em SIAM Journal on Computing}, 11(3):448--452, 1982.

\bibitem{nguyen2021implicit}
L{\^{e}}~Th{\`{a}}nh~Dung Nguy{\^{e}}n.
\newblock {\em Implicit automata in linear logic and categorical transducer
  theory}.
\newblock PhD thesis, Universit{\'e} Paris 13, 2021.

\bibitem{nous2020comparison}
L{\^{e}}~Th{\`{a}}nh~Dung Nguy{\^{e}}n, Camille No{\^{u}}s, and Pierre Pradic.
\newblock Comparison-free polyregular functions.
\newblock In {\em 48th International Colloquium on Automata, Languages, and
  Programming, {ICALP} 2021, July 12-16, 2021, Glasgow, Scotland (Virtual
  Conference)}, 2021.

\bibitem{shepherdson1959reduction}
John~C Shepherdson.
\newblock The reduction of two-way automata to one-way automata.
\newblock {\em IBM Journal of Research and Development}, 3(2):198--200, 1959.

\bibitem{simon1990factorization}
Imre Simon.
\newblock Factorization forests of finite height.
\newblock {\em Theor. Comput. Sci.}, 72(1):65--94, 1990.
\newblock \href {https://doi.org/10.1016/0304-3975(90)90047-L}
  {\path{doi:10.1016/0304-3975(90)90047-L}}.

\end{thebibliography}

\appendix

\newpage

\section{Proof of Proposition \ref{prop:blind-confuse}}

Since polyblind functions correspond to the closure of regular functions
under sums and Hadamard products, the result is an immediate consequence of
lemmas \ref{lem:confuse-regular} and \ref{lem:confuse-stability} below.

\begin{lemma} \label{lem:confuse-regular}
A regular function is $k$-repetitive for all $k \ge 1$.
\end{lemma}

\begin{lemma} \label{lem:confuse-stability}
Let $f,g: A^* \rightarrow \Nat $ be $k$-repetitive for some $k \ge 1$.
Then $f{+}g$ and $ f \had g$ are $k$-repetitive.
\end{lemma}

The proof of Claim \ref{claim:soustraction} (for $f{-}g$) is exactly the same
as that of Lemma \ref{lem:confuse-stability}.

\subsection{Proof of Lemma~\ref{lem:confuse-regular}}

Let $f: A^* \rightarrow \Nat$ be computed
by a bimachine $\mach = (A,M, \mu, \lambda)$.
In this proof, we shall use the notion of \emph{productions}
of a $k$-marble bimachine introduced in Subsection~\ref{ssec:prod-marbles},
and the properties of this productions shown in Appendix~\ref{app:prod-marbles}.
Note that we only use them for $k=1$
since we deal here with bimachines, i.e. $1$-marble bimachines.

Let $\idem{}$ be the idempotence index of $M$,
that is the smallest $\idem{} > 0$ such that $m^\idem{}$ is idempotent
for all $m \in M$. 
Let $k \ge 1$, $\alpha,\alpha_0,  {u_1}, \alpha_1, \dots,  {u_k}, \alpha_k, \beta \in A^*$
and $\omega \ge 1$ be a multiple of $\idem{}$.
Let  $W: \Nat^{k} \rightarrow A^*$ defined by:
\begin{equation*}
W : X_1, \dots, X_k \mapsto 
\left(\alpha_0 \prod_{i=1}^k  {u_i}^{\omega X_i}\alpha_i\right).
\end{equation*}
and let $w \defined W(1, \dots, 1)$.
By definition of $\idem{}$, we have for all $1 \le i \le k$
that $e_i \defined \mu( {u_i}^\omega)$ is an idempotent.
Hence
$
p 
\defined \mu(w) = \mu(\alpha_0) \prod_{i=1}^k e_i \mu(\alpha_i)
 = \mu\left(\alpha_0 \prod_{i=1}^k  {u_i}^{\omega X_i}\alpha_i\right)
$
is independent of $ X_1, X_2, \dots, X_k \ge 1$ and
$e \defined p^{\omega}$ is idempotent.

Let $\ov{X} \defined X_1 \dots, X_k \ge 3$ and
$\ov{Y} \defined Y_1 \dots, Y_k \ge 3$, then we can decompose
by Lemma~\ref{lem:coupe-factors}:
\begin{equation*}
\begin{aligned}
&f(\alpha w^{2\omega- 1}W(\overline{X}) w^{\omega-1} W(\overline{Y}) w^\omega \beta)\\
&= \rod{\mach}{}{\cro{\alpha w^{2\omega-1}} p p^{\omega-1} p p^\omega \mu(\beta)}
 +  \rod{\mach}{}{\mu(\alpha)p^{2\omega-1} \cro{\alpha_0 \prod_{i=1}^k  {u_i}^{\omega X_i}\alpha_i}p^{\omega-1} p p^\omega \mu(\beta)} \\
&+ \rod{\mach}{}{\mu(\alpha) p^{2\omega-1} p \cro{w^{\omega-1}} p p^\omega \mu(\beta)}
 +  \rod{\mach}{}{\mu(\alpha)p^{2\omega-1} p p^{\omega-1} \cro{\alpha_0 \prod_{i=1}^k  {u_i}^{\omega Y_i}\alpha_i} p^\omega \mu(\beta)} \\
& + \rod{\mach}{}{\mu(\alpha) p^{2\omega-1} p p^{\omega-1} p \cro{w^\omega \beta}} \\
\end{aligned}
\end{equation*}
\begin{equation*}
\begin{aligned}
&= \rod{\mach}{}{\cro{\alpha w^{2\omega-1}} e \mu(\beta)}
+ \rod{\mach}{}{\mu(\alpha) e \cro{w^\omega \beta}}
 + \rod{\mach}{}{\mu(\alpha) e \cro{w^{\omega-1}} p e \mu(\beta)}\\
& +  \rod{\mach}{}{\mu(\alpha)e p^{\omega-1} \cro{\alpha_0 \prod_{i=1}^k  {u_i}^{\omega X_i}\alpha_i} e \mu(\beta)}  +  \rod{\mach}{}{\mu(\alpha)e p^{\omega-1}  \cro{\alpha_0 \prod_{i=1}^k  {u_i}^{\omega Y_i}\alpha_i} e\mu(\beta)}.
\end{aligned}
\end{equation*}

The three first terms are constants, we only need
to focus on the two last ones. For this, we decompose
their productions using the fact that the $e_i = \mu(u_i^{\omega})$ are
idempotents.

\begin{sublemma} \label{slem:final-confuse}
For $m,n \in M$, there exists a polynomial
$L \in \Rel[X_1, \dots, X_k]$ of degree $1$ such that
for all $\ov{X} = X_1, \dots, X_k \ge 3$ we have:
\begin{equation*}
\rod{\mach}{}{m \cro{\alpha_0 \prod_{i=1}^k  {u_i}^{\omega X_i}\alpha_i} n}  = L(\ov{X}).
\end{equation*}
\end{sublemma}

\begin{proof}
By cutting the productions with Lemma~\ref{lem:coupe-factors}, we get:
\begin{equation*}
\begin{aligned}
&\rod{\mach}{}{m \cro{\alpha_0 \prod_{i=1}^k  {u_i}^{\omega X_i}\alpha_i} n}
= \sum_{j=0}^k \rod{\mach}{}{m \left(\prod_{i=0}^{j{-}1} \mu(\alpha_i)  {e_{i{+}1}}\right)
\cro{\alpha_j} \left(\prod_{i=j{+}1}^{k}  {e_i} \mu(\alpha_i)\right) n}\\
&+ \sum_{j=1}^k \rod{\mach}{}{m \left(\mu(\alpha_0)\prod_{i=1}^{j{-}1}  {e_i }\mu(\alpha_i) \right)
\cro{( {u_j}^\omega)^{X_j}} \mu(\alpha_j) \left(\prod_{i=j{+}1}^{k}  {e_i} \mu(\alpha_i)\right) n}.
\end{aligned}
\end{equation*}
Since $e_j = \mu(u_j^\omega)$, by item~\ref{po:u-x:2}
of Lemma~\ref{lem:u-x}, we see that each term of the second line
is a polynomial of degree at most $1$ in $X_j$.
Furthermore the sum in the first line is a constant.
\end{proof}

Let $L$ be the polynomial given by applying Sublemma~\ref{slem:final-confuse}
to any of the two last terms of its equation.
It follows that for all $\overline{X}, \overline{Y} \ge 3$ we have
for some $C \in \Nat$:
\begin{equation*}
\begin{aligned}
f(\alpha w^{2\omega- 1}W(\overline{X}) w^{\omega-1} W(\overline{Y}) w^\omega \beta)
&= L(\overline{X}) + L(\overline{Y}) + C.
\end{aligned}
\end{equation*}

Since $L$ is a polynomial of degree $1$, we finally
obtain the function $F$ of Lemma~\ref{lem:confuse-regular}
(which is in fact a polynomial) by grouping the
terms in $X_i$ and $Y_i$ together.

\subsection{Proof of Lemma~\ref{lem:confuse-stability}}

Let $f$ and $g$ be two $k$-repetitive functions for some $k \ge 1$.
We show that $f \had g \ge 0$ is $k$-repetitive (the other cases
are very similar). Let $\idem{f}$ and $\idem{g}$ be the $\idem{}$ of each functions
given in Definition \ref{def:repetitive},
we set $\idem{f \had g} \defined \idem{f} \times  \idem{g}$.

Let $\alpha,\alpha_0,  {u_1}, \alpha_1, \dots,  {u_k}, \alpha_k, \beta \in A^*$
and $\omega \ge 1$ multiple of $\idem{f \had g}$. Note that $\omega$ is 
also a multiple of $\idem{f}$ and $\idem{g}$.
Let $w \defined \alpha_0 \prod_{1 \le i \le k}  {u_i}^{\omega} \alpha_i$
and define $W: \Nat^{2k} \rightarrow A^*$:
\begin{equation*}
W: X_1, Y_1, \dots, X_k, Y_k \mapsto \alpha w^{2\omega- 1}
\left(\alpha_0 \prod_{i=1}^k  {u_i}^{\omega X_i}\alpha_i\right)
w^{\omega - 1}
\left(\alpha_0 \prod_{i=1}^k  {u_i}^{\omega Y_i} \alpha_i\right)
w^{\omega} \beta
\end{equation*}

Then there exists $F_f, F_g$ such that
for all $X_1, Y_1, \dots, X_k, Y_k \ge 3$:
\begin{equation*}
\left\{
    \begin{array}{l}
        f(W(X_1, Y_1, \dots, X_k, Y_k)) = F_f(X_1 + Y_1, \dots, X_k + Y_k) \\
	g(W(X_1, Y_1, \dots, X_k, Y_k)) = F_g(X_1 + Y_1, \dots, X_k + Y_k).
    \end{array}
\right.
\end{equation*}
Hence we have 
\begin{equation*}
\begin{aligned}
        (f\had g)(W(X_1, Y_1, \dots, X_k, Y_k))
        &= F_f(X_1 + Y_1, \dots, X_k + Y_k) \times
	F_g(X_1 + Y_1, \dots, X_k + Y_k)\\
	&= (F_f \had F_g)(X_1 + Y_1, \dots, X_k + Y_k)
    \end{aligned}
\end{equation*}
which concludes the proof with $F_{f \had g} \defined F_f \had F_g$.

\section{Properties of productions}

In the forthcoming appendices, we denote by $\multi{s_1 \con r_1, \dots, s_n \con r_n}$
a multiset with $n$ distinct elements $s_1, \dots, s_n$, in which $r_i \ge 0$ is the
multiplicity of $s_i$ (thus the cardinal is $r_1 + \cdots + r_n$).

\label{app:prod-marbles}

\subsection{Proof of Lemma~\ref{lem:coupe-prod}}

We want to show that if $J_1, \dots J_n$ is a partition of $[1{:}|w|]$,
then:
\begin{equation*}
f(w) = \sum_{\multi{I_1, \dots, I_k} \subseteq \{J_1, \dots, J_n\}} \rod{\mach}{w}{\multi{I_1, \dots, I_k} }.
\end{equation*}

\begin{sublemma} \label{slem:f-prod}
Let $\mach$ be a $k$-marble bimachine which computes
a function $f: A^* \rightarrow \Nat$. Then
for all $w \in A^*$ we have:
\begin{equation*}
f(w) = \sum_{\multi{i_1, \dots, i_k} \subseteq [1{:}|w|]}
\rod{\mach}{w}{\multi{i_1, \dots, i_k}}
= \rod{\mach}{w}{\multi{[1{:}|w|]{\con}k}}.
\end{equation*}
\end{sublemma}

\begin{proof} The second equality
is obtained by definition of $\rod{\mach}{w}{\multi{[1{:}|w|]{\con}k}}$.
We show the first one by induction on $k \ge 1$. The base
case is trivial, let us assume that $k \ge 2$. Let $\mach = (A, M, \mu, \oras, \lambda)$ 
be the $k$-marble bimachine which computes $f$. For $1 \le i_k \le |w|$,
let $\exte_{i_k} \defined \lambda(\mu(w[1{:}i_k{-}1]),w[i_k],\mu(w[i_k{+}1{:}|w|])) \in \oras$
and $\mach_{\exte_{i_k}}$ be the $(k{-}1)$-marble bimachine which computes it.
Then we have:
\begin{equation*}
\begin{aligned}
f(w) &= \sum_{i_k=1}^{|w|}  \exte_{i_k}(w[1{:}i_k])
\substack{\tnorm{~~by definition of a bimachine with external
marble functions}}\\
& = \sum_{i_k=1}^{|w|} \sum_{\multi{i_1 \le \cdots \le i_{k{-}1}}\subseteq{[1{:}i_k]}}
\rod{\mach_{\exte_{i_k}}}{w[1{:}i_k]}{\multi{i_1, \dots, i_{k{-}1}}}
\substack{\tnorm{~~by induction hypothesis}}\\
& = \sum_{\multi{i_1\le \dots \le i_{k{-}1} \le i_k}\subseteq{[1{:}|w|]}}
\rod{\mach_{\exte_{i_k}}}{w[1{:}i_k]}{\multi{i_1, \dots, i_{k{-}1}}} \\
& = \sum_{\multi{i_1\le \dots \le i_k} \subseteq{[1{:}|w|]}} \rod{\mach}{w}{\multi{i_1, \dots, i_k}}
\substack{\tnorm{~~by definition of } \prodd_\mach.}\\
\end{aligned}
\end{equation*}
\end{proof}

We now show how to decompose a production
when one set is split.

\begin{definition}
We say that $(r_1, \dots, r_n) \in \Nat^n$
is a $k$-sum if $n \ge 0$ and $r_1 + \cdots + r_n = k$.
We denote by $\Sigma_k$ the set of
$k$-sums.
\end{definition}

\begin{sublemma} \label{slem:prod-goes1}
Let $\mach$ be a $k$-marble bimachine, $w \in A^*$,
$I_1, \dots, I_n$ be disjoint subsets of $[1{:}|w|]$
and $(r_1, \dots, r_n) \in \Sigma_k$.
Assume that $I_1 = J \uplus J'$, then:
\begin{equation*}
\rod{\mach}{w}{\multi{I_1 \con r_1, \dots, I_n\con r_n}} = \sum_{{(j,j') \in \Sigma_{r_1}}}
\rod{\mach}{w}{\multi{J{\con}j, J' \con j', I_2{\con}r_2, \dots, I_n{\con}r_n} }.
\end{equation*}
\end{sublemma}

\begin{proof}
By definition of productions and since the sets are disjoint, we have:
\begin{equation*}
\begin{aligned}
& \rod{\mach}{w}{\multi{I_1\con r_1, \dots, I_n \con r_n}}\\
&= \sum_{\substack{\multi{i^1_1, \dots, i^1_{r_1}} \subseteq I_1}}
\cdots \sum_{\substack{\multi{i^n_1, \dots, i^n_{r_n}} \subseteq I_n}}
\rod{\mach}{w}{\multi{i^1_1, \dots, i^n_{r_n}}}\\
&= \sum_{ \substack{0 \le j \le r_1\\ \multi{i^1_1, \dots, i^1_{j}}
\subseteq J\\ \multi{i^1_{j+1}, \dots, i^1_{r_1}} \subseteq J'}}
\sum_{\substack{\multi{i^2_1, \dots, i^2_{r_2}} \subseteq I_2}}
\cdots \sum_{\substack{\multi{i^n_1, \dots, i^n_{r_n}} \subseteq I_n}}
\rod{\mach}{w}{\multi{i^n_1, \dots, i^n_{r_n}}}\\
\end{aligned}
\end{equation*}
The last line is obtained since $J$ and $J'$ are disjoint. It is easy
to see that it corresponds to the expression we are looking for.
\end{proof}

This result is generalized by induction to
any partition of $I_1$ in Sublemma \ref{slem:prod-goes}.

\begin{sublemma} \label{slem:prod-goes}
Let $\mach$ be a $k$-marble bimachine, $w \in A^*$,
$I_1, \dots, I_n$ be disjoint subsets of $[1{:}|w|]$
and $(r_1, \dots, r_n) \in \Sigma_k$.
Assume that $I_1 = J_1 \uplus \cdots \uplus J_p$,
then:
\begin{equation*}
\rod{\mach}{w}{\multi{I_1 \con r_1, \dots, I_n\con r_n}} = \sum_{{(j_1,\dots, j_p) \in \Sigma_{r_1}}}
\rod{\mach}{w}{\multi{J_1{\con}j_1, \dots, J_p \con j_p, I_2{\con}r_2, \dots, I_n{\con}r_n} }.
\end{equation*}
\end{sublemma}

\begin{proof} The result is shown by induction on $p$
with Sublemma~\ref{slem:prod-goes1}.
\end{proof}

\begin{remark}
This result is in fact stronger than what we need for Lemma \ref{lem:coupe-prod},
but it shall be re-used in the following sections.
\end{remark}

We conclude this subsection by showing Lemma~\ref{lem:coupe-prod}
as follows:
\begin{equation*}
\begin{aligned}
f(w) & = \sum_{(r_1, \dots, r_n) \in \Sigma_k} \rod{\mach}{w}{\multi{J_1 \con r_1, \dots, J_n \con r_n} }
\substack{\tnorm{~~~by Sublemma~\ref{slem:f-prod} and
Sublemma~\ref{slem:prod-goes}}}\\
&= \sum_{\multi{I_1, \dots, I_k} \subseteq \{J_1, \dots, J_n\}} \rod{\mach}{w}{\multi{I_1, \dots, I_k} }.
\end{aligned}
\end{equation*}

\subsection{Proof of Proposition-Definition \ref{depro:factor-k}}

We shall in fact show Proposition-Definition~\ref{depro:factors},
which is a generalization of Proposition-Definition~\ref{depro:factor-k}.
It shall be useful in the rest of these appendices.

\begin{sublemma} \label{slem:positi}
Let $\mach = (A,M, \mu, \oras, \lambda)$ be a $k$-marble bimachine.
Let $a_1, \dots, a_n \in A$, $(r_1, \dots, r_n) \in \Sigma_k$,
 $w \defined v_0 a_1 \cdots a_n v_n \in A^*$ and
$w' \defined v'_0 a_1 \cdots a_n v'_n \in A^*$
such that $\mu(v_j) = \mu(v'_j)$ for all $0 \le j \le n$.
For $1 \le j \le n$ let $i_j \in [1{:}|w|]$ (resp. $i'_j \in [1{:}|w'|]$) be the position
of $a_j$ in $w$ (resp. $w'$). Then
$
\rod{\mach}{w}{\multi{i_1 \con r_1 , \dots, i_n \con r_n}} =
\rod{\mach}{w'}{\multi{i'_1 \con r_1 , \dots, i'_n \con r_n}}.
$
\end{sublemma}

\begin{proof} The proof is done by induction on $k \ge 1$.
For $k=1$ the base case is obvious, let us assume that $k=2$.
Without loss of generality assume that $r_n \ge 1$.
Let $\exte \defined \lambda(\mu(w[1{:}i_n{-}1]),w[i_n],\mu(w[i_n{+}1{:}|w|]))$
and $\exte' \defined \lambda(\mu(w'[1{:}i'_n{-}1]),w[i'_n],\mu(w[i'_n{+}1{:}|w'|]))$.
Note that $i_1 < \cdots < i_n$ and $i'_1 < \cdots < i'_n$.
Thanks to the hypotheses of Sublemma \ref{slem:positi} we get $\exte = \exte'$.
Let $\mach_{\exte}$ be the submachine computing $\exte$, we get:
\begin{equation*}
\begin{aligned}
\rod{\mach}{w}{\multi{i_1 \con r_1 , \dots, i_n \con r_n}} 
& = \rod{\mach_{\exte}}{w[1{:}i_n]}{\multi{i_1 \con r_1 , \dots, i_n \con (r_n{-}1)}}
\substack{\tnorm{~~~by definition}}\\
&= \rod{\mach_{\exte}}{w'[1{:}i'_n]}{\multi{i'_1 \con r_1 , \dots, i'_n \con (r_n{-}1)}}
\substack{\tnorm{~~~by induction hypothesis}}\\
& = \rod{\mach}{w'}{\multi{i'_1 \con r_1 , \dots, i'_n \con r_n}}
\substack{\tnorm{~~~by definition}}.\\
\end{aligned}
\end{equation*}

\end{proof}

Let us now give an extension of Equation \ref{eq:facteurs}
to state our generalization of  Proposition-Definition \ref{depro:factor-k}.
Let $(r_1, \dots, r_n) \in \Sigma_k$,
 $w \defined v_0 u_1 \cdots u_n v_n \in A^*$, and
for $1 \le i \le n$, let $I_i \subseteq [1{:}|w|]$ be the set of positions corresponding to $w_i$.
We define:
\begin{equation*}
\rod{\mach}{}{v_0 \cro{u_1}_{r_1} v_1 \cdots \cro{u_n}_{r_n} v_n} \defined
\rod{\mach}{w}{\multi{I_1\con r_1, \dots, I_n\con r_n}}.
\end{equation*}
Note that Equation \ref{eq:facteurs} is the case when
$r_i = 1$ for all $1 \le i \le n$ (thus $n=k$). On the other
hand, if $r_i = 0$, then the element $\cro{u_i}_{r_i}$ can be equivalently replaced
by $\mu(u_i)$ (it means that in fact $u_i$ is not
used when making the production).

\begin{depro} \label{depro:factors}
Let $\mach = (A,M, \mu, \oras, \lambda)$  be a $k$-marble bimachine.
Let $(r_1, \dots, r_n) \in \Sigma_k$, $v_0 u_1 \cdots u_n v_n \in A^*$
and $v'_0 u_1 \cdots u_n v'_n \in A^*$
such that $\mu(v_i) = \mu(v'_i)$ for all $0 \le i \le n$.
Then we have
$
\rod{\mach}{}{v_0 \cro{u_1}_{r_1} v_1 \cdots \cro{u_n}_{r_n} v_n} =
\rod{\mach}{}{v'_0 \cro{u_1}_{r_1} v'_1 \cdots \cro{u_n}_{r_n} v'_n}.
$
Let $m_i \defined \mu(v_i) = \mu(v'_i)$, we define
$\rod{\mach}{}{m_0 \cro{u_1}_{r_1} m_1 \cdots \cro{u_n}_{r_n} m_n}$
as the previous value.
\end{depro}

\begin{proof}
Let $w \defined v_0 u_1 \cdots u_n v_n$ and $w' \defined v_0 u'_1 \cdots u'_n v_n$.
Let $I_1, \dots, I_n \subseteq [1{:}|w|]$ (resp. $I'_1, \dots, I'_n \subseteq [1{:}|w'|]$)
be the sets of positions of $u_1, \dots, u_n$ in $w$ (resp. $w'$). For $1 \le i \le n$, 
let $\sigma_i : I_i \rightarrow I'_i$ be the unique monotone (for $<$) bijection
between these sets.
Then if $\multi{i^1_1, \dots, i^1_{r_1}} \subseteq I_1, \dots, \multi{i^n_1, \dots, i^n_{r_n}} \subseteq I_n$,
it follows from Sublemma \ref{slem:positi} that:
\begin{equation*}
\begin{aligned}
&\rod{\mach}{w}{\multi{i^1_1, \dots, i^1_{r_1}} \uplus \cdots \uplus \multi{i^n_1, \dots, i^n_{r_n}}}\\
&= \rod{\mach}{w'}{\multi{\sigma_{1}(i^1_1), \dots, \sigma_1(i^1_{r_1})}
\uplus \cdots \uplus \multi{\sigma_n(i^n_1), \dots, \sigma_n(i^n_{r_n})})}.
\end{aligned}
\end{equation*}
Finally we get $\rod{\mach}{w}{\multi{I_1\con r_1, \dots, I_n\con r_n}}
= \rod{\mach}{w'}{\multi{I'_1\con r_1, \dots, I'_n\con r_n}}$ by summing
all the terms of the above form.
\end{proof}

\subsection{Further results}

In this subsection, we establish further properties of
productions which are to be used later.
We first give an analogous of Sublemma~\ref{slem:prod-goes},
in order to ``split'' the productions on the factors.

\begin{lemma} \label{lem:coupe-factors}
Let $\mach = (A,M, \mu, \oras, \lambda)$ be a $k$-marble bimachine.\\
Let $v_1 \cdots v_n u_1 \cdots u_{X} v'_1 \cdots v'_{n'} \in A^*$ and
$(\delta_1, \dots, \delta_n,r, \delta'_1, \dots, \delta'_{n'}) \in \Sigma_k$. \\
Let $\ce~\defined \cro{v_1}_{\delta_1} \cdots \cro{v_p}_{\delta_n}$
and $\de \defined \cro{v'_1}_{\delta'_1} \cdots \cro{v'_{p'}}_{\delta'_{n'}}$,
then we have:
\begin{equation*}
\rod{\mach}{}{\ce~\cro{u_1 \cdots u_{X}}_r~\de}
= \sum_{(r_1, \dots, r_X) \in \Sigma_r}
\rod{\mach}{}{\ce~\cro{u_1}_{r_1} \cdots \cro{u_{X}}_{r_X} ~\de}
\end{equation*}
\end{lemma}

\begin{proof} We consider those productions as productions
of sets of positions within the word $v_1 \cdots v_n u_1 \cdots u_{X} v'_1 \cdots v'_{n'}$.
The result follows from Sublemma~\ref{slem:prod-goes}
by choosing $J_1, \dots, J_X$ as the positions
corresponding to $u_1, \dots, u_X$,
and using Proposition-Definition~\ref{depro:factors}.
\end{proof}

We now want to describe in a more precise way
what happens in Lemma~\ref{lem:coupe-factors}
when $u_1 = \cdots = u_X = u$ and $\mu(u)$ is
an idempotent. This is the purpose of Lemma~\ref{lem:u-x}
below. Given $r \ge 0$ and $(r_1, \dots, r_X) \in \Sigma_r$
we define $\shape{r_1, \dots, r_X}$ as the tuple 
obtained from $(r_1, \dots, r_X)$ by replacing
the maximal blocks of the form $0, \dots, 0$
by a single $0$.
For instance $\shape{0,1,0,0,0,1,2} = (0,1,0,1,2) \in \Sigma_4$.
Let $\Shapes_r$ be the set of all
shapes of elements of $\Sigma_r$
(note that $\Shapes_r$ is a finite subset of $\Sigma_r$).

\begin{lemma} \label{lem:u-x} Let $\mach = (A,M, \mu, \oras, \lambda)$ be a $k$-marble bimachine.\\
Let $(\delta_1, \dots, \delta_n, r, \delta'_1, \dots, \delta'_n) \in \Sigma_k$
and $v_1, \dots, v_n,u,v'_1, \dots, v'_{n'} \in A^*$.
Assume that $e \defined \mu(u)$ is an idempotent.
Let $\ce~\defined \cro{v_1}_{\delta_1} \cdots \cro{v_n}_{\delta_n}$
and $\de\defined \cro{v'_1}_{\delta'_1} \cdots \cro{v'_{n'}}_{\delta'_{n'}}$.
If $X \ge 2r{+}1$, then:
\begin{enumerate}
\item
$
\displaystyle
\rod{\mach}{}{\ce~\cro{u^X}_r~\de} = \sum_{s = (s_1, \dots, s_Y) \in \Shapes_r} P_{s} (X) \times 
\rod{\mach}{}{\ce~\cro{u}_{s_1} \cdots \cro{u}_{s_Y}~\de} 
$\\
where $P_s$ is a polynomial in $X$ of degree at most $r$
which is independent of $\ce ,\de$ and $u$;
\item \label{po:u-x:2} $\rod{\mach}{}{\ce~\cro{u^X}_r~\de}$  is a polynomial
in $X$ of degree at most $r$;
\item \label{po:u-x:3} the coefficient in $X$ of
the polynomial $\rod{\mach}{}{\ce~\cro{u^X}_r~\de}$ is:
\begin{enumerate}
\item $0$ if $r=0$; 
\item \label{po:lele} $\rod{\mach}{}{\ce~e\cro{u} e~\de}$ if $r=1$.
\end{enumerate}
\end{enumerate}
\end{lemma}

\begin{proof}
We have
$
\displaystyle
\rod{\mach}{}{\ce~\cro{u^X}_r~\de} = \sum_{(r_1, \dots, r_X) \in \Sigma_r}
\rod{\mach}{}{\ce~\cro{u}_{r_1} \cdots \cro{u}_{r_X}~\de}
$ by Lemma~\ref{lem:coupe-factors}.
We shall recombine several terms of this sum thanks to the claim below.

\begin{claim} \label{claim:same-shape} Let $(r_1, \dots, r_X) \in \Sigma_r$
and $(s_1, \dots, s_Y) \defined \shape{r_1, \dots, r_X}$, then:
\begin{equation*}
\rod{\mach}{}{\ce~\cro{u}_{r_1} \cdots \cro{u}_{r_X}~\de}
= \rod{\mach}{}{\ce~\cro{u}_{s_1} \cdots \cro{u}_{s_Y}~\de}
\end{equation*}
\end{claim}

\begin{claimproof}
Transforming several consecutive $\cro{u}_0$ 
in a single one corresponds to transforming the
concatenation of several idempotent $e = \mu(u)$ in a single one.
\end{claimproof}

Given $s \in \Shapes_r$, let $P_s(X)$ be the number of tuples
$(r_1, \dots, r_X) \in \Sigma_r$ such that $\shape{r_1, \dots, r_X} = s$.
This quantity does not depend on $\ce, \de$ nor $u$. Furthermore, it
follows from Claim \ref{claim:same-shape} that:
\begin{equation}
\label{eq:p-s}
\rod{\mach}{}{\ce~\cro{v^X}_r~\de} = \sum_{s = (s_1, \dots, s_Y) \in \Shapes_r} P_{s} (X) \times 
\rod{\mach}{}{\ce~\cro{v}_{s_1} \cdots \cro{v}_{s_Y}~\de} 
\end{equation}
We still need to show that the $P_s$ are polynomials of degree at most $r$,
this is Claim \ref{claim:ps-poly}.

\begin{claim} \label{claim:ps-poly} For $X \ge 2r{+}1$, $P_s$ is a polynomial in $X$
of degree at most $r$.
\end{claim}

\begin{claimproof} If $s \defined (s_1, \dots, s_Y) \in \Shapes_r$, one has
$Y \le 2r{+}1$ since there are no two consecutive $0$.
Let $0 \le z \le r{+}1$ be the number of zeros in $s$,
then $P_s(X)$ is the number of tuples $(q_1, \dots, q_z) \in \Sigma_{X-Y}$
(it describes how many times we count each $0$).
We show by induction on $z \ge 0$ that
it is a polynomial in $X$ of degree at most $z{-}1$
whenever $X-Y \ge 0$.
\end{claimproof}

It follows directly from Claim \ref{claim:ps-poly} and Equation \ref{eq:p-s}
that $\rod{\mach}{}{\ce~\cro{u^X}_r~\de}$  is a polynomial in $X$
for $X \ge 2r{+}1$. It remains to deal specifically with the cases $r=0$
(which is obvious, we only have a constant)
and $r=1$. Note that $\Shapes_1 = \{(0,1), (1,0),  (0,1,0)\}$.
By definition of $P_s$ we get
$P_{(0,1)}(X) = P_{(1,0)}(X) = 1$ and $P_{(0,1,0)} (X) = X-2$
for $X \ge 3$. Hence the coefficient in $X$ is
$\rod{\mach}{}{\ce~\cro{u}_{0} \cro{u}_{1} \cro{u}_{0}~\de}$
which can be rewritten $\rod{\mach}{}{\ce~e\cro{u} e~\de}$.
\end{proof}

\section{Proof of Proposition~\ref{prop:factor-perm}}

In order to show Proposition~\ref{prop:factor-perm},
we shall use Lemma~\ref{lem:ecarte} below.
Intuitively, it says that under the same conditions
as those of Definition \ref{def:factor-perm},
if we extract one factor $u_j$ while preserving 
$\lefe_j$ and $\rige_j$, then the production will be
the same (see Figure \ref{fig:ecartee}).

\begin{lemma} \label{lem:ecarte}
Let $K \ge 0$.
Let $f$ be a $k$-repetitive function computed
by a $k$-marble bimachine $\mach = (A,M, \mu, \oras, \lambda)$.
The following holds whenever $\ell+x+r = k$:
\begin{itemize}
\item let $\ce$ (resp. $\de$) be a $(\ell,K)$-iterator (resp. $(r,K)$-iterator);
\item let
$
\displaystyle m_0 \left(\prod_{i=1}^x e_i \cro{u_i} e_i m_i\right)
$
be an $(x,K)$-iterator with $ \displaystyle e \defined m_0 \left(\prod_{i=1}^x e_i m_i\right)$
idempotent;
\item choose $1 \le j \le x$, and define:
\begin{equation*}
        \lefe_j \defined e \left(\prod_{i=1}^{j} m_{i{-1}}  e_i\right) 
        \substack{\tnorm{~~and~~}}
        \rige_j \defined \left(\prod_{i=j}^{x} e_i m_i\right) e
\end{equation*}
\end{itemize}
then we have:
\begin{equation*}
\begin{aligned}
&\rod{\mach}{}{\ce~e m_0\left(\prod_{i=1}^x e_i \cro{u_i} e_i m_i\right) e~\de}\\
&= \rod{\mach}{}{\ce~em_0\left(\prod_{i=1}^{j-1} e_i \cro{u_i} e_i m_i\right) 
e_j m_j  \left(\prod_{i=j+1}^{x} e_i \cro{u_i} e_i m_i\right)  e
\left( \lefe_{j} \cro{u_{j}} \rige_{j} \right)~\de}. 
\end{aligned}
\end{equation*}

\end{lemma}

Proposition~\ref{prop:factor-perm} is obtained by
 induction on $x$ from Lemma~\ref{lem:ecarte}. The proof is depicted in
 Figure \ref{fig:ecartee}  for $x=3$ and $\sigma = (3,1,2)$.
 Since $u_2$ has to be the last element after substitution,
 we first apply Lemma~\ref{lem:ecarte} with $j=2$ to send it ``on the right'',
 then we do the same with $u_1$.

\begin{figure}[h!]

     \begin{subfigure}[b]{1\textwidth}
     \centering
     \begin{tikzpicture}{scale=1}
	\node[above] at (0,-0.25) {$\ce$};
	\node[above] at (11,-0.25) {$\ce'$};	
	
	\draw[-,very thick,black](0.5,0) -- (10.5,0);
	
	\draw[-,very thick,black](0.5,0) -- (0.5,0.15);
	\draw[-,very thick,black](1,0) -- (1,0.15);
	\draw[-,very thick,black](2,0) -- (2,0.15);

	\draw[-,very thick,black](3,0) -- (3,0.15);
	\draw[-,very thick,black](5,0) -- (5,0.15);

	\draw[-,very thick,black](6,0) -- (6,0.15);
	\draw[-,very thick,black](8,0) -- (8,0.15);

	\draw[-,very thick,black](9,0) -- (9,0.15);
	\draw[-,very thick,black](10,0) -- (10,0.15);
	\draw[-,very thick,black](10.5,0) -- (10.5,0.15);

	\draw[very thick,\myBlue, fill = \myBlue] (2,0) rectangle (3,-0.1);
	\draw[very thick,\myBlue, fill = \myBlue] (5,0) rectangle (6,-0.1);	
	\draw[very thick,\myBlue, fill = \myBlue] (8,0) rectangle (9,-0.1);
	
	\node[above] at (0.75,0) {\small $e$};	
	\node[above] at (1.5,0) {\small $m_0e_1$};		
	\node[above] at (2.5,0) {\small $e_1$};	
	\node[below] at (2.5,-0.1) {\small $\cro{u_1}$};	
	\node[above] at (4,0) {\small $e_1m_1e_2$};	
	\node[above] at (5.5,0) {\small $e_2$};	
	\node[below] at (5.5,-0.1) {\small $\cro{u_2}$};	
	\node[above] at (7,0) {\small $e_2m_2e_3$};
	\node[above] at (8.5,0) {\small $e_3$};	
	\node[below] at (8.5,-0.1) {\small $\cro{u_3}$};	
	\node[above] at (9.5,0) {\small $e_3m_3$};
	\node[above] at (10.25,0) {\small $e$};
	
	\draw[thick, dashed,purple] (1,0.75)--(10,0.75);
	\draw[thick,purple] (1,0.75)--(1,0.5);
	\draw[thick,purple] (10,0.75)--(10,0.5);
	\node[above,purple] at (5.5,0.75) {\small $=e$};
	
	\draw[thick, dashed,purple] (0.5,-0.75)--(5,-0.75);
	\draw[thick,purple] (0.5,-0.75)--(0.5,-0.5);
	\draw[thick,purple] (5,-0.75)--(5,-0.5);
	\node[below,purple] at (2.75,-0.75) {\small $= \lefe_2$};

	\draw[thick,dashed,purple] (6,-0.75)--(10.5,-0.75);
	\draw[thick,purple] (6,-0.75)--(6,-0.5);
	\draw[thick,purple] (10.5,-0.75)--(10.5,-0.5);
	\node[below,purple] at (8.25,-0.75) {\small $= \rige_2$};
	
\end{tikzpicture}
\caption{Initial production with $x=3$}
     \end{subfigure}

     \begin{subfigure}[b]{1\textwidth}
     \centering
     \begin{tikzpicture}{scale=1}
	\node[above] at (0,-0.25) {$\ce$};
	\node[above] at (11.5,-0.25) {$\ce'$};	
	
	\draw[-,very thick,black](0.5,0) -- (7.5,0);
	
	\draw[-,very thick,black](0.5,0) -- (0.5,0.15);
	\draw[-,very thick,black](1,0) -- (1,0.15);
	\draw[-,very thick,black](2,0) -- (2,0.15);

	\draw[-,very thick,black](3,0) -- (3,0.15);
	\draw[-,very thick,black](5,0) -- (5,0.15);

	\draw[-,very thick,black](6,0) -- (6,0.15);
	\draw[-,very thick,black](7,0) -- (7,0.15);
	\draw[-,very thick,black](7.5,0) -- (7.5,0.15);
	
	\draw[-,very thick,black](8,0) -- (11,0);

	\draw[-,very thick,black](8,0) -- (8,0.15);
	\draw[-,very thick,black](9,0) -- (9,0.15);
	\draw[-,very thick,black](10,0) -- (10,0.15);
	\draw[-,very thick,black](11,0) -- (11,0.15);

	\draw[very thick,\myBlue, fill = \myBlue] (2,0) rectangle (3,-0.1);
	\draw[very thick,\myBlue, fill = \myBlue] (5,0) rectangle (6,-0.1);	
	\draw[very thick,\myBlue, fill = \myBlue] (9,0) rectangle (10,-0.1);
	
	\node[above] at (0.75,0) {\small $e$};	
	\node[above] at (1.5,0) {\small $m_0e_1$};		
	\node[above] at (2.5,0) {\small $e_1$};	
	\node[below] at (2.5,-0.1) {\small $\cro{u_1}$};	
	\node[above] at (4,0) {\small $e_1m_1e_2m_2e_3$};	
	\node[above] at (5.5,0) {\small $e_3$};	
	\node[below] at (5.5,-0.1) {\small $\cro{u_3}$};
	\node[above] at (6.5,0) {\small $e_3m_3$};
	\node[above] at (7.25,0) {\small $e$};

	\node[above] at (8.5,0) {\small $\lefe_2$};		
	\node[above] at (9.5,0) {\small $e_2$};	
	\node[below] at (9.5,-0.1) {\small $\cro{u_2}$};	
	\node[above] at (10.5,0) {\small $\rige_2$};	
	
	\draw[thick, dashed,purple] (1,0.75)--(7,0.75);
	\draw[thick,purple] (1,0.75)--(1,0.5);
	\draw[thick,purple] (7,0.75)--(7,0.5);
	\node[above,purple] at (4,0.75) {\small $=e$};
	
	\draw[thick, dashed,purple] (0.5,-0.75)--(2,-0.75);
	\draw[thick,purple] (0.5,-0.75)--(0.5,-0.5);
	\draw[thick,purple] (2,-0.75)--(2,-0.5);
	\node[below,purple] at (1.25,-0.75) {\small $= \lefe_1$};

	\draw[thick, dashed,purple] (3,-0.75)--(7.5,-0.75);
	\draw[thick,purple] (7.5,-0.75)--(7.5,-0.5);
	\draw[thick,purple] (3,-0.75)--(3,-0.5);
	\node[below,purple] at (5.25,-0.75) {\small $= \rige_1$};
	
\end{tikzpicture}
\caption{Production after applying Lemma \ref{lem:ecarte} once with $x=3$}
     \end{subfigure}

     \begin{subfigure}[b]{1\textwidth}
     \centering
     \begin{tikzpicture}{scale=1}
	\node[above] at (-1.5,-0.25) {$\ce$};
	\node[above] at (12,-0.25) {$\ce'$};	
	
	\draw[-,very thick,black](-1,0) -- (4.5,0);

	\draw[-,very thick,black](-1,0) -- (-1,0.15);
	\draw[-,very thick,black](-0.5,0) -- (-0.5,0.15);
	\draw[-,very thick,black](2,0) -- (2,0.15);
	\draw[-,very thick,black](3,0) -- (3,0.15);
	\draw[-,very thick,black](4,0) -- (4,0.15);
	\draw[-,very thick,black](4.5,0) -- (4.5,0.15);
	
	\draw[-,very thick,black](5,0) -- (8,0);
	
	\draw[-,very thick,black](5,0) -- (5,0.15);
	\draw[-,very thick,black](6,0) -- (6,0.15);
	\draw[-,very thick,black](7,0) -- (7,0.15);
	\draw[-,very thick,black](8,0) -- (8,0.15);
	
	\draw[-,very thick,black](8.5,0) -- (11.5,0);

	\draw[-,very thick,black](8.5,0) -- (8.5,0.15);
	\draw[-,very thick,black](9.5,0) -- (9.5,0.15);
	\draw[-,very thick,black](10.5,0) -- (10.5,0.15);
	\draw[-,very thick,black](11.5,0) -- (11.5,0.15);

	\draw[very thick,\myBlue, fill = \myBlue] (2,0) rectangle (3,-0.1);
	\draw[very thick,\myBlue, fill = \myBlue] (6,0) rectangle (7,-0.1);	
	\draw[very thick,\myBlue, fill = \myBlue] (9.5,0) rectangle (10.5,-0.1);

	\node[above] at (-0.75,0) {\small $e$};	
	\node[above] at (0.75,0) {\small $m_0e_1m_1e_2m_2e_3$};

	\node[above] at (2.5,0) {\small $e_3$};	
	\node[below] at (2.5,-0.1) {\small $\cro{u_3}$};	
	\node[above] at (3.5,0) {\small $e_3m_3$};	
	\node[above] at (4.25,0) {\small $e$};	

	\node[above] at (5.5,0) {\small $\lefe_1$};	
	\node[below] at (6.5,-0.1) {\small $\cro{u_1}$};
	\node[above] at (6.5,0) {\small $e_1$};
	\node[above] at (7.5,0) {\small $\rige_1$};

	\node[above] at (9,0) {\small $\lefe_2$};		
	\node[above] at (10,0) {\small $e_2$};	
	\node[below] at (10,-0.1) {\small $\cro{u_2}$};	
	\node[above] at (11,0) {\small $\rige_2$};	
	
	\draw[thick, dashed,purple] (-1,-0.75)--(2,-0.75);
	\draw[thick,purple] (-1,-0.75)--(-1,-0.5);
	\draw[thick,purple] (2,-0.75)--(2,-0.5);
	\node[below,purple] at (0.5,-0.75) {\small $= \lefe_3$};

	\draw[thick, dashed,purple] (3,-0.75)--(4.5,-0.75);
	\draw[thick,purple] (4.5,-0.75)--(4.5,-0.5);
	\draw[thick,purple] (3,-0.75)--(3,-0.5);
	\node[below,purple] at (3.75,-0.75) {\small $= \rige_3$};
	
\end{tikzpicture}
         \caption{Production after applying Lemma \ref{lem:ecarte} once again with $x=2$}
     \end{subfigure}

\caption{\label{fig:ecartee} Proof idea for Proposition~\ref{prop:factor-perm}
with $x=3$ and  $\sigma = (3,1,2)$}
\end{figure}
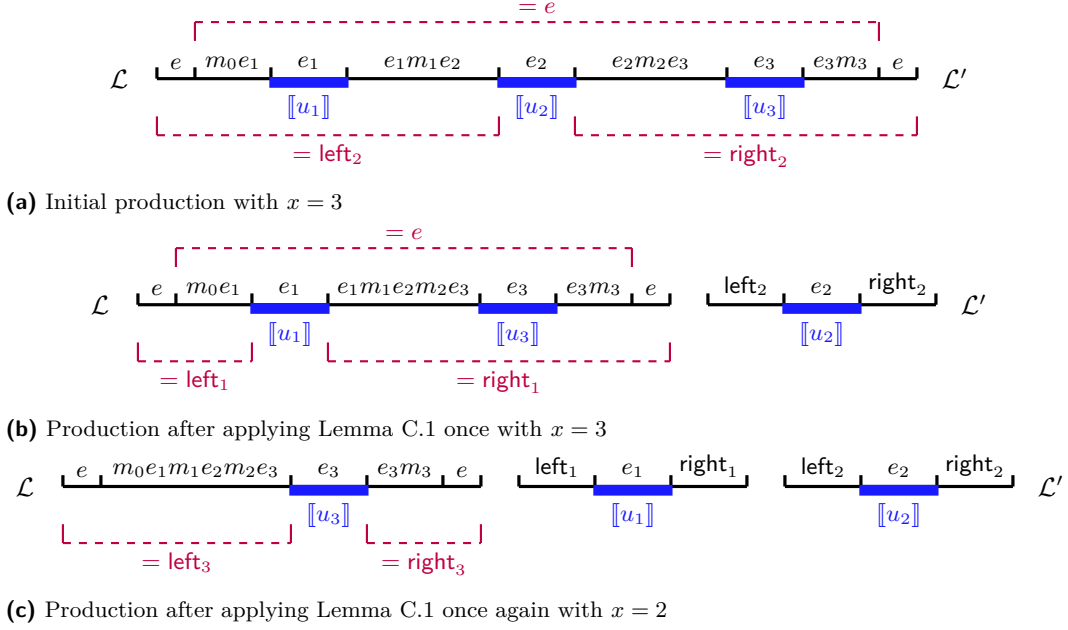

\subparagraph*{Proof of Lemma~\ref{lem:ecarte}.}
The rest of this section is devoted to the proof of Lemma~\ref{lem:ecarte}. The idea is to build
a word in which the two productions compared in
Lemma~\ref{lem:ecarte} occur. Then, by iterating some factors in this word,
we use the $x$-repetitiveness (since $x \le k$) of $f$ to show
that these  productions must be equal. Let $\omega$
be given by Definition \ref{def:repetitive}, assume that $\omega \ge 3$.

Assume that:
\begin{equation*}
\ce = p_0 \left( \prod_{i=1}^{\ell} f_i \cro{v_i} f_i p_i \right)
\substack{\tnorm{~~~and~~~}}
\de = p'_0 \left( \prod_{i=1}^{r} f'_i \cro{v'_i} f'_i p'_i \right)
\end{equation*}
For every $m \in M$ in these notations, we
denote $\pre{m} \in A^*$ a word such that $\mu(\pre{m}) = m$
(it exists since $\mu$ is supposed surjective). We then define:
\begin{itemize}
\item for $\overline{L} \defined L_1, \dots, L_{\ell} \ge 0$, $ \displaystyle U(\overline{L} )  \defined \pre{p_0} \left(\prod_{i=1}^{\ell} v_i^{L_i} \pre{p_i}\right)$ {;}
\item for $\overline{X} \defined  X_1, \dots, X_x \ge 0$, $\displaystyle W(\overline{X} )  \defined
\pre{m_0} \left(\prod_{i=1}^{x} u_i^{\omega X_i} \pre{m_i}\right)${;}
\item for $\overline{R} \defined R_1, \dots, R_{r} \ge 0$ , $\displaystyle V(\overline{R}) 
    \defined \pre{p'_0} \left(\prod_{i=1}^{r} {v'_i}^{R_i} \pre{p'_i}\right)$.
\end{itemize}
Let $w \defined W(1, \dots, 1)$, note that  for $\overline{X} \ge 1$, $\mu(W(\overline{X})) = \mu(w) = e$.
We define:
\begin{equation*}
\begin{array}{lcll}
P: &  \Nat^{k+1} & \rightarrow  & \Nat \\
 &  (\ov{L}, \ov{X}, X'_j, \ov{R}) &
 \mapsto & f\left(U(\overline{L}) w^{2\omega{-}1} W(\overline{X}) w^{\omega-1}
W(3, \dots, 3, X'_j, 3, \dots, 3) w^{\omega} V(\overline{R})\right)
 \end{array}
\end{equation*}
where $X'_j$ is in position $j$ of $W(3, \dots, 3, X'_j, 3, \dots, 3) $.

Let $T \defined L_1 \cdots L_{\ell} X_1 \cdots X_x R_1 \cdots R_{r} / X_j$.
By applying Sublemma~\ref{slem:bim-big} stated in Section~\ref{sec:bim-big} (it does not
use the fact that $f$ is $k$-repetitive) and the definitions 
of our words, we get: 

\begin{claim} \label{claim:c-c}
For  $\overline{L}, \overline{X}, X'_j, \overline{R} \ge 2k+1$,
$P(\overline{L}, \overline{X}, X'_j, \overline{R} )$
is a polynomial and:
\begin{itemize}
\item the coefficient  in $TX_j$ of $P$  is:
\begin{equation*}
c \defined \rod{\mach}{}{\ce~e m_0\left(\prod_{i=1}^x e_i \cro{u_i} e_i m_i\right) e~\de}{;}
\end{equation*}
\item the coefficient in  $TX'_j$ of $P$  is:
\begin{equation*}
c' \defined \rod{\mach}{}{\ce~em_0\left(\prod_{i=1}^{j-1} e_i \cro{u_i} e_i m_i\right) 
e_j m_j  \left(\prod_{i=j+1}^{x} e_i \cro{u_i} e_i m_i\right)  e
\left( \lefe_{j} \cro{u_{j}} \rige_{j} \right)~\de}. 
\end{equation*}
\end{itemize}
\end{claim}

Since $f$ is $k$-repetitive, we also get the following
from our construction.

\begin{claim} \label{claim:grand-F}
There exists a function $F: \Nat^k \rightarrow \Nat$
such that for $\overline{L}, \overline{X}, X'_j, \overline{R} \ge 2k+1$:
\begin{equation*}
P(\overline{L}, \overline{X}, X'_j, \overline{R})
= F(\overline{L}, X_1, \dots, X_{j-1},  X_{j+1}, \dots, X_x, \overline{R}, X_j + X'_j).
\end{equation*}
\end{claim}

\begin{claimproof} The function $f$ is $x$-repetitive
since it is $k$-repetitive and $x \le k$.
Let $\overline{L}$ and $\overline{R}$ be fixed,
using the definition of $P$ and Definition \ref{def:repetitive}
(with $\alpha = U(\overline{L})$ and $\beta = V(\overline{R})$)
we have that $P$ only depends on
$X_1 {+} 3, \dots, X_{j-1} {+} 3, X_j {+} X'_j, X_{j+1} {+} 3 , \dots, X_x {+}3$.
\end{claimproof}

From Claim \ref{claim:grand-F} we deduce that
for $\overline{L}, \overline{X}, X'_j, \overline{R}  \ge 2k+1$:
\begin{equation}
\label{eq:polyly}
\begin{aligned}
&P(\overline{L}, X_1, \dots, X_{j-1},X_j, X_{j+1}, \dots, X_x, X'_{j}, \overline{R})\\
 =~&P(\overline{L}, X_1, \dots, X_{j-1},2k{+}1, X_{j+1}, \dots, X_x, X'_j {+} X_j {-} (2k{+}1), \overline{R}).\\
 \end{aligned}
\end{equation}

By developing the second expression in Equation \ref{eq:polyly},
it is easy to see that the coefficients 
in $TX_j$ and in $TX'_j$ of $P$ are equal. Hence $c=c'$ in Claim \ref{claim:c-c}.

\subsection{Statement and proof of Sublemma~\ref{slem:bim-big}}

\label{sec:bim-big}

Intuitively, the result below says that if we iterate
$k{+}1$ idempotent factors in a word, then the output 
will be a multivariate polynomial in the number of iterations.
Furthermore, specific terms of the polynomial
allow to recover the productions over these factors.

\begin{sublemma} \label{slem:bim-big} Let $\mach = (A,M, \mu, \oras, \lambda)$ be a $k$-marble bimachine.\\
Let  $\alpha_0 u_1 \alpha_1 \cdots u_{k+1} \alpha_{k+1} {\in} A^*$.
Let $m_i \defined \mu(\alpha_i)$, $e_i \defined \mu(u_i)$
and assume that the $e_i$ are idempotent. Then
$
P:X_1, \dots, X_{k+1} \mapsto
f\left( \alpha_0 u_1^{X_1} \alpha_1 \cdots \alpha_k u_{k+1}^{X_{k+1}} \alpha_{k+1} \right)
$
is a polynomial for $X_1, \dots, X_{k+1} \ge 2k+1$.
For $1 \le j \le k+1$, the coefficient of $P$ in ${X_1 \cdots X_{k+1}}/{X_j}$ is:
\begin{equation*}
\rod{\mach}{}{m_0 \left(\prod_{i=1}^{j{-}1} e_i \cro{u_i} e_i m_i  \right) e_j m_j
\left(\prod_{i=j+1}^{k+1} e_i \cro{u_i} e_i m_i \right)}.
\end{equation*}

\end{sublemma}

We shall prove a stronger result by induction.

\subparagraph*{Induction hypothesis.}
Given $1 \le \ell \le k+1$, let
$\beta(X_1, \dots, X_{\ell}) \defined \alpha_0 u_1^{X_1} \cdots {u_{\ell}}^{X_{\ell}} \alpha_{\ell}$.
The induction hypothesis on $\ell$ states that
for all $v_1, \dots, v_n \in A^*$, for all
$(\delta_1, \dots, \delta_n,r) \in \Sigma_k$, if $\ce~= \cro{v_1}_{\delta_1} \cdots \cro{v_n}_{\delta_n}$,
then
$
X_1, \dots, X_{\ell} \mapsto \rod{\mach}{}{\cro{\beta(X_1, \dots, X_{\ell})}_r~\ce}
$\\
is a polynomial for $X_1, \dots, X_{\ell} \ge 2k+1$. Furthermore:
\begin{enumerate}
\item if $\ell \ge 1$, the coefficient in  $X_1 \dots X_{\ell}$ is:
\begin{enumerate}
\item  $0$ if $r < \ell$;
\item $\displaystyle \rod{\mach}{}{m_0 \left(\prod_{i=1}^{\ell} e_i \cro{u_i} e_i m_i  \right)~\ce}$
if  $r = \ell$;
\end{enumerate}
\item if $\ell \ge 2$ and $1 \le j \le \ell$, the coefficient in  ${X_1 \dots X_{\ell}}/{X_j}$ is:
\begin{enumerate}
\item $0$ if $r < \ell {-} 1$;
\item $\displaystyle \rod{\mach}{}{m_0 \left(\prod_{i=1}^{j{-}1} e_i \cro{u_i} e_i m_i  \right) e_j m_j
\left(\prod_{i=j+1}^{\ell} \left(e_i \cro{u_i} e_i m_i \right)\right)\ce}$
 if $r = \ell {-} 1$.
\end{enumerate}
\end{enumerate}

\subparagraph*{Proof by induction.}
Let $\ell \ge 1$ and
assume that the result holds for $\ell{-}1$
(define  $\beta() \defined \alpha_0$, then the
result is true by emptiness for $\ell = 0$).
Then:
\begin{equation}
\label{eq:prod-beta}
\begin{aligned}
&\rod{\mach}{}{\cro{\beta(X_1, \dots, X_{\ell})}_r~\ce}
 = \rod{\mach}{}{\cro{\beta(X_1, \dots, X_{\ell{-}1}) u_{\ell}^{X_{\ell}} \alpha_{\ell}}_r~\ce} \\
& = \sum_{(r_1,r_2,r_3) \in \Sigma_r}  
\rod{\mach}{}{\cro{\beta(X_1, \dots, X_{\ell{-}1})}_{r_1}
\cro{u_{\ell}^{X_{\ell}}}_{r_2} \cro{\alpha_{\ell}}_{r_3}~\ce~}
\substack{\tnorm{~~by Lemma~\ref{lem:coupe-factors}.}}
\end{aligned}
\end{equation}

We now want to cut the factor $u_{\ell}^{X_{\ell}}$ in
several pieces, for this we use item~\ref{po:u-x:2} of Lemma~\ref{lem:u-x}.
It follows from Equation \ref{eq:prod-beta} that for all $X_{\ell} \ge 2k+1$:
\begin{equation}
\label{eq:big-sum}
\begin{aligned}
&\rod{\mach}{}{\cro{\beta(X_1, \dots, X_{\ell})}_r~\ce}\\
& = \sum_{(r_1,r_2,r_3) \in \Sigma_r}  \sum_{s = (s_1, \dots, s_y) \in \Shapes_{r_2}} \Big(P_s(X_{\ell})\\
&\times \rod{\mach}{}{\cro{\beta(X_1, \dots, X_{\ell{-}1})}_{r_1} \cro{u_{\ell}}_{s_1} \cdots \cro{u_{\ell}}_{s_Y}  \cro{\alpha_{\ell}}_{r_3}~\ce~} \Big)
\end{aligned}
\end{equation}

Each last term is by induction hypothesis a polynomial in $X_1, \dots, X_{\ell-1}$
if they are $\ge 2k+1$. Since $P_s$ only depends on $s$ by Lemma~\ref{lem:u-x}, we conclude that
$\rod{\mach}{}{\cro{\beta(X_1, \dots, X_{\ell})}_r~\ce}$ is a polynomial in $X_1, \dots, X_{\ell}$.
Now let us treat the specific cases:
\begin{enumerate}
\item consider the coefficient in
$X_1 \cdots X_{\ell}$ in $\rod{\mach}{}{\cro{\beta(X_1, \dots, X_{\ell})}_r~\ce}$:
\begin{enumerate}
\item if $r <\ell$, then by induction the only term in the first sum
of Equation \ref{eq:big-sum} which contains a factor in $X_1 \cdots X_{\ell-1}$ with a
possibly non-zero coefficient is for $r_1 = r$ and $r_2 = r_3 = 0$,
but then we have no $X_{\ell}$. Hence the coefficient in $X_1 \cdots X_{\ell}$ is $0$;
\item if $r=\ell$, then by induction hypothesis three terms can provide a factor in
$X_1 \dots X_{\ell}$:
\begin{itemize}
\item $r_1 = \ell$, $r_2 = 0$, $r_3 = 0$ or
$r_1 = \ell{-}1$, $r_2 = 0$, $r_3 = 1$, but then we
also have no $X_{\ell}$;
\item $r_1 = \ell{-}1$, $r_2 = 1$ and $r_3 = 0$. Then by
using  item~\ref{po:u-x:3} of Lemma~\ref{lem:u-x} and induction hypothesis,
the coefficient in $X_1 \dots X_{\ell}$  has the good shape;
\end{itemize}
\end{enumerate}
\item now let $\ell \ge 2$ and $1 \le j \le \ell$, and consider
the coefficient in  ${X_1 \dots X_{\ell}}/{X_j}$. We assume that
of $j < \ell$ and $\ell \ge 3$ (the other cases can be treated
using similar techniques):
\begin{enumerate}
\item if $r < \ell {-}1$, then  by induction the only term in Equation \ref{eq:big-sum}
which contains a factor in ${X_1 \cdots X_{\ell-1}}/{X_j}$ with a
possibly non-zero coefficient is for $r_1 = r{-1}$ and $r_2 = r_3 = 0$,
but then we have no $X_{\ell}$. Hence the coefficient in ${X_1 \cdots X_{\ell}}/{X_j}$ is $0$;
\item if $r=\ell{-}1$,
then by induction hypothesis three terms can provide a factor in
$X_1 \dots X_{\ell}$:
\begin{itemize}
\item $r_1 = \ell{-}1$, $r_2 = 0$, $r_3 = 0$
or $r_1 = \ell{-}2$, $r_2 = 0$, $r_3 = 1$, but then we
also no $X_{\ell}$;
\item $r_1 = \ell{-}2$, $r_2 = 1$ and $r_3 = 0$. Then by
using  item~\ref{po:u-x:3} of Lemma~\ref{lem:u-x} and induction hypothesis,
the coefficient in $X_1 \dots X_{\ell}$ has the good shape.
\end{itemize}
\end{enumerate}
\end{enumerate}

\section{Proof of Lemma \ref{lem:sumd}}

Given a $k$-marble bimachine $\mach = (A,M, \mu, \oras, \lambda)$,
we want to show that the following function is polyregular and
has growth at most $k{-}1$:
\begin{equation*}
\sumd{\mach}: (\apar{A})^* \rightarrow \Nat, \forest \mapsto
\left\{
    \begin{array}{l}
        \displaystyle \sum_{\Nod \in \Deps{\forest}}
        \rod{\mach}{\forest}{\Nod}
	\substack{\tnorm{{}~if } \forest \in \Facts{\mu}{\hei}{w} \tnorm{{} for some } w \in A^*;} \\
        0 \substack{\tnorm{{} otherwise.}}\\
    \end{array}
\right.
\end{equation*}
We show the two parts of this statement separately.
We first recall a result from \cite{doueneau2021pebble}.

\begin{lemma}[\cite{doueneau2021pebble}] \label{lem:partition}
If $\mu:A^* \rightarrow M$, $w \in A^*$ and  $\forest \in \Facts{\mu}{}{w}$, then:\\
$ \{\fr{\nod}{\forest}: \nod \in \itera{\forest} \cup \{\forest\}\}$
is a partition of $[1{:} |w|]$.
\end{lemma}

\subsection{Proof that $\sumd{\mach}(\forest) = \mc{O}(|\forest|^{k{-1}})$}

\label{ssec:bound}

We first show that each term of the sum defining $\sumd{\mach}$
is bounded independently from $w$, $\forest$ and $\Nod$.
This is the following sublemma.

\begin{sublemma} If $\mach = (A,M, \mu, \oras, \lambda)$ is a $k$-marble
bimachine, there exists $K \ge 0$ such that
for all $w \in A^*$, $\forest \in \Facts{\mu}{\hei}{w}$,
for all $\Nod \subseteq \itera{\forest}\cup \{\forest\}$,
we have
$
\rod{\mach}{\forest}{\Nod} \le K.
$
\end{sublemma}

\begin{proof} Let us assume that $ \Nod = \multi{\nod_1 \con r_1, \dots, \nod_n \con r_n}$.
Then the frontiers of the $\nod_i$ are disjoint by Lemma \ref{lem:partition},
and by definition we have:
\begin{equation}
\label{eq:bor}
\begin{aligned}
 &\rod{\mach}{\forest}{\Nod}
  = \sum_{\substack{\multi{i^1_1, \dots, i^1_{r_1}} \subseteq \fr{\nod_1}{\forest}}}
\cdots \sum_{\substack{\multi{i^n_1, \dots, i^n_{r_n}} \subseteq  \fr{\nod_n}{\forest}}}
\rod{\mach}{w}{\multi{i^1_1, \dots, i^n_{r_n}}}\\
 \end{aligned}
\end{equation}

Note that for $\nod \in \Nodes{\forest}$, we have
$|\fr{\nod}{\forest}{}| \le 2^{\hei}$. Indeed, the frontier are the leaves of 
the skeleton which is a binary tree whose height is bounded by $\hei$
since $\forest \in \Facts{\mu}{\hei}{w}$.
Furthermore, it is easy to see that if $\multi{i_1, \dots, i_k} \subseteq [1{:}|w|]$, one has
$\rod{\mach}{w}{\multi{i_1, \dots, i_k}} \le B$ for some constant
$B$ depending only on $\mach$.
Hence the sums in Equation~\ref{eq:bor} are sums
of bounded terms over a bounded range: they are bounded.
\end{proof}

To conclude the proof, it remains to show
that the number of multisets in $\Deps{\forest}$
grows in $\mc{O}(|\forest|^{k{-}1})$. For this we note
that:
\begin{equation*}
\begin{aligned}
\left|\Deps{\forest}\right|
&\le \left|\{(\nod_1, \dots, \nod_k):\multi{\nod_1, \dots, \nod_k} \in \Deps{\forest} \}\right|\\
&\le \sum_{1 \le i \neq j \le k}
\left|\{(\nod_1, \dots, \nod_k) \in (\Nodes{\forest})^k: \multi{\nod_i, \nod_j} \text{ is not independent}\}\right|\\
&\le 12 |M| k^2 \left|\Nodes{\forest}\right|^{k{-}1} = \mc{O}(|\forest|^{k{-}1})
\end{aligned}
\end{equation*}
Indeed, given $\nod_i \in \itera{\forest} \cup \{\forest\}$, there are
at most $6|M|$ nodes $\nod_j$ whose depth is smaller than $\nod_i$
and such that $\multi{\nod_i, \nod_j}$ is not independent.

\subsection{Proof that $\sumd{\mach}$ is polyregular}

We first note that the set of $\{\forest \in \Facts{\mu}{\hei}{w} : w\in A^*\}$ is regular.
Since regular properties can be checked by pebble transducers
(e.g. by launching a one-way automaton to test whether the property
holds before beginning the computation), we can without loss of generality
assume that our inputs belongs to this set.

Let $w \in A^*$ and $\forest \in \Facts{\mu}{}{w}$.  If
$ 1  \le i \le |w|$, by Lemma~\ref{lem:partition} let $\frm{i}{\forest}$ be the unique
$\nod \in \itera{\forest} \cup \{\forest\}$ such that $i \in \fr{\nod}{\forest}$.
We then describe by Algorithm~\ref{algo:dependent}
the behavior of a $k$-pebble transducer which computes $\sumd{\mach}$.
We show its correctness as follows:
\begin{equation*}
\begin{aligned}
\sumd{\mach}(\forest) &
= \sum_{\multi{\nod_1, \dots, \nod_k} \in \Deps{\forest}}
\sum_{\substack{\multi{i_1, \dots, i_k} \\ \text{with } i_j \in \fr{\nod_j}{\forest}}}
\rod{\mach}{w}{\multi{i_1, \dots, i_k}}\\
&= \sum_{\Nod \in \Deps{\forest}}
\sum_{\substack{\multi{i_1 \le \dots \le i_k} \\ \Nod = \multi{\frm{i_1}{\forest}, \dots, \frm{i_k}{\forest}}}}
\rod{\mach}{w}{\multi{i_1, \dots, i_k}}\\
&= \sum_{1 \le i_1 \le \dots \le i_k \le |w|}
\rod{\mach}{w}{\multi{i_1, \dots, i_k}} \times \indic{\multi{\frm{i_1}{\forest}, \dots, \frm{i_k}{\forest}} \in \Deps{\forest}}\\
\end{aligned}
\end{equation*} 
where $\indic{\cdot \in \cdot}$ denotes the characteristic function.

\begin{algorithm}[h!]
\SetKw{KwVar}{Variables:}
\SetKwProg{Fn}{Function}{}{}
\SetKw{In}{in}
\SetKw{Let}{let}
\SetKw{Out}{Output}
 \Fn{$\sumd{\mach}(\forest)$}{

 		$w \leftarrow \operatorname{\text{word factored by}} \forest$;

		\For{$1 \le i_1 \le \dots \le i_k \le |w|$}{
		
			\Let{$\nod_1, \dots, \nod_k \leftarrow \frm{i_1}{\forest}, \dots, \frm{i_k}{\forest}$}
		
			\If{$\multi{\nod_1, \dots, \nod_k} \in \Deps{\forest}$}{
						
			\Out{$\rod{\mach}{w}{\multi{i_1, \dots, i_k}}$}

			}
			
		}}
	
 \caption{\label{algo:dependent} Computing $\sumd{\mach}$ with a $(k{+}1)$-pebble transducer}
\end{algorithm}

Then we justify that it can be implemented by 
a $k$-pebble transducer. The loop over
$i_1 \le \cdots \le i_k$ is implemented using
$k$ pebbles ranging over the leaves of $\forest$.
Whether $\multi{\nod_1, \dots, \nod_k} \in \Deps{\forest}$
is a regular property which is checked in a standard way
with a lookahead. Finally, the $\rod{\mach}{w}{\multi{i_1, \dots, i_k}}$
are bounded independently from $w$ (see Subsection~\ref{ssec:bound})
and checking if $\rod{\mach}{w}{\multi{i_1, \dots, i_k}} = C$ is also 
a regular property.

\section{Linearizations}

We presented in Defintion~\ref{def:lili-node} the linearization of a node,
which is a $1$-multicontext.
We now generalize it to any set of independent nodes in Definition~\ref{def:lili}.
It is easy to see that coincides on singletons with the former definition.

\begin{definition}[Linearization] \label{def:lili} Let $\mu : A^* \rightarrow M$,
$w \in A^*$ and $\forest \in \Facts{\mu}{}{w}$.
If $\Nod \in \bigcup_{k \ge 0} \Indeps{\forest}$ or $\Nod = \{\forest\}$, we define
its \emph{linearization} by induction:
\item
\begin{itemize}
\item if $|\Nod| = 0$, then $\lin{\forest}{\Nod} \defined \mu(\forest)$;
\item if $\Nod = \set{\forest}$ then
$\lin{\forest}{\Nod} \defined \cro{w[\fr{\forest}{\forest}]}$;
\item otherwise $\forest = \tree{\forest_1} \cdots \tree{\forest_n}$ and $\forest \not \in \Nod$,
then for $1 \le j \le n$ we let $\Nod_j \defined \Nod \cap \Nodes{\forest_j}$.
Then we define
$
\lin{\forest}{\Nod} \defined \lin{\forest_1}{\Nod_1} \cdots  \lin{\forest_n}{\Nod_n}.
$
\end{itemize}
\end{definition}

In the following subsections, we give some properties of linearizations
which shall be useful for the proofs of appendices~\ref{app:depro:prod-archi}
and \ref{app:last}.

\subsection{Productions and linearizations}

We show here that the production on a set of independent nodes
is the same as the production on its linearization.
This result will be extremely useful in order to transform
the permutability of $\mach$ in a property about
the productions on its independent sets of nodes.

\begin{lemma} \label{lem:linear} Let $\mach = (A,M, \mu, \oras, \lambda)$ a $k$-marble bimachine.
Let $w \in A^*$, $\forest \in \Facts{\mu}{}{w}$ and
$\Nod \in \Indeps{\forest}$. Then
$
\rod{\mach}{\forest}{\Nod} = \rod{\mach}{}{\lin{\forest}{\Nod}}.
$
\end{lemma}

The rest of this subsection is devoted to the proof of Lemma~\ref{lem:linear}.

Assume that $\Nod = \set{\nod_1, \dots, \nod_n}$ and for all $1 \le i \le k$
let $I_i \defined \fr{\nod_i}{\forest}$.
Since their is no ancestor-relationship in $\set{\nod_1, \dots, \nod_k}$,
their frontiers do not overlap and
one can assume that $\min(I_1) \le \max(I_1) < \min(I_2) \le \max(I_2) < \cdots < \min(I_k) \le \max(I_k)$
(and they exist since the frontiers are always non-empty).
For $0 \le i \le k$, let $v_i \defined w[\max(I_i){+1} {:} \min(I_{i+1}){-1}]$ (with the convention
that $\max(I_0) = 0$ and $\min(I_{k+1}) = |w|{+}1$).
It it easy to check by construction of the linearization that
$
\lin{\forest}{\Nod} = \mu(v_0) \cro{w[I_1]} \mu(v_1) \cdots \cro{w[I_k]} \mu(v_k).
$

Therefore by Proposition-Definition~\ref{depro:factor-k} it follows:
\begin{equation*}
\begin{aligned}
\rod{\mach}{}{\lin{\forest}{\set{\nod_1, \dots, \nod_k}}}
&= \rod{\mach}{}{v_0 \cro{w[I_1]} v_1 \cdots \cro{w[I_k]} v_k}\\
&= \rod{\mach}{w'}{\set{I'_1, \dots, I'_k}}.\\
\end{aligned}
\end{equation*}
where $w' \defined v_0 w[I_1] v_1 \cdots w[I_k] v_k \in A^*$
and $I'_1, \dots, I'_n \subseteq [1{:}|w'|]$ denote the sets of positions of $w[I_1], \dots, w[I_k]$ in $w'$.
On the other hand one has by definition of productions:
\begin{equation*}
\rod{\mach}{\forest}{\set{\nod_1, \dots, \nod_k}}
= \rod{\mach}{w}{\set{I_1, \dots, I_k}}.
\end{equation*}

We now show that all the terms occurring in the sums defining
$\rod{\mach}{w}{\set{I_1, \dots, I_k}}$ and
$\rod{\mach}{w'}{\set{I'_1, \dots, I'_k}}$ are equal.
For $1 \le i \le k$,  let $\sigma_i : I_i \rightarrow I'_i$
be the unique monotone bijection.
To conclude the  proof, we
show that for all $i_1, \dots, i_k \in I_1, \dots, I_k$, we have:
\begin{equation*}
\rod{\mach}{w}{\set{i_1, \dots, i_k}} = \rod{\mach}{w'}{\set{\sigma_1(i_1), \dots, \sigma_k(i_k)}} 
\end{equation*}
For this it is sufficient to show (see Sublemma \ref{slem:positi}) that:
\begin{equation}
\label{eq:mumu}
\begin{aligned}
&\mu(w[1{:}i_1{-}1]) \cro{w[i_1]}  \cdots \cro{w[i_k]} \mu(w[i_{k}{+}1{:}|w|])\\
&= \mu(w'[1{:}\sigma_1(i_1){-}1]) \cro{w'[\sigma_1(i_1)]} 
\cdots \cro{w'[\sigma_k(i_k)]} \mu(w'[\sigma_k(i_{k}){+}1{:}|w'|])\\
\end{aligned}
\end{equation}
Equation \ref{eq:mumu} directly follows from Claim \ref{claim:fr-monoid}.

\begin{claim} \label{claim:fr-monoid} Let $u \in A^*$, $\forest \in \facto{\mu}{u}$,
$I \defined \fr{\forest}{\forest}$,  $u' \defined u[I]$ and $\sigma : I \rightarrow [1{:}|u'|]$
be the unique monotone bijection. Then for all $i \in I$ we have:
\begin{equation*}
\mu(u[1{:}i{-}1])\cro{u[i]}\mu(u[i{+}1{:}|u|])
= \mu(u'[1{:}\sigma(i){-}1])\cro{u'[\sigma(i)]}\mu(u'[\sigma(i){+}1{:}|u'|]).
\end{equation*}
\end{claim}

\begin{remark}
\label{rem:morphism-fr}
This result implies in particular that $\mu(u) = \mu(u') = \mu(u[\fr{\forest}{\forest}])$.
\end{remark}

\begin{claimproof} 
We show the result by induction on the factorization as follows:
\begin{itemize}
\item if $\forest = a \in A$, one has $u = u' = a$ and it is obvious;
\item if $\forest = \tree{\forest_1} \cdots \tree{\forest_n}$ and $u = u_1 \cdots u_n$
is the according factorization. Then we have
$u' = u'_1 u'_n$ where $u'_j \defined u_j[\fr{\forest_j}{\forest_j}]$ for $j \in \{1,n\}$.
We suppose that $n \ge 3$ (the other case is easier) hence
$\mu(u_1) = \cdots = \mu(u_n)$ is an idempotent.
Assume by symmetry that $i \in \fr{\forest}{\forest_1}$,
then $1 \le \sigma(i) \le |u'_1|$
and $\sigma$ restricted to $[1,|u_1|]$ is
the monotone bijection between $\fr{\forest}{\forest_1}$ and $[1,|u'_1|]$.
We have:
\begin{equation*}
\begin{aligned}
&\mu(u[1{:}i{-}1])\cro{u[i]}\mu(u[i{+}1{:}|u|])\\
&= \mu(u[1{:}i{-}1]\cro{u_1[i]}\mu(u_1[i{+}1{:}|u_1|]) \mu(u_2) \cdots \mu(u_n)\\
&= \mu(u_1[1{:}i{-}1]\cro{u_1[i]}\mu(u_1[i{+}1{:}|u_1|]) \mu(u_n)
\substack{\tnorm{~~by idempotence of $\mu(u_n)$}}\\
&= \mu(u'_1[1{:}\sigma(i){-}1]
\cro{u'_1[\sigma(i)]}
\mu(u'_1[\sigma(i){+}1{:}|u'_1|]) \mu(u'_n)
\substack{\tnorm{~~by induction hypothesis}}\\
&= \mu(u'[1{:}\sigma(i){-}1])
\cro{u'[\sigma(i)]}
\mu(u'[\sigma(i){+}1{:}|u'|])\\
\end{aligned}
\end{equation*}
\end{itemize}
which concludes the proof.
\end{claimproof}

\subsection{Linearizations are iterators}

Here we show that the linearization of an independent
set of nodes is an iterator. Furthermore, the linearizations
of its elements can (up to re-ordering them) be recovered from it.

\begin{lemma}\label{lem:indep-lin} Let $\mu : A^* \rightarrow M$
and $k, K \ge 0$. Let $w \in A^*$, $\forest \in \Facts{\mu}{K}{w}$ and
$\Nod = \set{\nod_1, \dots, \nod_k} \in \Indep^{k}(\forest)$, then:
\begin{enumerate}
\item \label{po:il1} $\lin{\forest}{\Nod}$ is a $(k,2^{K})$-iterator
of the form
$
m_0 \left( \prod_{i=1}^{k} e_i \cro{u_i} e_i m_i \right)
$;
\item \label{po:il2} $\mu(\forest) = m_0 \left( \prod_{i=1}^{k} e_i  m_i \right)$;
\item \label{po:il3} up to a permutation of $[1{:}k]$ we have for all $1 \le j \le k$:
\begin{equation*}
\lin{\forest}{\nod_j} = m_0 \left(\prod_{i=1}^{j-1} e_i m_{i}  \right)
e_j \cro{u_j} e_j \left(\prod_{j=i+1}^{k} m_{i-1} e_i \right) m_{k}.
\end{equation*}
\end{enumerate}
\end{lemma}

\begin{proof} The proof is done by induction.
For $k = 0$, then $\lin{\forest}{\varnothing} = \mu(\forest)$
and the result is obvious.
Otherwise we have $\forest = \tree{\forest_1} \cdots \tree{\forest_n}$.
For all $1 \le j \le n$ let $\Nod_j \defined \Nod \cap \Nodes{\forest_j}$,
so that we have
$
\lin{\forest}{\Nod} = \lin{\forest_1}{\Nod_1} \cdots  \lin{\forest_n}{\Nod_n}.
$

For the $j$ such that $\Nod_j \neq \{\forest_j\}$, then $\lin{\forest_j}{\Nod_j}$
is $(|\Nod_j|, 2^K)$-iterator by induction hypothesis
(since in that case we have $\Nod_j \in \Indep^{|\Nod_j|}(\forest_j)$).
Now if  $\Nod_j = \{\forest_j\}$, we must have $1 < j < n$ (since this node
is iterable) and $\Nod_{j{-}1} = \Nod_{j+1} = \varnothing$
(by definition of independent multisets). Hence we have:
$
\lin{\forest_{j-1}}{\Nod_{j-1}} \lin{\forest_{j}}{\Nod_{j}} \lin{\forest_{j+1}}{\Nod_{j+1}}
= e \cro{w[\fr{\forest_j}{\forest}]}e
$
$= \lin{\forest}{\Nod_j}$
where $e = \mu(\forest)$ is idempotent since $n \ge 3$ and $e  = \mu(w[\fr{\forest}{\forest_j}])$
by Remark \ref{rem:morphism-fr}. Furthermore, by definition
of the frontier we have $|\fr{\forest}{\forest_j}| \le 2^{K}$.
\end{proof}

\section{Proof of Proposition-Definition \ref{depro:prod-archi}}

\label{app:depro:prod-archi}

Let $\mach = (A,M, \mu, \oras, \lambda)$ be a $2^{\hei}$-permutable $k$-marble bimachine.
We show that the production on a set of independent
nodes only depends on its architecture.

\begin{definition}[Rank] Let $w \in A^*$, $\forest \in \Facts{\mu}{\hei}{w}$,
$\Nod \in  \Indep^{x}(\forest)$. We say that the architecture
$\arch{\forest}{\Nod}$ has rank $x$.
\end{definition}

\begin{remark} An architecture has only one rank (it is
sum of the sizes of its multisets).
\end{remark}

We show by induction on the structure of $\ar$ of rank $0 \le x \le k$ the following result:
\begin{itemize}
\item  if $w,w' \in A^*$ with $\forest \in \Facts{\mu}{\hei}{w}$, $\forest' \in \Facts{\mu}{\hei}{w'}$;
\item  if $\Nod \in \Indep^{x}(\forest)$
and $\Nod' \in \Indep^{x}(\forest')$
such that
$
\ar = \arch{\forest}{\Nod}
=  \arch{\forest'}{\Nod'};
$
\item and if $(r, \ell) \in \Sigma_{k{-}x}$ and 
$\ce$ (resp. $\de$) is a $(\ell, 2^{\hei})$-iterator
(resp. $(r,2^{\hei})$-iterator);
\end{itemize}
then we have
$
\rod{\mach}{}{\ce~\lin{\forest}{\Nod}~\de}
= \rod{\mach}{}{\ce~\lin{\forest'}{\Nod'}~\de}.
$

\subsection{Cases for $x=0$}

Three cases are possible. If $\ar=\tree{m}$ (resp. $\ar=a$, resp. $\ar = \vide$),
then $\forest=\tree{\forest_1} \cdots \tree{\forest_n}$
and $\forest'=\tree{\forest'_1} \cdots \tree{\forest'_{n'}}$
with $n,n' \ge 1$ and $\mu(\forest) = \mu(\forest') = m$
(resp. $\forest = \forest'= a$, resp. $\forest = \forest'= \vide$).
Hence we have $\lin{\forest}{\varnothing} = \lin{\forest'}{\varnothing} = m$
(resp. $= \mu(a)$, resp. $ = \mu(\vide)$).

\subsection{Case $x \ge 1$ and $\ar = \tree{\ar_1} \tree{\ar_2} \cdots \tree{\ar_p}$ with $p \ge 1$}

Let us first assume that $\ar_1$ is an architecture of rank $x_1 \ge 1$
(otherwise, we must have that $\ar_p$ has rank $x_p \ge 1$,
and the reasoning is similar).
Let $y \defined x{-}x_1 \ge 0$ be the rank
of $\br \defined \tree{\ar_2} \cdots \tree{\ar_p}$.
The only way to have $\arch{\Nod}{\forest} = \ar$ for
$\Nod \in \Indeps{\forest}$ is that
$\forest = \tree{\forest_1} \tree{\forest_2} \cdots \tree{\forest_n} $ 
with $n \ge 1$. Let $\gorest \defined \tree{\forest_2} \cdots \tree{\forest_n} $
and $\Nod_1 \defined \Nod \cap \Nodes{\forest_1}$. Then we must have
$\arch{\forest_1}{\Nod_1} = \ar_1$ and $\arch{\gorest}{\Nod\smallsetminus\Nod_1} = \br$
(indeed we have $\Nod\smallsetminus\Nod_1 \in \Indep^y(\gorest))$.
Both are iterators by Lemma~\ref{lem:indep-lin}.
It follows from the definition of linearizations that
 $\lin{\forest}{\Nod} = \lin{\forest_1}{\Nod_1} \lin{\gorest}{\Nod\smallsetminus \Nod_1}$.
 Similar results hold
 for $\forest' = \tree{\forest'_1} \gorest'$ and we have:
\begin{equation*}
\begin{aligned}
\rod{\mach}{}{\ce~\lin{\forest}{\nod}~\de} &=
\rod{\mach}{}{\ce~\lin{\forest_1}{\Nod_1} \lin{\gorest}{\Nod\smallsetminus \Nod_1}~\de}\\
&=\rod{\mach}{}{\ce~\lin{\forest'_1}{\Nod'_1} \lin{\gorest'}{\Nod'\smallsetminus \Nod'_1}~\de}
~~\substack{\tnorm{by induction hypothesis} \\ \tnorm{on $\ar_1$ and then on $\br$}}\\
&=\rod{\mach}{}{\ce~\lin{\forest'}{\Nod'}~\de}.\\
\end{aligned}
\end{equation*}

\subsection{Case $x \ge 1$ and $\ar =\tree{\Le}$}

In this case $\ar= \tree{\Le}$ where $\Le$ is a non-empty
multiset of elements of the form $(m\cro{u}m', d)$.
Note that $x = |\Le|$ and let $e$ be the idempotent such
that $e = m\mu(u)m'$
for all $(m\cro{u}m', d) \in \Le$ (by construction and item~\ref{po:il2} of
Lemma~\ref{lem:indep-lin}
it is the same for all the elements occurring in $\Le$).

It follows from the definitions that we must have 
$
\forest = \tree{\forest_1} \cdots \tree{\forest_{n}}
$
and
$
\forest' = \tree{\forest'_1} \cdots \tree{\forest'_{n'}}.
$
with $n,n' \ge 3$ and $\mu(\forest) = \mu(\forest') = e$.
For all $1 \le i \le n$ (resp. $1 \le i \le n'$), let $\Nod_i \defined \Nod \cap \forest_i$
(resp. $\Nod'_i \defined \Nod' \cap \forest'_i$).
By construction of architectures we must
have $\Nod_1 = \Nod_{n} = \Nod'_{1} = \Nod'_{n'} = \varnothing$,
thus $\Nod \defined \biguplus_{i=2}^{n-1} \Nod_i$
and $\Nod' \defined \biguplus_{i=2}^{n'-1} \Nod'_i$.

\begin{claim} \label{claim:pe}
$
\lin{\forest}{\Nod}$
is a $(x, 2^{\hei})$-iterator of the form
$ e m_0 \left( \prod_{i=1}^{x} e_i \cro{u_i} e_i m_i \right)e$\\
such that
$ m_0 \left( \prod_{i=1}^{x} e_i  m_i \right) = e.$
\end{claim}

\begin{claimproof} Since $\Nod \in \Indeps{\forest}$,
by applying Lemma~\ref{lem:indep-lin} we get a $(x, 2^{\hei}$)-iterator, therefore
$
\displaystyle \lin{\forest}{\Nod} = m_0 \left( \prod_{i=1}^{x} e_i \cro{u_i} e_i m_i \right)
$ with
$\displaystyle m_0 \left( \prod_{i=1}^{x} e_i  m_i \right) = \mu(\forest) = e.
$\\
But since $\lin{\forest_1}{\Nod_1} = \mu(\forest_1) = \lin{\forest_n}{\Nod_{n}} = \mu(\forest_n) = e$, then
by definition of linearizations we must have $m_0 e_1 = e m'_0$ and $e_x m_x = m'_x e$.
In particular $m_0 e_1 = e m_0 e_1$ and $e_x m_x = e_x m_x e$ since $e$ is idempotent.
Thus we get
$
m_0 \left( \prod_{i=1}^{x} e_i \cro{u_i} e_i m_i \right)
= em_0 \left( \prod_{i=1}^{x} e_i \cro{u_i} e_i m_i \right)e
$
which concludes the proof.
\end{claimproof}

We get with Claim \ref{claim:pe'} a similar result
for $\forest'$ (note that it uses the same $e$).

\begin{claim} \label{claim:pe'}
$
\lin{\forest'}{\Nod'}$
is a $(k, 2^{\hei})$-iterator of the form
$ e m'_0 \left( \prod_{i=1}^{x} e'_i \cro{u'_i} e'_i m'_i \right)e$\\
such that
$ m'_0 \left( \prod_{i=1}^{x} e_i  m'_i \right) = e.$
\end{claim}

We show that the productions
on  $\lin{\forest}{\Nod}$ and $\lin{\forest'}{\Nod'}$
are equal. For $1 \le j \le x$ define:
\begin{equation*}
\left\{
    \begin{array}{l}
       \displaystyle \lefe_j \defined e \left(\prod_{i=1}^{j} m_{i{-1}}  e_i\right)
       \substack{\tnorm{~~~and~~~}}
       \rige_j \defined \left(\prod_{i=j}^{x} e_i m_i\right) e
       \\
       \displaystyle \lefe'_j \defined e \left(\prod_{i=1}^{j} m'_{i{-1}}  e'_i\right)
       \substack{\tnorm{~~~and~~~}}
       \rige'_j \defined \left(\prod_{i=j}^{x} e'_i m'_i\right) e
    \end{array}
\right.
\end{equation*}

\begin{claim} \label{claim:all}
There exists a permutation $\sigma$
of $[1{:}x]$ such that for all $1 \le i \le x$,
$u'_i = u_{\sigma(i)}$, $\lefe'_i = \lefe_{\sigma(i)}$
and $\rige'_i = \rige_{\sigma(i)}$.
\end{claim}

\begin{claimproof} By definition of $\Le$ we have
$
\left\{\lin{\forest}{\set{\nod}} : \nod \in \Nod\right\}
= \left\{\lin{\forest}{\set{\nod'}} : \nod' \in \Nod'\right\}
$.
Hence by applying item~\ref{po:il3} of Lemma~\ref{lem:indep-lin} twice, we get
the suitable bijection $\sigma$.
\end{claimproof}
Let $(\ell,r) \in \Sigma_{k-x}$, and $\ce$ (resp. $\de$)
be a $(\ell,2^{\hei})$-iterator (resp. $(r,2^{\hei})$-iterator), then:
\begin{equation*}
\begin{aligned}
&\rod{\mach}{}{\ce~\lin{\forest}{\Nod}~\de}\\
&= \rod{\mach}{}{\ce~
\left(e m_0 \left( \prod_{i=1}^{x} e_i \cro{u_i} e_i m_i \right) e\right)
~\de}
~~\substack{\tnorm{by Claim \ref{claim:pe}}}\\
\end{aligned}
\end{equation*}
\begin{equation*}
\begin{aligned}
& = \rod{\mach}{}{\ce~
\left(\prod_{i=1}^{x} \lefe_{\sigma(i)} \cro{u_{\sigma(i)}} \rige_{\sigma(i)}
\right) ~\de}
~~\substack{\tnorm{since $|u_i| \le 2^{\hei}$} \\ \tnorm{and $\mach$ is $2^{\hei}$-permutable}}\\
& = \rod{\mach}{}{\ce~
\left(\prod_{i=1}^{x} \lefe'_{i} \cro{u'_i} \rige'_{i}
\right) ~\de}
~~\substack{\tnorm{by Claim \ref{claim:all}}}\\
\end{aligned}
\end{equation*}
\begin{equation*}
\begin{aligned}
& = \rod{\mach}{}{\ce~
\left(e m'_0 \left( \prod_{i=1}^{x} e'_{i} \cro{u'_{i}} e'_{i} m'_i \right) e\right)
~\de}
~~\substack{\tnorm{since $|u'_i| \le 2^{\hei}$} \\ \tnorm{and $\mach$ is $2^{\hei}$-permutable}}\\
& = \rod{\mach}{}{\ce~\lin{\forest'}{\Nod'}~\de}
~~\substack{\tnorm{by Claim \ref{claim:pe'}.}}\\
\end{aligned}
\end{equation*}

\section{Proof of Lemma \ref{lem:counts}}

\label{app:last}

For $\ar \in \Archs{\mu}{3|M|}$ of rank $k \ge 0$,
we had defined
$
\cou_{\ar}(\forest) \defined
|\{\Nod \in \Indeps{\forest}:
\arch{\forest}{\Nod} = \ar\}|.
$

We show by induction on the inductive structure of $\ar $
that one can build two functions
${\cou'_{\ar}}, {\cou''_{\ar}}: (\apar{A})^* \rightarrow \Nat$
such that $\cou_{\ar} =  {\cou'_{\ar}} +  {\cou''_{\ar}}$ and:
\begin{itemize}
\item ${\cou'_{\ar}}$ is polyblind and has growth
at most $k$;
\item  ${\cou''_{\ar}}$ is polyregular and has growth at most $k{-}1$.
\end{itemize}
The most interesting case is that of $\ar = \tree{\Le}$ treated in Subsection \ref{ssec:Types}.

Once more, it is enough to describe our functions on the set
of inputs $\forest \in (\apar{A})^*$ such that $\forest \in \Facts{\mu}{\hei}{w}$
for some $w \in A^*$. Indeed, this domain is regular.

\subsection{Cases for $k=0$}

Three cases occur, and we treat them in a similar
way. If $\ar=\tree{m}$ (resp. $\ar=a$, resp. $\ar = \vide$),
then $\cou_{\ar}(\forest) = 1$ if
$\forest=\tree{\forest_1} \cdots \tree{\forest_n}$
with $n \ge 1$ and $\mu(\forest) = m$
(resp. $\forest = a$, resp. $\forest = \vide$)
and $0$ otherwise. In this cases the function
$\cou_{\ar}$ is the indicator function
of a regular language, thus it is polyblind with growth at most $0$.
We define $\cou'_{\ar} \defined \cou_{\ar}$
and $\cou''_{\ar} \defined 0$.

\subsection{Case $k \ge 1$ and $\ar = \tree{\ar_1} \tree{\ar_2} \cdots \tree{\ar_p}$ with $p \ge 1$}

We assume that $\ar_1$ is an architecture of rank $k_1 \ge 1$
(otherwise, we must have that $\ar_p$ has rank $k_p \ge 1$,
and the reasoning is similar).
Let $b = k{-}k_1 \ge 0$ be the rank
of $\br \defined \tree{\ar_2} \cdots \tree{\ar_p}$.

\begin{claim} \label{claim:count-ai}
We have for $w \in A^*$ and $\forest \in \Facts{3|M|}{\mu}{w}$:
\begin{equation*}
\cou_{\ar}(\forest) = 
\left\{
    \begin{array}{l}
        0 \text{~~if~~} \forest \text{~~is not of the form~~} \tree{\forest_1} \tree{\forest_2} \cdots \tree{\forest_n} \text{~~with~~} n \ge 1\\
	\cou_{\ar_1}(\forest_1)\times \cou_{\br}(\tree{\forest_2} \cdots \tree{\forest_n}) \text{~~otherwise.}
    \end{array}
\right.
\end{equation*}
\end{claim}

\begin{claimproof}
The only way to get $\arch{\Nod}{\forest} = \tree{\ar_1} \br$ for
$\Nod \in \Indeps{\forest}$ is $\forest = \tree{\forest_1} \tree{\forest_2} \cdots \tree{\forest_n} $ 
with $n \ge 1$ and $\Nod \cap \Nodes{\forest_1} \neq \varnothing$.
If $\forest$ is of this form, let $\gorest \defined \tree{\forest_2} \cdots \tree{\forest_n}$, we have:
\begin{equation*}
\begin{aligned}
&|\{\Nod \in \Indeps{\forest} : \arch{\forest}{\Nod} = \ar\}|\\
&= |\{(\Nod_1, \Nod_2) : \Nod_1 \in  \Indep^{k_1}(\forest_1),
\arch{\forest_1}{\Nod_1} = \ar_1 \text{~and~}
\Nod_2 \in  \Indep^{b}(\gorest), \arch{\gorest}{\Nod_2} = \br\}|.
\\
\end{aligned}
\end{equation*}
Indeed, the function: $\Nod \mapsto (\Nod\cap \Nodes{\forest_1}),\Nod\cap \Nodes{\gorest})$
is a bijection between these two sets. First note that we have
$\Nod\cap \Nodes{\forest_1} \in \Indep^{k_1}(\forest_1)$ (since $\forest_1 \not \in \Nod$).
Furthermore $\forest_2 \not \in \Nod$ since otherwise it would be
the sibling of an ancestor of a node of $\Nod$,
thus $\Nod\cap \Nodes{\gorest} \in \Indep^{b}(\gorest)$.
The inverse of the bijection
is clearly $(\Nod_1, \Nod_2) \mapsto \Nod_1 \cup \Nod_2$.
\end{claimproof}

By applying Claim \ref{claim:count-ai} and induction hypothesis, we get:
\begin{equation*}
\label{}
\begin{aligned}
\cou_{\ar}(\forest) &= \left({\cou'_{\ar_1}}(\forest_1) +  {\cou''_{\ar_1}}(\forest_1) \right)
\left({\cou'_{\br}}(\gorest) +  {\cou''_{\br}}(\gorest) \right)\\
& = \underbrace{{\cou'_{\ar_1}}(\forest_1) {\cou'_{\br}}(\gorest)}_{\defined \cou'_\ar(\forest)}\\
&+ \underbrace{{\cou'_{\ar_1}}(\forest_1)  {\cou''_{\br}}(\gorest)
+ {\cou''_{\ar_1}}(\forest_1)  {\cou'_{\br}}(\gorest) 
+ \cou''_{\ar_1}(\forest_1)  {\cou''_{\br}}(\gorest)}_{\defined \cou''_\ar(\forest)}.\\
\end{aligned}
\end{equation*}

We first note that the functions
$f_1 : \tree{\forest_1} \gorest \mapsto \cou'_{\ar_1}(\forest_1)$
and $f_2 : \tree{\forest_1} \gorest \mapsto \cou'_{\br}(\gorest)$
(defined only on inputs of the form $ \tree{\forest_1} \gorest \in \Facts{3|M|}{\mu}{w}$
for some $w \in A^*$) are polyblind. 
Indeed, since the height of the forest is bounded,
a bimachine can detect the $\righttree$ which matches
the first $\lefttree$, and simulate the computation of $\cou'_{\ar_1}$
(resp.  $\cou'_{\br}$) on $\forest_1$ (resp. $\gorest$).
Hence  $\cou'_{\ar} = f_1 \had f_2$ is polyblind.
Furthermore $f_1$ (resp. $f_2$) has growth at most $k_1$ (resp. $b$),
thus  $\cou'_{\ar}$ has growth at most $k_1 + b = k$.

Similarly, $\cou''_{\ar} \defined  \cou_{\ar} {-}  \cou'_{\ar}$
is a polyregular function (it also detects $\forest_1$ and $\gorest$). Furthermore each of its
terms has growth at most $k{-}1$ by induction hypothesis.

\subsection{Case $k \ge 1$ and $\ar =\tree{\Le}$} 

\label{ssec:Types}

In this case $\ar= \tree{\Le}$ where $\Le$ is a non-empty
multiset of elements of the form $(m\cro{u}m', d)$.
Note that $x = |\Le|$ and let $e$ be the idempotent such
that $e = m\mu(u)m'$
for all $(m\cro{u}m', d) \in \Le$ (by construction and  item~\ref{po:il2} of
Lemma~\ref{lem:indep-lin}
it is the same for all the elements occurring in $\Le$).

\begin{definition} Let $w \in A^*$, $\forest \in \Facts{\mu}{}{w}$ and
$\Nod \in \Nodes{\forest}$, we define
the multiset
$\Types{\forest}{\Nod} \defined \multi{(\lin{\forest}{\nod}, \depth{\forest}{\nod}) : \nod \in \Nod}
$.
It can have multiplicities even if $\Nod$ is a set.
\end{definition}

By definition of architectures and $\cou$, we have:
\begin{equation*}
\cou_{\tree{\Le}}(\forest) = 
\left\{
    \begin{array}{l}
        0 \text{~~if~~} \forest \text{~~is not of the form~~}
	\tree{\forest_1} \cdots \tree{\forest_n}\\
        \text{~~~with~~} n \ge 3 \text{~~and~~} \mu(\forest_1) = \mu(\forest_n) = e;\\
	 \left|\left\{\Nod \in \Indeps{\forest} : \Nod \subseteq  \Nodes{\tree{\forest_2} \cdots \tree{\forest_{n-1}}}\text{~~and~}
\Types{\forest}{\Nod} = \Le \right\}\right|\\
	\text{~~~otherwise.}\\
    \end{array}
\right.
\end{equation*}

From now, we assume that $\forest = \tree{\forest_1} \cdots \tree{\forest_n}$
where $\mu(\forest) = e$ (this is a regular property which can be checked).
Sublemma \ref{slem:decompose-direct} directly concludes
the proof if it is applied inductively.

\begin{sublemma} \label{slem:decompose-direct}
Let $\tau = (m \cro{u} n, d)$. Assume that $\Le =: \Le_1 \uplus \multi{\tau\con r}$
where $r > 0$,  $\tau \not \in \Le_1$ and
for all $(m' \cro{u'} n', d')  \in \Le_1$, one has 
$d \le d'$. Then one can build:
\begin{itemize}
\item a polyblind function $g' : (\apar{A})^* \rightarrow \Nat$ with growth at most $r$;
\item a polyregular function $g'':  (\apar{A})^* \rightarrow \Nat$ with growth at most $k{-}1$.
\end{itemize} 
such that $\cou_{\tree{\Le}}(\forest) = g'(\forest) \times \cou_{\tree{\Le_1}}(\forest) + g''(\forest)$.
\end{sublemma}

The rest of this subsection is devoted to the proof of Sublemma \ref{slem:decompose-direct}.
We suppose that the $\forest_1$ and $\forest_n$
have no iterable nodes, thus $\cou_{\tree{\Le}}(\forest) = |\{\Nod \in \Indeps{\forest} :
\Types{\forest}{\Nod} = \Le \}|$
(this assumption is just used to simplify the notations).
We first show how to decompose $ \cou_{\tree{\Le}}$ as
a sum indexed by the independent sets of type $\Le_1$. Let $k_1 \defined k{-}r \ge 0$
(this way $|\Le_1| = k_1$) and $A(\forest) \defined  \left\{\nod \in \itera{\forest}:
\Types{\forest}{\{\nod\}} = \tau\right\}$.

\begin{claim} \label{claim:decdec}
$\displaystyle \cou_{\tree{\Le}}(\forest)
= \sum_{\substack{\Nod_1 \in \Indep^{k_1}(\forest) \\ \Types{\forest}{\Nod_1} = \Le_1}}
\left|\left\{ \Nod_2 \subseteq A(\forest) : \Nod_1 {\cup} \Nod_2 \in \Indeps{\forest} \right\}\right|$
\end{claim}

\begin{claimproof} Since $\tau \not \in \Le_1$, the function
$\Nod \mapsto (\Nod \cap \{\nod : \Types{\forest}{\{\nod\}} \in \Le_1\},
\Nod \cap \{\nod : \Types{\forest}{\{\nod\}} = \tau\})$
is a bijection between the set of sets
$\left\{\Nod \in \Indeps{\forest} : \Types{\forest}{\Nod} = \Le \right\}$
and the set of couples of sets $\{(\Nod_1, \Nod_2) : \Types{\forest}{\Nod_1}
= \Le_1, \Nod_2 \subseteq A(\forest) \text{~and~}
\Nod_1 {\cup} \Nod_2 \in \Indeps{\forest} \}$
\end{claimproof}

We shall give two constructions for $g'$ and $g''$, depending on
whether $|A(\forest)| < 3k_1 + 2r$ or not. Since this condition
is a regular property of $\forest$, it can be checked before the computation
and does not matter.
To simplify the notations, we shall describe transducers which range over the
nodes of $\forest$. This can be implemented in practice
by ranging over the opening $\lefttree$ corresponding to the nodes.
Using this convention, the ordering $<$ on the positions
of $\forest$ (seen as a word) induces a total ordering $\pc$ on $\Nodes{\forest}$.
Furthermore, the transducer can access (using a lookahead)
the type of a node, check if two nodes marked by its pebbles
are independent or not, or if one is $\pc$ than another.
Finally, note that ranging over ordered tuples of nodes $\nod_1 \pc \dots \pc \nod_{k}$
exactly corresponds to ranging over sets of $k$ nodes.

\subparagraph*{First case: if $|A(\forest)| < 3k_1 + 2r$.}
We set $g'(\forest) \defined 0$ and
$g''(\forest) \defined \cou_{\tree{\Le}}(\forest) $. We show in Algorithm \ref{algo:cnt1}
how to implement $g''$ with
$k$ pebbles which range over the sets of independent nodes.
Furthermore, $g''$ has growth at most
$k_1 \le k{-}1$ by Claim~\ref{claim:decdec} since for a given set $\Nod_1$, there is only a bounded
number of sets $\Nod_2 \subseteq A(\forest)$ such that $\Nod_1 \cup \Nod_2 \in \Indeps{\forest}$.

\begin{algorithm}[h!]
\SetKw{KwVar}{Variables:}
\SetKwProg{Fn}{Function}{}{}
\SetKw{In}{in}
\SetKw{Let}{let}
\SetKw{Out}{Output}
 \Fn{$\cou_{\tree{\Le}}(\forest)$}{

		\For{$\nod_1 \pc \dots \pc \nod_{k} \in \itera{\forest}{}$}{
			\If{${\set{\nod_1, \dots, \nod_k} \in \Indep^{k}(\forest)}$
			\tnorm{and}		
			$\Types{\forest}{\set{\nod_1, \dots, \nod_k}} = \Le$}{
						
			\Out{$1$}

			}
			
		}}
	
 \caption{\label{algo:cnt1} First case: computing $\cou_{\tree{\Le}}$ with a pebble transducer}
\end{algorithm}

\subparagraph*{Second case: if $|A(\forest)| \ge 3k_1 + 2r$.}
Given $\Nod_1 \in \Indep^{k_1}(\forest)$
such that $\Types{\forest}{\Nod_1} = \Le_1$,
we define $A_{\Nod_1}(\forest) \defined
\{\nod \in A(\forest) : \{\nod\} \cup \Nod_1 \in \Indep^{k_1+1}(\forest) \} \subseteq A(\forest)$.

\begin{claim}
If $\Types{\forest}{\Nod_1} = \Le_1$, we have
$|A(\forest) \smallsetminus A_{\Nod_1}(\forest)| \le 3k_1$
\end{claim}

\begin{claimproof} The nodes of $A$ have a fixed
depth $d$, which is $\le$ than the depths of the nodes
of $\Nod_1$. Hence $A(\forest) \smallsetminus A_{\Nod_1}(\forest)$
contains the nodes of $A(\forest)$ which are either an ancestor
or the sibling of an ancestor of a node from $\Nod_1$,
and there are at most $3|\Nod_1| = 3k_1$ such nodes.
\end{claimproof}

Since $\pc$ is a total ordering on $\Nodes{\forest}$,
let $B_{\Nod_1}(\forest)$
denote the \mbox{$3k_1{-}|A(\forest) \smallsetminus A_{\Nod_1}(\forest)| \ge 0$} first
elements of $A_{\Nod_1}(\forest)$ (with respect to $\pc$) and
$C_{\Nod_1}(\forest) \defined A_{\Nod_1}(\forest) \smallsetminus B_{\Nod_1}(\forest)$.
It follows immediately that $|C_{\Nod_1}(\forest)| = |A(\forest)| - 3k_1 \ge 2r$.
We say that two nodes $\nod \pc \nod' \in C_{\Nod_1}(\forest)$ are \emph{close}
if there is no $\nod'' \in C_{\Nod_1}(\forest)$ such that $\nod \pc \nod'' \pc \nod'$.
Since $\Nod_1 \in \Indep^{k_1}(\forest)$, we have:
\begin{equation}
\begin{aligned}
\label{eq:dix}
&\left\{ \Nod_2 \subseteq A(\forest) :
\Nod_1 {\cup} \Nod_2 \in \Indeps{\forest} \right\}\\
&= \left\{ \Nod_2 \subseteq A_{\Nod_1}(\forest) :
\Nod_1 {\cup} \Nod_2 \in \Indeps{\forest} \right\}\\
&= \left\{ \Nod_2 \subseteq C_{\Nod_1}(\forest) :
\Nod_1 {\cup} \Nod_2 \in \Indeps{\forest} \text{ and no } t,t' \in \Nod_2 \text{ are close} \right\}\\
&\uplus \left\{ \Nod_2 \subseteq C_{\Nod_1}(\forest) :
\Nod_1 {\cup} \Nod_2 \in \Indeps{\forest} \text{ and } \exists t,t' \in \Nod_2 \text{ which are close} \right\}\\
&\uplus \left\{ \Nod_2 \subseteq A_{\Nod_1}(\forest) :
\Nod_1 {\cup} \Nod_2 \in \Indeps{\forest} \text{ and } \Nod_2 \cap B_{\Nod_1}(\forest) \neq \varnothing \right\}.\\
\end{aligned}
\end{equation}

Let us consider the function $P_r : \Nat \rightarrow \Nat$ which maps
$X \ge 0$ to the cardinal of the set $W$ of words $w \in \{0,1\}^{X}$ such that $|w|_1 = r$
and there are no two consecutive $1$ in $w$.

\begin{claim} For all $\Nod_1 \in \Indep^{k_1}(\forest)$ such that 
$\Types{\forest}{\Nod_1} = \Le_1$, we have
$ P_{r}(|A(\forest)| {-} 3k_1)$
$ =  |\{ \Nod_2 \subseteq C_{\Nod_1}(\forest) :
\Nod_1 {\cup} \Nod_2 \in \Indeps{\forest} \text{ and no } t,t' \in \Nod_2 \text{ are close} \}|$.
\end{claim}

\begin{claimproof} We first see that $ \{ \Nod_2 \subseteq C_{\Nod_1}(\forest) :
\Nod_1 {\cup} \Nod_2 \in \Indeps{\forest} \text{ and no } t,t' \in \Nod_2 \text{ are close} \}$
\linebreak
$= \{ \Nod_2 \subseteq C_{\Nod_1}(\forest) :
\text{no } t,t' \in \Nod_2 \text{ are close} \}$. Finally, we note that $|C_{\Nod_1}(\forest)| = |A(\forest)| - 3k_1$
and that $P_r$ exactly counts subsets of $r$ nodes
such that no two of them are close.
\end{claimproof}

Since this cardinal \emph{does not} depend on $\Nod_1$,
we define $g'(\forest) \defined P_{r}(|A(\forest)| {-} 3k_1)$. 

\begin{claim} \label{claim:P-poly}
The function $g'$ is polyblind 
and has growth at most $r$.
\end{claim}

\begin{proof} The function which
maps some $w \in W$ to itself in which each $10$ factor
(excepted the last one) is replaced by $1$, is a
bijection between $W$ and $\{ w \in \{0,1\}^{X-r+1} : |w|_r = 1\}$.
Hence $ P_r(X) = \binom{X{-}r{+}1}{r}
= \frac{(X{-}r{+}1)!}{r! (X{-}2r{+}1)!}$,
and thus
$g'(\forest) = \frac{1}{r!}{} \prod_{i=0}^{r-1} (|A(\forest)| {-} 3k_1{-}r{-}i{+}1)$.
It is clear that $r! {\times} g'$ is a polyblind function, since it
is the Hadamard product of $r$ regular functions. Then, dividing
by $r!$ consists in a post-composition by a regular function (with
both unary input and output alphabets),
which preserves polyblindness \cite{nous2020comparison}.
\end{proof}

If we denote by $c_{\Nod_1}(\forest)$ the cardinal
of the two last terms of Equations \ref{eq:dix}, we get:
\begin{equation*}
\begin{aligned}
\cou_{\tree{\Le}}(\forest)
= g'(\forest) \times \cou_{\tree{\Le_1} }(\forest)
+ \underbrace{\sum_{\substack{\Nod_1 \in \Indep^{k_1}(\forest) \\ \Types{\forest}{\Nod_1} = \Le_1}}
c_{\Nod_1}(\forest)}_{g''(\forest)}.
\end{aligned}
\end{equation*}

The function $g''$ is polyregular
(it can be computed by ranging over all possible
sets $\Nod_1$ as in Algorithm~\ref{algo:cnt1}
and then sets $\Nod_2$). Furthermore,
it has growth at most $k{-}1$ since
$c_{\Nod_1}(\forest)$ has growth at most $r{-}1$ (and
the bound is independent from $\Nod_1$).
Indeed, it has either two elements which are close in $A_{\Nod_1}(\forest)$,
or one element which is among the $3k_1$ first ones:
in both cases there is one less degree of freedom
(like in Appendix~\ref{ssec:bound}).

\newpage

\end{document}